\documentclass[a4paper,11pt,abstract,DIV=classic]{scrartcl}
\usepackage{amsmath,amssymb,amsthm}
\usepackage{booktabs}
\usepackage[tt=false,type1=true]{libertine}
\usepackage[libertine]{newtxmath}

\newenvironment{acks}{\section*{Acknowledgments}}{}
\usepackage{amsthm}
\usepackage{color}
\usepackage{csquotes}
\usepackage{mathtools}
\usepackage{xcolor}
\usepackage[textwidth=2cm,textsize=scriptsize,shadow]{todonotes}
\usepackage{xspace}
\usepackage{pgffor}
\usepackage{mathrsfs}
\usepackage{fontawesome}

\usepackage[inline]{enumitem}

\usepackage[normalem]{ulem}

\usepackage{multicol}
\usepackage{subcaption}

\usepackage{bm}

\definecolor{myblue}{HTML}{3D6288}
\definecolor{myred}{HTML}{9F1D2B}
\definecolor{mygreen}{HTML}{5B892D}

\definecolor{Orange}{HTML}{fc8d59}
\definecolor{Green}{HTML}{99d594}
\definecolor{Yellow}{HTML}{ffffbf}

\definecolor{ShadingLight}{HTML}{fef0d9}
\definecolor{ShadingMedium}{HTML}{fdcc8a}
\definecolor{ShadingDark}{HTML}{fc8d59}
\definecolor{ShadingVeryDark}{HTML}{d7301f}

\usepackage{tikz}
\usetikzlibrary{positioning}
\usetikzlibrary{decorations.pathreplacing}
\usetikzlibrary{shapes}
\usetikzlibrary{shapes.misc}
\usetikzlibrary{matrix}
\usetikzlibrary{calligraphy}
\usetikzlibrary{arrows}
\usetikzlibrary{calc}
\usetikzlibrary{cd}
\usetikzlibrary{backgrounds}

\theoremstyle{plain}
\newtheorem{theorem}{Theorem}[section]

\newtheorem{proposition}[theorem]{Proposition}
\newtheorem{observation}[theorem]{Observation}
\newtheorem{lemma}[theorem]{Lemma}
\newtheorem{claim}[theorem]{Claim}
\newtheorem{corollary}[theorem]{Corollary}

\theoremstyle{definition}

\newtheorem{remark}[theorem]{Remark}
\newtheorem{example}[theorem]{Example}
\newtheorem{definition}[theorem]{Definition}

\newenvironment{claimproof}[1][Proof]{%
        \begin{proof}[#1]%
    }{%
        \end{proof}%
    }

\let\epsilon\varepsilon
\let\phi\varphi

\let\from\colon
\newcommand*{\under}{\mathbin{\mid}}

\let\with\colon
\let\tup\bm
\newcommand*{\?}{\mathord{\:\cdot\:}}

\DeclareMathOperator*{\ar}{ar}

\DeclareMathOperator*{\adom}{adom}
\DeclareMathOperator*{\free}{free}
\DeclareMathOperator*{\facts}{facts}

\DeclareMathOperator{\Expectation}{E}

\DeclarePairedDelimiter{\size}{\lvert}{\rvert}
\DeclarePairedDelimiter{\set}{\lbrace}{\rbrace}

\DeclarePairedDelimiter{\ceil}{\lceil}{\rceil}

\newcommand*{\pdb}[1][D]{\csname#1\endcsname}

\newcommand*\edgeindepgraphs{\ensuremath{\mathscr G}}

\newcommand{\worlds}{\mathrm{worlds}}

\newcommand*{\cover}{\mathrm{FC}}
\newcommand*{\mincover}{\mathrm{FC}_{\mathrm{min}}}

\newcommand*{\minedgecover}{\mathrm{EC}_{\mathrm{min}}}

\newcommand*{\Fa}{F_{\mathrm{always}}}
\newcommand*{\Fs}{F_{\mathrm{sometimes}}}
\newcommand*{\Fn}{F_{\mathrm{never}}}

\newcommand*{\NEXT}{N}
\newcommand{\rest}[2]{#1[#2]}

\newcommand*{\classstyle}[1]{\bm{\mathsf{#1}}}

\newcommand*{\PDB}{\ensuremath{\classstyle{PDB}}}
\newcommand*{\TI}{\ensuremath{\classstyle{TI}}}
\newcommand*{\BID}{\ensuremath{\classstyle{BID}}}

\newcommand*{\PDBfin}{\PDB_{\mathsf{fin}}}
\newcommand*{\TIfin}{\TI_{\mathsf{fin}}}
\newcommand*{\BIDfin}{\BID_{\mathsf{fin}}}

\newcommand*{\FO}{\ensuremath{\classstyle{FO}}}
\newcommand*{\CQ}{\ensuremath{\classstyle{CQ}}}
\newcommand*{\sjfCQ}{\ensuremath{\classstyle{sjfCQ}}}
\newcommand*{\UCQ}{\ensuremath{\classstyle{UCQ}}}

\newcommand*{\FOTI}{\ensuremath{\FO(\TI)}}
\newcommand*{\FOTIFO}{\ensuremath{\FO(\TI\under\FO)}}
\newcommand*{\FOBID}{\ensuremath{\FO(\BID)}}

\DeclarePairedDelimiterX{\dbraces}[1]{\lbrace}{\rbrace}{%
	\nbrace{\lbrace}{#1}\delimsize\lbrace\mathopen{}%
	#1%
	\mathclose{}\delimsize\rbrace\nbrace{\rbrace}{#1}%
}
\newcommand{\dummydelim}[2]{$\left#1\vphantom{#2}\right.$}
\newcommand{\nbrace}[2]{\sbox0{\dummydelim{#1}{#2}}\hspace{\the\dimexpr -0.85\wd0 + 2pt\relax}}

\makeatletter
\def\bag{\@ifstar\@bag\@@bag}
\def\@bag#1{\dbraces{\smash{#1}}}
\def\@@bag#1{\dbraces*{#1}}
\newcommand{\Bags}[3][\@nil]{%
	\def\tmp{#1}%
	\ifx\tmp\@@nil%
		\dparens*{\smash{\begin{smallmatrix}#2\\#3\end{smallmatrix}}}%
	\else%
		\dparens*{\begin{smallmatrix}#2\\#3\end{smallmatrix}}%
	\fi%
}
\makeatother

\allowdisplaybreaks

\usepackage{authblk}
\title{Tuple-Independent Representations of Infinite Probabilistic Databases}

\author[1]{Nofar Carmeli}
\author[2]{Martin Grohe}
\author[2]{Peter Lindner}
\author[2]{Christoph Standke}

\affil[1]{snofca@cs.technion.ac.il}
\affil[2]{\{grohe,lindner,standke\}@informatik.cs.rwth-aachen.de}
\affil[1]{Technion---Israel Institute of Technology, Haifa, Israel}
\affil[2]{RWTH Aachen University, Aachen, Germany}

\usepackage{hyperref} %
\definecolor{myblue}{HTML}{3D6288}
\definecolor{myred}{HTML}{9F1D2B}
\definecolor{mygreen}{HTML}{5B892D}
\hypersetup{
    breaklinks,
    pdfencoding=auto,
    colorlinks,
    linkcolor=myred,%
    citecolor=mygreen,%
    urlcolor=myblue,%
}
\usepackage[capitalise,noabbrev]{cleveref}
\Crefname{assumption}{Assumption}{Assumptions}
\Crefname{conjecture}{Conjecture}{Conjectures}
\Crefname{observation}{Observation}{Observations}
\Crefname{claim}{Claim}{Claims}
\Crefname{fact}{Fact}{Facts}
\Crefname{notation}{Notation}{Notations}

\let\cal\mathcal
\let\bb\varmathbb
\let\fr\mathfrak

\makeatletter
\def\stylei#1{\cal{#1}}
\def\styleii#1{\bb{#1}}
\def\styleiii#1{\fr{#1}}
\foreach \x in {A,...,Z}{%
	\expandafter\xdef\csname \x\endcsname{\noexpand\stylei{\x}}%
	\expandafter\xdef\csname \x\x\endcsname{\noexpand\styleii{\x}}%
	\expandafter\xdef\csname \x\x\x\endcsname{\noexpand\styleiii{\x}}%
}
\makeatother

\begin{document}

\maketitle

\begin{abstract}
    Probabilistic databases (PDBs) are probability spaces over database instances. They provide a framework for handling uncertainty in databases, as occurs due to data integration, noisy data, data from unreliable sources or randomized processes. Most of the existing theory literature investigated finite, tuple-independent PDBs (TI-PDBs) where the occurrences of tuples are independent events. Only recently, Grohe and Lindner (PODS '19) introduced independence assumptions for PDBs beyond the finite domain assumption.
    In the finite, a major argument for discussing the theoretical properties of TI-PDBs is that they can be used to represent any finite PDB via views. This is no longer the case once the number of tuples is countably infinite. In this paper, we systematically study the representability of infinite PDBs in terms of TI-PDBs and the related block-independent disjoint PDBs.\par 
    
    The central question is which infinite PDBs are representable as first-order views over tuple-independent PDBs. We give a necessary condition for the representability of PDBs and provide a sufficient criterion for representability in terms of the probability distribution of a PDB. With various examples, we explore the limits of our criteria. In addition, we show that conditioning on first-order properties yields no additional power in terms of expressivity, and we discuss the relation between purely logical and arithmetic reasons for (non\mbox{-})representability.
    Finally, we inspect the relative expressivity of additional representations, where the views are restricted to be fragments of first-order logic (such as conjunctive queries with or without self-joins and unions of conjunctive queries).
\end{abstract}

\section{Introduction}

Probabilistic databases (PDBs)  provide a framework for dealing with uncertainty in databases, as could occur due to data integration, the acquisition of noisy data or data from unreliable sources or as outputs of randomized processes. Often the appropriate probability spaces are infinite. For example, fields in a fact may contain measurements from a noisy sensor, which we model as real numbers with some error distribution, or approximate counters, modeled by some probability distribution over the integers, or text data scraped from unreliable web sources, modeled by a probability distribution over the strings over a suitable finite alphabet.

Formally, PDBs are probability spaces over (relational) database instances. A key issue when working with PDBs in practice is the question of how to represent them. For finite PDBs, this is always possible in principle (ignoring numerical issues)%
by just listing all instances in the PDB and their probabilities. But of course, it is usually infeasible to store the whole sample space of a probability distribution over reasonably sized database instances. Instead, various compact representation systems have been proposed. Arguably the simplest is based on tuple-independent PDBs. In a \emph{tuple-independent PDB ($\TI$-PDB)} the truths of all facts $f$ are regarded as independent events that hold with a probability $p_f\in[0,1]$. The sample space of a $\TI$-PDB with $n$ facts has size $2^n$, but we can represent the $\TI$-PDB by just listing the $n$ marginal probabilities $p_f$ of the facts $f$. Much of the theoretical work on PDBs is concerned with $\TI$-PDBs. A justification for this focus on $\TI$-PDBs is the following nice Theorem~\cite{Suciu+2011}: \emph{every finite PDB can be represented by a first-order view over a (finite) $\TI$-PDB}. In other words: first-order views over $\TI$-PDBs form a complete representation system for finite PDBs.%

In this paper, we investigate similar questions in a broadened scope. Our main focus lies on probabilistic databases with a countably infinite sample space. The idea of modeling uncertainty in databases by viewing them as infinite collections of possible worlds has been around for quite some time in the context of incomplete databases \cite{Imielinski+1984}. While existing PDB systems are already using infinite domains and sample spaces, such as \cite{Deshpande+2004,Singh+2008b,Agrawal+2009,Kennedy+2010,Jampani+2011,Cai+2013}, a formal framework of probabilistic databases with infinite sample spaces has only been introduced recently \cite{Grohe+2019,GroheLindner2022}. Let us emphasize that in an infinite PDB, it is the sample space of the probability distribution that is infinite; every single instance of an infinite PDB is still finite. As a special case, it is natural to consider infinite PDBs of \emph{bounded instance size}, which means that all instances have size at most $b$ for some fixed bound $b$. A simple example that can be modeled as a PDB of bounded size is a table describing a fixed number of samples of the fermentation process when brewing Guinness beer, containing the date of each sample and the number of yeast cells in the sample. This number is known to be Poisson distributed \cite{Boland1984}.

For infinite PDBs, the issue of representing them becomes even more difficult, because clearly, not every infinite PDB has a finite or computable representation. In this paper, we set out to study representations of infinite PDBs by considering representations over infinite $\TI$-PDBs. Infinite $\TI$-PDBs have been considered in \cite{Grohe+2019}, and it has been observed there that it is \emph{not} the case that every countably infinite PDB can be represented by a first-order view over an infinite $\TI$-PDB. The reason for this is relatively simple: it can be shown that the expected instance size in a $\TI$-PDB is always finite, but there exist infinite PDBs with an infinite expected instance size. As first-order views preserve instance size up to a polynomial factor, such infinite PDBs with an infinite expected instance size cannot be represented by first-order views over $\TI$-PDBs.

Let us denote the class of all PDBs that can be represented by a first-order view over a $\TI$-PDB by $\FOTI$. We prove that the class $\FOTI$ is quite robust: somewhat unexpectedly, conditioning a $\TI$-PDB on a first-order constraint prior to applying a first-order view does not yield additional expressive power, i.\,e.\ $\FOTIFO = \FOTI$. This allows us to show that all block-independent disjoint PDBs ($\BID$-PDBs) are in $\FOTI$. A $\BID$-PDB is a PDB where the set of facts is partitioned into blocks, such that facts from different blocks are independent, while facts from the same block are disjoint. $\BID$-PDBs form a practically quite relevant class of PDBs that includes PDBs of the form described in the yeast cell example above, where specific fields of facts in a table store the outcome of a random variable.

In \cite{Grohe+2019}, the authors exploited that $\TI$-PDBs have finite expected instance size to construct a PDB that is not contained in \FOTI{}. We generalize this idea and prove that for every PDB in $\FOTI$, all \emph{size moments}, that is, moments of the random variable that maps each instance to its size, are finite. This imposes a fairly strong restriction on the probability distribution of PDBs in $\FOTI$. In addition, we give an example showing that there are even PDBs that have finite size moments but that are still not in $\FOTI$. Complementing these non-representability results, we prove that all PDBs of bounded instance size are in $\FOTI$. Furthermore, we give a sufficient criterion on the growth rate of the probabilities of PDBs for membership in $\FOTI$.

All the non-representability results mentioned so far are caused by unwieldy probability distributions. We say that the reasons for these non-representability results are \emph{arithmetical}. We asked ourselves if it can also happen that a PDB is not in $\FOTI$ for \emph{logical} reasons, for example, because there are large gaps in the range of instance sizes. We formalize this question by saying that a PDB is not in $\FOTI$ \emph{for logical reasons} if there is no probability distribution that assigns positive probabilities to all instances in the sample space such that the resulting PDB is in $\FOTI$. Thus, being not representable for logical reasons is a property of the sample space and not of the probability measure. Arguably, this notion is closer to the theory of incomplete databases than to probabilistic databases. Surprisingly, we prove that there are no logical reasons for non-representability.

As we have seen that various classes of first-order views over $\TI$- and $\BID$-PDBs coincide with $\FOTI$, one might wonder about the consequences of allowing only fragments of first-order logic in representation. We investigate such classes and show that, while most of these classes collapse in the finite setting, they remain separated in the infinite.

\paragraph*{Related Work}
By now, there exists an abundance of theoretical work on finite PDBs (see the surveys \cite{Suciu+2011,VanDenBroeck+2017} as well as \cite{Dalvi+2009,Suciu2020}). The most prominent theoretical problem regarding finite representation systems is probabilistic query evaluation (PQE) complexity \cite{Dalvi+2011,Dalvi+2012,Fink+2016,Amarilli+2020}, typically subject to independence assumptions. In particular, \cite{Amarilli+2015} considers PQE on structurally restricted $\BID$-PDBs. Among query languages, conjunctive queries and unions of conjunctive queries received the most attention \cite{Dalvi+2012}. On the side of expressiveness, \cite{Benjelloun+2006,Green+2006,Sen+2009,Antova+2009} consider more sophisticated PDB representations built upon independence assumptions.
Probabilistic models for sensor networks typically feature a finite number of facts with continuously distributed attributes \cite{Faradjian+2002,Deshpande+2004,Singh+2008a}. Overviews over representation formalisms for probabilistic databases can be found in \cite{Green+2006,Zhang+2008,Sarkar+2009,Suciu+2011,VanDenBroeck+2017}. 
\par 

Conditioning PDBs has been considered before in order to introduce correlations or dependencies \cite{Re2007,Jha+2012}, for updates of probabilistic data \cite{Koch+2008}, for making queries tractable \cite{Olteanu+2009}, in cleaning problems \cite{Gribkoff+2014a,DeSa+2019} and for introducing ontologies \cite{Jung+2012,Borgwardt+2017}. In \cite{Dalvi+2007}, the authors consider limit probabilities of conjunctive queries over conditioned PDBs.

There already exist PDB systems supporting infinite probability spaces \cite{Deshpande+2004,Singh+2008b,Agrawal+2009,Kennedy+2010,Jampani+2011,Cai+2013}. The investigations in these works are generally directly tied to their representation mechanisms, which complicates an abstract comparison. However, all of the mentioned approaches directly transition to a uncountable setting with continuous distributions, which is different from our setting of countably infinite PDBs. The formal possible worlds semantics have been extended towards infinite probability spaces, allowing for instances of unbounded size in \cite{Grohe+2019,GroheLindner2022}. Allowing for distributions over worlds of unbounded size is strongly motivated by incorporating the open-world assumption into PDBs \cite{Ceylan+2016,Grohe+2019,Borgwardt+2017,Friedman+2019}. Semi-structured models (probabilistic XML \cite{Kimelfeld+2013}) have been extended towards infinite spaces as well \cite{Abiteboul+2011,Benedikt+2010}.

As discussed by \cite{Green+2006}, incomplete databases \cite{Imielinski+1984,Abiteboul+1995,VanDerMeyden1998} are closely related to PDBs. In \cite{DasSarma+2006}, the authors study expressiveness among classes of incomplete databases with strong associations with $\TI$- and $\BID$-PDBs. Recent work on incomplete databases considers (probabilistic) measures of certainty for infinite domains \cite{Libkin2018,Console+2020}.

It has been noticed that there are strong connections between PDBs and probabilistic models in AI research \cite{VanDenBroeck+2017}, in particular probabilistic graphical models \cite{Koller+2009} and weighted model counting (WMC) \cite{Gribkoff+2014b}. There exist well-established modeling formalisms that support infinite spaces, for example \cite{Milch+2005,Singla+2007,Gutmann+2011}. Recently, WMC has also been introduced for infinite universes \cite{Belle2017}.

\paragraph*{Paper Outline}
We review the background of probabilistic databases in \cref{sec:preliminaries}. In \cref{sec:negative}, we explore the limits of tuple-independent representations by first-order views. We extend the aforementioned non-representability result that relies on the expected instance size to apply to all size moments, and we provide a new necessary condition for representability that applies also in the case of finite moments. In \cref{sec:conditions}, we show that $\FOTI$ is closed under conditioning under first-order constraints.     \Cref{sec:power} contains positive results. In \cref{ssec:sizesandprobs} we establish a sufficient criterion on the growth rate of the probabilities and conclude also that PDBs of bounded instances size are in $\FOTI$. In \cref{ssec:bid}, we show the same for $\BID$-PDBs. In \cref{sec:incomplete}, we consider logical reasons. In \Cref{ssec:nonrep-logical}, we demonstrate how the logic alone can be used to show non-representability of PDBs in the infinite. In \Cref{ssec:fo-logical}, we show that, when the instance size is unbounded, any argument regarding representability using $\TI$-PDBs with FO-views must take the instance probabilities into account. In \cref{sec:fragments}, we consider various important fragments of first-order logic and investigate their relative expressive power for representations over PDBs with independence assumptions. We close the paper with concluding remarks in \cref{sec:conclusion}.

\smallskip

This paper is the journal version of the paper \cite{Carmeli+2021} that was presented at PODS 2021. It contains detailed proofs and examples that were excluded from the conference version (\cref{lem:momentsinequality,lemma:view-prob,lem:condition,exa:ti-not-cond,exa:bid-not-cond,lem:no-logical-reason}). Some of the arguments have been streamlined and simplified. In addition to the main body of results from the conference version, this paper contains a new section (\cref{sec:fragments}) that focuses on representations using fragments of $\FO$ and their relative expressiveness. This answers open questions from the original paper. %
\section{Preliminaries}\label{sec:preliminaries}

In this section, we provide definitions and state known results that we use throughout this paper.

We denote the set of non-negative integers by $\NN$, whereas the set of positive integers is denoted $\NN_+$. We write $(0,1)$, $[0,1)$, $(0,1]$ and $[0,1]$ for the open, half-open and closed intervals of real numbers between 0 and 1.

\subsection{Probability Spaces}

A \emph{discrete probability space} is a pair $\pdb[S] = (\SS,P)$ where $\SS$ is a non-empty, countable set, called the \emph{sample space} and $P$ is a probability measure (or \emph{probability distribution}) on $\SS$. That is, $P \from 2^{\SS} \to [0,1]$ with the property that $ P( A ) = \sum_{s \in A} P\big( \set{ s }\big) $ for all $A \subseteq \SS$, and $P(\SS) = 1$. We denote probability spaces with curly letters and their sample spaces with double-struck letters. We denote probability distributions with variants of $P$.
Throughout this paper, all probability spaces are assumed to be discrete. 

We write $S \sim \pdb[S]$ to indicate that $S$ is a random element, drawn from $\SS$ according to distribution $P$. If $P$ is anonymous, or when we want to emphasize this perspective of drawing a random element, we write 
\[
    \Pr_{S \sim \pdb[S]}\big( S \text{ has property } \phi \big) \coloneqq P\big(\set{S \in \SS \with S \text{ has property }\phi} \big)\text.
\]
Subsets of $\SS$ are called \emph{events}. A collection $(A_i)_{i\in I}$ of events in $\pdb[S]$ is called \emph{(mutually) independent} if $P(\bigcap_{i\in J} A_i) = \prod_{i \in J} P(A_i)$ for all finite subsets $J$ of $I$. It is called \emph{(mutually) exclusive} if $P(A_i \cap A_j) = 0$ for all $i \neq j$.

A \emph{random variable} $X$ on $\pdb[S] = (\SS,P)$ is a function $X \from \SS \to \RR$. Its \emph{expectation} is $\Expectation(X) \coloneqq \sum_{s \in \SS} X(s) \cdot P(\set{s})$ and, in general, its \emph{$k$th moment} is $\Expectation(X^k)$. We write $\Expectation_{\pdb[S]}$ for the expectation in $\pdb[S]$ if the probability space is not clear from the context.

If $\pdb[S] = (\SS,P_{\pdb[S]})$ is a probability space, $\TT$ a (countable) set and $f \from \SS \to \TT$ a function, then $f$ introduces a probability distribution on $\TT$ via $P_{\pdb[T]}(\set{t}) \coloneqq P_{\pdb[S]}(\set{s \in \SS\with f(s) = t})$.

\subsection{Relational Databases}

We fix some non-empty set $\UU$ (called the \emph{universe}). A \emph{database schema} $\tau$ is a finite, nonempty set of relation symbols with \emph{arities} $\ar(R) \in \NN$ for all $R\in\tau$. A ($\tau$-)fact is an expression of the shape $R(u_1,\dots,u_{\ar(R)})$ where $R\in\tau$ and $u_i \in \UU$ for all $1\leq i\leq \ar(R)$. A \emph{$\tau$-instance} $D$ is a finite set of $\tau$-facts. The \emph{active domain} of $D$, denoted by $\adom(D)$, is the set of elements of $\UU$ appearing among the facts of $D$. Throughout this paper, we assume that the universe $\UU$ is countably infinite.

In general, a \emph{view} may be any function that maps database instances of an input schema $\tau$ to instances of an output schema $\tau'$. A \emph{query} is a view whose output schema consists of a single relation symbol. Then, a view may also be thought of as a finite collection of queries, one per each relation in the output schema. In this paper, we focus on views and queries that are expressed in first-order logic.

\subsection{First Order Logic, Relational Calculus, and Fragments}

Let $\tau$ be a database schema. A \emph{relational atom} is an expression $R( \tup u )$ where $\tup u$ is a tuple over variables or constants from $\UU$. An \emph{equality atom} is an expression $x = u$ where $x$ is a variable and $u$ is a variable or a constant. A formula of \emph{first-order logic} $\FO$ over $\tau$ is built from atoms using the standard Boolean connectives and existential and universal quantification. When using first-order logic as a database query language, we assume the \emph{active domain semantics} in this paper, see \cite{Abiteboul+1995}.\footnote{We could equivalently restrict ourselves to \emph{domain-independent} formulae and use the standard semantics of first-order logic \cite[pp. 77ff]{Abiteboul+1995}. As we will construct a plethora of formulae in this paper, using the active domain semantics is more convenient though.} That is, when a formula $\Phi$ is evaluated on a $\tau$-instance $D$, the quantifiers, and all valuations are taken to range over $\adom( D, \Phi ) \coloneqq \adom( D ) \cup \adom( \Phi )$, where $\adom(\Phi)$ is the set of constants appearing in $\Phi$.

We denote the free variables of an $\FO$-formula $\Phi$ by $\free( \Phi )$ and write $\Phi = \Phi( \tup x )$ where $\tup x$ consists of pairwise distinct variables to indicate that $\Phi$ has exactly the free variables $\tup x$. A formula $\Phi$ is called a \emph{sentence} or \emph{Boolean} if $\free(\Phi) = \emptyset$. Essentially, we employ first-order logic in the guise of \emph{relational calculus}. That is, if $D$ is a $\tau$-instance, and $\Phi( \tup x )$ is an $\FO$-formula over $\tau$, the result of $\Phi = \Phi( x_1,\dots,x_n) $ on $D$ is given as
\begin{equation}\label{eq:rc}
    \Phi( D ) \coloneqq \set[\big]{ R_{\Phi}(\tup a) \with \tup a = (a_1,\dots,a_n) \in \adom(D,\Phi)^n \text{ such that } D \models \Phi[ \tup x / \tup a ] }\text,
\end{equation}
where $\models$ refers to the active domain semantics, and $\Phi[ \tup x / \tup a ]$ is the sentence that emerges from $\Phi$ by replacing $x_i$ with the constant $a_i$ for all $i =1,\dots,n$. Moreover, $R_{\Phi}$ is a designated relation symbol of the output database schema. Thus, we treat $\Phi(D)$ as a database instance. Of course, \eqref{eq:rc} depends on the order imposed on the free variables. We make this order explicit by writing $\Phi = \Phi( \tup x )$. 

In general, a tuple $\tup x$ may contain repeated variables or constants, like in $\Phi(x,x,y,a,y)$. This is equivalent to the first-order formula $\Phi'(x,x',y,x'',y'')$ that emerges from $\Phi$ by conjoining the formula with $x'=x$, $x''=a$ and $y'$ equals $y$, where $x',x'',y'$ are new variables.

A first-order formula is a \emph{conjunctive query} ($\CQ$) if it is built from atoms by only using $\exists$ and $\wedge$. Every $\CQ$ is equivalent to a formula in the shape $\exists \tup y \with \bigwedge_{ i = 1 }^m \Phi_i$ where the $\Phi_i$ are atomic formulae (hence, the name $\CQ$). A $\CQ$ is called \emph{self-join free} ($\sjfCQ$), if every relation symbol occurs at most once in the formula. Note that we can rewrite every satisfiable conjunctive query $\Phi(\tup x)$ into a conjunctive query $\Phi'(\tup x')$ without equality atoms using substitutions. This might then introduce constants or repeated variables in $\tup x'$ (see \cite[Chapter 4]{Abiteboul+1995}).

A first-order formula is a \emph{union of conjunctive queries} ($\UCQ$) if it is built from atoms by only using $\exists$ and $\wedge$ and $\vee$. Every $\UCQ$ is equivalent to a formula in the shape $\bigvee_{ i = 1 }^n \Phi_i$ where the $\Phi_i$ are $\CQ$s.

For more background on first-order logic, relational calculus and the fragments introduced above, see \cite{Abiteboul+1995,Libkin2004}.

\subsection{Probabilistic Databases}

A probabilistic database is a probability space over a set of database instances.

\begin{definition}\label{def:pdb}
    A \emph{probabilistic database (PDB)} of database schema $\tau$ over $\UU$ is a discrete probability space $\pdb = ( \DD, P )$ where $\DD$ is a set of $\tau$-instances. We denote the class of PDBs by $\PDB$. If $\classstyle{D}$ is a class of PDBs and $\pdb \in \classstyle{D}$, we call $\pdb$ a \emph{$\classstyle{D}$-PDB}.
\end{definition}

Note that while $\size{ \DD }$ may be (countably) infinite, $\DD$ still only contains database instances, that is, \emph{finite} sets of facts. We call a PDB \emph{finite} if $\size{ \DD }$ is finite. The class of finite PDBs is denoted by $\PDBfin$.

\begin{remark}
    As seen in \cref{def:pdb}, all probabilistic databases that are considered in this paper have a sample space of at most countable size. In general, a PDB may also be an uncountable probability space over some uncountable domain~\cite{Grohe+2019,GroheLindner2022}. 
\end{remark}

In a (discrete) probabilistic database $\pdb$, the instances of positive probability are often called its \emph{possible worlds}. We denote the set of possible worlds of the PDB $\pdb$ by $\worlds( \pdb )$.\footnote{In most cases, we can just assume that $\DD$ contains no instances of probability zero. However, the notation we introduce here proves convenient later in the paper.} We let $\facts( \pdb ) \coloneqq \bigcup_{ D \in \worlds(\pdb) } D$ denote the set of all facts appearing among the instances of $\pdb$. For all $f \in \facts( \pdb )$, the \emph{marginal probability} of $f$ is $\Pr_{ D \sim \pdb }\big( f \in D \big)$, and is usually denoted by $p_f$.

\subsubsection{Instance Size}
In every probabilistic database $\pdb = (\DD, P)$, the \emph{instance size} $\size{\?}$ is a random variable $\size{\?} \from \DD \to \NN$ that maps every database instance $D$ to the number $\size{D}$ of facts it contains. Its expectation is given by 
\[
    \Expectation(\size{\?}) = \sum_{D \in \DD} \size{D} \cdot P\big(\set{D})
    \text.
\]

\begin{definition}[Finite Moments Property]
    A PDB $\pdb = ( \DD, P )$ has the \emph{finite moments property} if for all $k \in \NN_+$, the $k$th moment of the instance size is finite, that is, 
    \[
        \Expectation\big( \size{ \? }^k \big) = \sum_{ D \in \DD } \size{D}^k \cdot P\big( \set{D} \big) < \infty\text.
    \]
    A class $\classstyle{D}$ of PDBs has the \emph{finite moments property}, if every $\pdb \in \classstyle{D}$ does.
\end{definition}

\subsubsection{Probabilistic Query Semantics}

Applying a view on a PDB yields a new output PDB: If $\pdb = (\DD,P)$ is a PDB, $\DD'$ a set of database instances and $V \from \DD \to \DD'$ a view, then $V(\pdb) = (\DD',P')$ is a PDB where 
\[
    P'(\set{D'}) = P(\set{D \in \DD \with V(D) = D'})
\] 
for all $D' \in \DD'$.

If $\classstyle V$ is a class of views and $\classstyle D$ is a class of PDBs, then $\classstyle{V}(\classstyle{D})$ denotes the class of images of PDBs of $\classstyle D$ under views of $\classstyle V$. We call $\classstyle{D}$ \emph{closed under} $\classstyle V$ if $\classstyle{V}(\classstyle{D}) = \classstyle{D}$.

An $\FO$-view is a view that consists of an $\FO$-formula for each relation symbol in the target schema. Then, $\FO(\classstyle{D})$ denotes the class of PDBs that are the image of a $\classstyle{D}$-PDB under an $\FO$-view. These notions are defined analogously for $\sjfCQ$, $\CQ$ and $\UCQ$.

\subsubsection{Independence Assumptions}
Even when leaving out probabilities, the number of possible worlds existing over a fixed set of $n$ facts is exponential in $n$. In the finite setting, this motivates the introduction of simplifying structural assumptions that allow for succinct representations of PDBs. While for infinite PDBs such a representation might still be infinite, its description still enjoys the simple structure and thus may have advantages with respect to approximate query answering.

For ease of reading, we denote PDBs that obey independence assumptions by $\pdb[I] = ( \II, P )$ rather than $\pdb = ( \DD, P )$. Instances are then denoted by $I$ accordingly.

\begin{definition}[Tuple-Independent PDBs]\label{def:ti-pdb}
    A PDB $\pdb[I]$ is called \emph{tuple-independent} if for all $k \in \NN$ and all pairwise distinct facts $f_1,\dots,f_k\in \facts(\pdb[I])$ it holds that
    \begin{equation*}
        \Pr_{ I \sim \pdb[I] } \big( f_1 \in I, \dots, f_k \in I \big) 
            =
        \prod_{ i = 1 }^{ k } \Pr_{ I \sim \pdb[I] } \big( f_i \in I \big)\text.
    \end{equation*}
    We let $\TI$ and $\TIfin$ denote the class of tuple-independent PDBs and finite tuple-independent PDBs, respectively.\footnotemark{}
    \footnotetext{%
        Usually, \emph{tuple-independent} PDBs are defined by the property that the events \enquote{$f \in I$} are stochastically independent. This definition only makes sense for countable PDBs, as in every PDB there are at most countably many facts with positive marginal probability (see \cite[Proposition 3.4]{Grohe+2019}). 
        There is an extension of tuple-independent PDBs to uncountable domains \cite{Grohe+new}.
    }
\end{definition}

The following provides a necessary and sufficient criterion for the existence of $\TI$-PDBs in terms of its marginal probabilities.

\begin{theorem}[{\cite[Theorem 4.8]{Grohe+2019}}]\label{thm:well-def-ti}
    Let $F$ be a set of facts over a schema $\tau$ and let $( p_f )_{ f \in F }$ with $p_f \in [0,1]$ for all $f \in F$. The following are equivalent:
    \begin{enumerate}
        \item There exists $\pdb[I] \in \TI$ with $\facts( \pdb[I] ) = F$ and marginal probabilities $\Pr_{ I \sim \pdb[I] }\big( f \in I \big) = p_f$ for all $f \in F$.
        \item It holds that $\sum_{ f \in F } p_f < \infty$.
    \end{enumerate}
\end{theorem}

$\TI$-PDBs are a special case of the following, more general model.

\begin{definition}[Block-Independent Disjoint PDBs]
    A PDB $\pdb[I]$ is called \emph{block-independent disjoint} if there exists a partition $\B$ of $\facts( \pdb[I] )$ into \emph{blocks} such that:
    \begin{enumerate}
        \item for all $k \in \NN$ and all $f_1,\dots,f_k$ from pairwise different blocks,
        \begin{equation*}
            \Pr_{ I \sim \pdb[I] } \big( f_1 \in I, \dots, f_k \in I \big) 
                =
            \prod_{ i=1 }^k \Pr_{ I \sim \pdb[I] } \big( f_i \in I \big)\text.
        \end{equation*}
        \item for all $B \in \B$ and all $f,f' \in B$ with $f \neq f'$ it holds that 
        \begin{math}
            \Pr_{ I \sim \pdb[I] } \big( f \in I \text{ and } f' \in I \big) = 0
        \end{math}.
    \end{enumerate}
    We let $\BID$ and $\BIDfin$ denote the class of block-independent disjoint and finite block-independent disjoint PDBs, respectively.
\end{definition}

A theorem similar to \cref{thm:well-def-ti} exists for $\BID$-PDBs:

\begin{theorem}[{\cite[Theorem 4.15]{Grohe+2019}}]\label{thm:well-def-bid}
    Let $F$ be a set of facts over a schema $\tau$, let $( p_f )_{ f \in F}$ with $p_f \in [0,1]$, and let $\B$ be a partition of $F$ such that $\sum_{f \in B} p_f \leq 1$ for all $B \in \B$. The following are equivalent:
    \begin{enumerate}
        \item There exists $\pdb[I] \in \BID$ with $\facts( \pdb[I] ) = F$, blocks $\B$ and marginal probabilities $\Pr_{ I \sim \pdb[I] }\big( f \in I \big) = p_f$ for all $f \in F$.
        \item It holds that $\sum_{f \in F} p_f = \sum_{B \in \B} \sum_{f \in B} p_f < \infty$.
    \end{enumerate}
\end{theorem}

We note that both $\TI$-, and $\BID$-PDBs are, in case of existence, uniquely determined by their marginals (and partition into blocks) \cite[Section 4]{Grohe+new}.

\section{Limitations of Tuple-Independent Representations}\label{sec:negative}

Recall that every finite probabilistic database can be represented as an $\FO$-view over a finite $\TI$-PDB \cite[Proposition 2.16]{Suciu+2011}. A similar statement does \emph{not} hold for infinite PDBs.

\begin{proposition}[{\cite[Proposition 4.9]{Grohe+2019}}]\label{pro:foti-not-everything}
     $\FOTI \subsetneq \PDB$.
\end{proposition}

In this section, we investigate classes of PDBs for which we can prove that they do not belong to $\FOTI$. The argument of \cite{Grohe+2019} that establishes the proposition above, shows that in any $\FOTI$-PDB the expected instance size is necessarily finite (while there exist well-defined PDBs of infinite expected instance size). In \cref{sec:negative-infinite}, we will first generalize their argument to show that for membership in $\FOTI$ it is even necessary that \emph{all} moments of the instance size are finite, that is, that the PDB satisfies the finite moments property. In \cref{sec:negative-finite}, however, we show that this requirement is not \emph{sufficient} for membership in $\FOTI$ by giving constructions of non-$\FOTI$-PDBs that nevertheless exhibit the finite moments property.

\subsection{Infinite Moments}\label{sec:negative-infinite}

Before considering the class $\FOTI$ directly, we show that all $\TI$-PDBs have the finite moments property. This relies upon the following observation about independent $\set{0,1}$-valued random variables.

\begin{lemma}\label{lem:momentsinequality}
    Suppose $\big( X_i \big)_{ i \in \NN }$ is a family of\/ $\set{0,1}$-valued, independent random variables and let $X \coloneqq \sum_{ i = 1 }^{ \infty } X_i$ with\/ $\Expectation( X ) < \infty$. Then for all $k \in \NN_+$, $k \geq 2$ it holds that
    \begin{equation}\label{eq:moments-ineq}
        \Expectation\big( X^k \big) \leq \Expectation\big( X^{k-1} \big) \cdot \big( k-1 + \Expectation( X ) \big)\text.
    \end{equation}
\end{lemma}

\begin{proof}
  Let $k \in \NN_+$ be arbitrary with $k \geq 2$. It holds that\footnote{The calculation is not as trivial as it may seem, since we sum up infinitely many random variables. However, the equalities can easily be verified using \cite[Theorem 5.3]{Klenke2014}.}
  \[
    \Expectation\big( X^k \big) 
        =
    \Expectation\bigg( \Big( \sum_{ i = 1 }^{ \infty } X_i \Big)^k \bigg)
        = 
    \Expectation\bigg(
        \sum_{i_1,\dots,i_k}
        X_{i_1} \cdot \dotsc \cdot X_{i_k}
        \bigg)
        =
    \sum_{ i_1, \dots, i_k}
    \Expectation\big(
        X_{i_1} \cdot \dotsc \cdot X_{i_k}
    \big)
  \]
  (with the indices of sums on the right ranging over non-negative integers). Now consider the product $X_{i_1} \cdot \dotsc \cdot X_{i_k}$ for arbitrary $i_1,\dots,i_k \in \NN$. Using that the $X_i$ are $\set{0,1}$-valued and independent, it holds that
  \[
    \Expectation\big( X_{i_1} \cdot \dotsc \cdot X_{i_k} \big)
        =
    \begin{cases}
      \Expectation\big( X_{i_1} \cdot \dotsc \cdot X_{i_{k-1}} \big)                                        & \text{if }i_k \in \set{i_1,\dots,i_{k-1}}\text{ and}\\
      \Expectation\big( X_{i_1} \cdot \dotsc \cdot X_{i_{k-1}} \big) \cdot \Expectation\big( X_{i_k} \big)  & \text{otherwise.}
    \end{cases}
  \]
  Using this to continue our calculation from before, we get
  \begin{align*}
      \Expectation\big( X^k \big)
      &=
      \sum_{ i_1,\dots,i_k } \Expectation\big( X_{i_1} \cdot \dotsc \cdot X_{i_k} \big) \\
      &=
      \sum_{ i_1,\dots,i_{k-1} } \bigg( \sum_{i_k} \Expectation\big( X_{i_1} \cdot \dotsc \cdot X_{i_k} \big) \bigg)\\
      &=
      \sum_{ i_1,\dots,i_{k-1} } \bigg(
        \sum_{ i_k \in \set{ i_1,\dots,i_{k-1} } }\Expectation\big( X_{i_1} \cdot \dotsc \cdot X_{i_{k-1}} \big) +
        \sum_{ i_k \notin \set{i_1,\dots,i_{k-1} } }\Expectation\big( X_{i_1} \cdot \dotsc \cdot X_{i_{k-1}} \big) \cdot \Expectation\big( X_{ i_k } \big)
      \bigg)\\
      &=
      \sum_{ i_1,\dots,i_{k-1} } 
        \Expectation\big( X_{i_1} \cdot \dotsc \cdot X_{i_{k-1}} \big)
        \cdot 
        \bigg(
            \sum_{ i_k \in \set{i_1,\dots,i_{k-1} } } 1 +
            \sum_{ i_k \notin \set{i_1,\dots,i_{k-1} } } \Expectation\big( X_{i_k} \big)
        \bigg)\\
      &\leq
      \bigg(
        \sum_{ i_1,\dots,i_{k-1} } 
        \Expectation\big( X_{i_1} \cdot \dotsc \cdot X_{i_{k-1}} \big)
      \bigg)
        \cdot 
      \bigg(
        k-1 + 
        \sum_{ i_k } \Expectation\big( X_{i_k} \big)
      \bigg)\\
      &= \Expectation\big( X^{k-1} \big) \cdot \big( k-1 + \Expectation( X ) \big)\text.\qedhere
  \end{align*}
\end{proof}

In our application, the random variables $X_i$ from the previous lemma will correspond to the indicator random variables of the facts of a $\TI$-PDB.

\begin{proposition}\label{pro:finite-moments-ti}
    $\TI$ has the finite moments property.
\end{proposition}

\begin{proof}
    Suppose $\facts(\pdb[I]) = \set{ f_0, f_1, \dots }$ and let $X_i \colon \II \to \set{ 0, 1 }$ be the indicator random variable of the event $\set{ f_i \in I }$ for $I \sim \pdb[I]$. Since the size of any instance $I \in \II$ is exactly the number of facts in $I$, we can express the instance size random variable $\size{ \? }$ using the indicator random variables $X_i$ as follows:
    \[
        \size{ \? } = \sum_{i = 1}^\infty X_i
        \text.
    \]
    For simplicity, we write $\Expectation$ for $\Expectation_{\mathcal{I}}$. Since $X_i$ is $\set{0,1}$-valued, we know that 
    \[
        \Expectation\big( X_i \big) 
            = 
        \Pr_{ I \sim \mathcal{I} } \big( X_i(I) = 1 \big)
            = 
        \Pr_{ I \sim \mathcal{I} } \big( f_i \in I \big) 
            = 
        p_i
    \] 
    for all $i \in \NN$, where $p_i$ is the marginal probability of $f_i$ in $\pdb[I]$. Thus, 
    \[
        \Expectation\big(\size{\?}\big) 
            = 
        \Expectation\Big(\sum_{i = 1}^\infty X_i\Big)
            = 
        \sum_{i = 1}^\infty\Expectation\big( X_i \big) 
            =   
        \sum_{i = 1}^\infty p_i
        \text.
    \] 
    By \cref{thm:well-def-ti} the last sum is finite, so $\Expectation\big( \size{\?} \big) < \infty$. In order to obtain this for higher moments $\size{ \? }^k$ with $k > 1$, we proceed inductively (going from $k-1$ to $k$) using \cref{lem:momentsinequality}. This yields
    \[
        \Expectation\big( \size{ \? }^k \big)
            \leq 
        \Expectation\big( \size{ \? }^{k-1} \big) \cdot \Big( k-1 + \Expectation\big( \size{ \? } \big)\Big)\text.
    \]
    The right-hand side of the above is finite by the induction hypothesis, i.\,e. $\Expectation\big( \size{ \? }^ k \big) < \infty$.
\end{proof}

\begin{remark}
    The same argument can be employed to show that $\BID$ has the finite moments property. For this, $X_i$ needs to be taken as the indicator random variable of the $i$th block of the $\BID$-PDB. Then one can proceed analogously to the proof of \cref{pro:finite-moments-ti}.
\end{remark}

We next show that $\FO$-views over classes with the finite moments property preserve the finite moments property. This is because $\FO$-views can only cause a polynomial blowup in the instance size per individual possible world.

\begin{lemma}\label{lem:fo-preserves-finite-moments}
    If\/ $\pdb$ is a PDB with the finite moments property and $\Phi$ is an $\FO$-view, then $\Phi(\pdb)$ also has the finite moments property.
\end{lemma}

\begin{proof}
    Let $\pdb$ be a PDB with the finite moments property and let $\Phi$ be an $\FO$-view over the schema of $\pdb$. Suppose that $\Phi$ consists of $m$ first-order formulae $\Phi_1,\dots,\Phi_m$ with arities $r_1,\dots,r_m$, respectively. Observe that due to the active domain semantics, for all $D \sim \pdb$ we have
    \[
        \size[\big]{ \Phi(D) }
            =
        \sum_{ i = 1 }^{ m } \size[\big]{ \Phi_i( D ) }
            \leq
        \sum_{ i = 1 }^{ m } \size[\big]{ \adom( D, \Phi_i ) }^{ r_i }
            \leq
        m \cdot \size[\big]{ \adom(D,\Phi) }^{ r }
            \leq 
        m \cdot \big(  r_{\mathrm{max}} \cdot \size{D} + \size{\adom( \Phi ) } \big)^r
    \]
    where $r = \max_{ i = 1,\dots,m } r_i$ and $r_{\mathrm{max}}$ is the maximum arity of a relation symbol in the schema of $\pdb$. Thus, for all $k \in \NN_+$ it follows that
    \[
        \Expectation_{ \mkern1mu\Phi\mkern-1mu( \pdb )\mkern-1mu}\big( \size{ \? }^k \big)
            \leq
        \Expectation_{ \pdb } \Big( m^k \cdot \big( \size{ \? } r_{\mathrm{max}} + \size{\adom(\Phi)} \big)^{ rk } \Big)
            = 
        m^k \cdot \sum_{ j = 0 }^{ rk } \binom{ rk }{ j } \cdot r_{\mathrm{max}}^j \cdot \size{ \adom(\Phi) }^{ rk - j } \cdot \Expectation_{ \pdb }\big( \size{ \? }^j \big)
        \text,
    \]
    using the binomial formula and linearity of expectation. Since $\pdb$ has the finite moments property, all the $\Expectation_{\pdb}\big( \size{ \? }^j \big)$ with $j = 0, 1, \dots, rk$ are finite. Thus, $\Expectation_{ \mkern1mu\Phi\mkern-1mu( \pdb )\mkern-1mu }\big( \size{\?}^k \big)$ is finite as well.
\end{proof}

We note, however, that the property of having a finite $k$th moment (but not necessarily a finite $(k+1)$st moment) for fixed $k$, is in general not preserved under $\FO$-views.

\begin{example}
    Consider the schema that only contains a single unary relation symbol $R$ and let $D_i = \set[\big]{ R(1), \dots, R(i) }$ for all $i \in \NN_+$. Let $P\big( \set{D_i} \big) = \frac{Z}{i^3}$ where $Z = \sum_{ i = 1 }^{ \infty }\frac{1}{i^3}$ is the normalizing constant. Then $\pdb = (\DD, P)$ is a PDB with expected instance size $\sum_{ i = 1 }^{ \infty } i \cdot \frac{Z}{i^3} = Z \cdot \sum_{i=1}^{\infty}\frac{1}{i^2} < \infty$. Note, however, that the second moment of the instance size of $\pdb$ is infinite.
    
    Now consider the view $\Phi(x,y) = R(x) \vee R(y)$. Then $\size[\big]{ \Phi( D_i ) } = \size{ D_i }^2$. Thus, $\Phi(\pdb)$ has expected instance size $\sum_{ i = 1 }^{ \infty } i^2 \cdot \frac{Z}{i^3} = Z \cdot \sum_{ i = 1 }^{ \infty } \frac{1}{i} = \infty$.
\end{example}

\medskip

With \cref{lem:fo-preserves-finite-moments} the finite moments property of $\TI$ (\cref{pro:finite-moments-ti}) directly extends to $\FOTI$.

\begin{proposition}\label{pro:finite-moments}
    $\FOTI$ has the finite moments property.
\end{proposition}

This insight yields numerous new examples of PDBs that are not in $\FOTI$.

\begin{example}
    Let $\pdb = (\DD,P)$ be a PDB with $\DD = \set{ D_1, D_2, \dots }$, where $\size{ D_i } = 2^i$ and $P\big( \set{D_i} \big) = \frac{3}{4^i}$. Then,
    \begin{align*}
        \Expectation_{\pdb} \big( \size{\?} \big) 
            &= 
        \sum_{D \in \DD} \size{D} \cdot P\big( \set{D} \big) 
            = 
        3 \cdot \sum_{ i=1 }^{\infty} 2^i \cdot \frac{1}{4^i} 
            = 
        3 \cdot \sum_{ i=1 }^{\infty} \frac{1}{2^i}
            = 
        3\text,
        \shortintertext{but}
        \Expectation_{\pdb} \big( \size{\?}^2 \big) 
            &= 
        \sum_{D \in \DD} \size{ D }^2 \cdot P\big( \set{D} \big) 
            = 
        3 \cdot \sum_{ i=1 }^{ \infty } 2^{2i} \cdot \frac{1}{4^i}
            =
        3 \cdot \sum_{ i=1 }^{ \infty } 1 = \infty
        \text.
    \end{align*}
    That is, $\pdb$ is a PDB of finite expected instance size where the second moment of the instance size is infinite. According to \Cref{pro:finite-moments}, $\pdb\not\in\FOTI$.
\end{example} %
\subsection{Balancing the Marginal Probabilities}\label{sec:negative-finite}

In this subsection, we show that the finite moments property is not a sufficient condition for membership in $\FO(\TI)$. That is, the main result of this subsection is the following.

\begin{theorem}\label{thm:not-all-finite}
    There are PDBs having the finite moments property that are not in $\FOTI$.
\end{theorem}

The proof of this theorem is quite involved. However, it may be skipped by the reader, as the techniques are not essential for the understanding of the rest of the paper.

To obtain \cref{thm:not-all-finite}, we essentially aim to break a balance in fact probabilities that is required for representations: On the one hand, in $\TI$-PDBs, the sum of fact probabilities must converge, meaning that these probabilities have to approach zero fast enough. On the other hand, fact probabilities in the $\TI$-PDB need to be sufficiently high in order to represent instances of a certain probability with respect to an $\FO$-view on $\Phi$. This is due to the fact that images of $\FO$-views only use domain elements that either appeared in the input instance, or as constants in the view itself. While the first part of this (the convergence of fact probabilities in the $\TI$-PDB) is clear, this section focuses first on formalizing the second part of our statement. In particular, we start by turning this idea into an upper bound for the probabilities of instances in the image of an $\FO$-view on a $\TI$-PDB.

\begin{lemma}\label{lemma:view-prob}
    Let $\pdb[I]$ be a $\TI$-PDB and let $\Phi$ be an $\FO$-view over the schema of $\pdb[I]$. Then for all instances $D$ of $\Phi( I )$ it holds that
    \[
        \Pr_{ I \sim \pdb[I] } \big( \Phi( I ) = D \big)
        \leq
        \size{ A_D^* } \cdot \bigg( r^2 \cdot \size{ A_D^* }^{ r- 1 } \cdot \sum_{ f \in F_D^* } p_f \bigg)^{ \size{ A_D^* }/r }\text,
    \]
    where $A_D^*$ is the set of domain elements appearing in $D$ that are not constants in $\Phi$, $F_D^*$ is the set of facts in $D$ that contain at least one element from $A_D^*$ and $r$ is the maximum arity among the relations of $\pdb[I]$.
\end{lemma}

The idea behind this auxiliary result connects to the previous discussion as follows: If $D$ is an instance in the image of $\Phi$ on $\pdb[I]$, then the active domain elements of $D$ must appear either in the active domain of an input instance from $\pdb[I]$ or as a constant in $\Phi$. Summing over all the possible ways to \enquote{cover} the active domain of $D$ (minus the constants of $\Phi$) with facts from the $\pdb[I]$ can thus be turned into an upper bound for the probability of $D$ being the output of the view.

\begin{proof}
    Let $\pdb[I] = ( \II, P )$ be a $\TI$-PDB and let $\Phi$ be an $\FO$-view over the schema of $\pdb[I]$. We let $\AA \coloneqq \bigcup_{ I \in \II } \adom( I )$ be the set of domain elements appearing among the instances of $\pdb[I]$ and $\FF \coloneqq \facts(\pdb[I])$. We say that a set $F \subseteq \FF$ \emph{covers} a finite set $A \subseteq \AA$ if for all $a \in A$ there exists a fact $f \in F$ such that $a$ appears in $f$. In that case, we call $F$ a \emph{fact cover} of $A$. A fact cover $F \subseteq \FF$ of $A$ is called \emph{minimal} if no proper subset of $F$ is a fact cover of $A$. Note that every fact cover contains a minimal fact cover. The set of all fact covers of $A$ is denoted by $\cover(A)$ and the set of all minimal fact covers of $A$ is denoted by $\mincover(A)$. 
    
    We start by bounding the probability of representing an instance in terms of minimal fact covers.
    
    \begin{claim}
        For all instances $D$ of\/ $\Phi( \pdb[I] )$ it holds that
        \begin{equation}\label{eq:mincover}
            \Pr_{ I \sim \pdb[I] }\big( \Phi(I) = D \big) 
                \leq 
            \sum_{ F \in \mincover( A_D^* ) }
                \prod_{ f \in F } p_f\text,
        \end{equation}
        where $A_D^* = \adom(D) \setminus \adom(\Phi)$.
    \end{claim}
    
    \begin{claimproof}
        Let $D$ be an arbitrary instance of $\Phi( \pdb[I] )$. For all instances $I \in \II$ with $\Phi(I) = D$ we have $\adom( D ) \subseteq \adom( I ) \cup \adom( \Phi )$, that is, $I$ is a fact cover of $A_D^*$. Therefore,
        \[
            \Pr_{ I \sim \pdb[I] } \big( \Phi(I) = D \big) 
                \leq
            \Pr_{ I \sim \pdb[I] } \Big( I \in \cover\big( A_D^* \big)\Big)
                =
            \sum_{ F \in \cover( A_D^* ) } 
                \Pr_{ I \sim \pdb[I] } \big( I = F \big)
                \leq
            \sum_{ F' \in \mincover( A_D^* ) }
                \sum_{ F \subseteq \FF \setminus F' }
                    \Pr_{ I \sim \pdb[I] }\big( I = F' \cup F'' \big)\text.
        \]
        The last step above is not an equality, because any fact cover can contain multiple minimal fact covers, so we might be double counting on the right-hand side.
        
        Now, for any set $F$ of facts, we let $\pdb[I][F]$ denote the restriction of the $\TI$-PDB $\pdb[I]$ to the fact set $F$. This is again a $\TI$-PDB that retains the marginal probabilities of the facts of $F$ (and has marginal probability $0$ for all other facts). By splitting instances into a minimal fact cover and the remaining facts, we then get
        \begin{align*}
            \sum_{ F' \in \mincover( A_D^* ) }
                \sum_{ F'' \subseteq \FF \setminus F' }
                    \Pr_{ I \sim \pdb[I] }\big( I = F' \cup F'' \big)
                &=
            \sum_{ F' \in \mincover( A_D^* ) }
                \sum_{ F'' \subseteq \FF \setminus F' }
                    \Pr_{ I \sim \pdb[I][F'] }\big( I = F' \big) 
                \cdot
                    \Pr_{ I \sim \pdb[I][\FF \setminus F'] }\big( I = F'' \big)\\
                &=
            \sum_{ F' \in \mincover( A_D^* ) }
                \Pr_{ I \sim \pdb[I][F'] }\big( I = F' \big) 
                    \cdot
                    \underbrace{\sum_{ F'' \subseteq \FF \setminus F' }
                    \Pr_{ I \sim \pdb[I][\FF \setminus F'] }\big( I = F'' \big)}_{ = 1 }\\
                &= 
            \sum_{ F' \in \mincover( A_D^* ) }
                \Pr_{ I \sim \pdb[I][F'] }\big( I = F' \big)\\
                &=
            \sum_{ F' \in \mincover( A_D^* ) }
                \prod_{ f \in F' } p_f\text.\qedhere
        \end{align*}
    \end{claimproof}
    
    From now on, let $D$ be an arbitrary but fixed instance of $\Phi(\pdb[I])$. We let $F_D^*$ denote the set of facts of $\pdb[I]$ that contain at least one element from $A_D^* = \adom(D) \setminus \adom(\Phi)$ and define a function $s_D^* \from F_D^* \to 2^{ A_D^* }$ as follows:
    \[
        s_D^*\big( f \big) \coloneqq \set{ a \in A_D^* \with a \text{ occurs in } f }\text.
    \]
    Essentially, $s_D^*$ turns $f$ into the set of the occurring domain elements (by forgetting the order and multiplicities) and restricts this set to $A_D^*$. The function $s_D^*$ is lifted to sets of facts $F$ by letting $s_D^*( F ) = \set[\big]{ s_D^*(f) \with f \in F } \subseteq 2^{A_D^*}$.
    
    Let $H$ be the hypergraph with vertex set $A_D^*$ and (hyper)edge set 
    \[
        s_D^*\big( F_D^* \big) 
            = 
        \set[\big]{ e \subseteq A_D^* \with e = s_D^*( f ) \text{ for some } f \in F_D^* }
        \text.
    \]
    A set of hyperedges $C$ is an \emph{edge cover} of $H$ if every vertex appears in at least one of the hyperedges from $C$. Such an edge cover $C$ is called \emph{minimal} if no proper subset of $C$ covers all vertices. We let $\minedgecover( H )$ denote the set of minimal edge covers of $H$. The connection between minimal fact covers of $A_D^*$ and minimal edge covers of $H$ is a follows.
    
    \begin{claim}
        A set of facts $F\subseteq F_D^*$ is a minimal fact cover of $A_D^*$, if and only if the set $s_D^*( F )$ is a minimal edge cover of $H$ with $\size{F} = \size{ s_D^*( F ) }$.
    \end{claim}
    
    \begin{claimproof}
        First, let $F$ be a minimal fact cover of $A_D^*$. Since $F$ is a fact cover of $A_D^*$, the hyperedges $s_D^*( F )$ are an edge cover of $H$. Note that $s_D^*$ is injective on $F$: if there were $f\neq f'$ with $s_D^*( f ) = s_D^*( f' )$, then $F \setminus \set{ f' }$ would be a fact cover of $A_D^*$, contradicting the minimality of $F$. Thus, $\size{ s_D^*(F) } = \size{ F }$. We still have to show that $s_D^*( F )$ is a \emph{minimal} edge cover. If $s_D^*( F )$ is not minimal, then there exists $e \in s_D^*( F )$ such that $s_D^*( F ) \setminus \set{ e }$ is still an edge cover of $H$. Then there exists $f \in F$ with $s_D^*(f) = e$ such that the domain elements occurring in $f$ are already covered by other facts from $F$. But this means that $F \setminus \set{f}$ is a fact cover of $A_D^*$ in contradiction to $F$ being minimal. Thus, $s_D^*( F )$ is a minimal edge cover of $H$.
        
        For the other direction, suppose that $F \subseteq F_D^*$ is a set of facts such that $s_D^*(F)$ is a minimal edge cover of $H$ with $\size{F} = \size{s_D^*(F)}$. First note that $F$ is a fact cover of $A_D^*$. Suppose $F$ is not minimal. Then there exists $F' \subsetneq F$ such that $F'$ is a minimal fact cover of $A_D^*$. But then $s_D^*( F' ) \subseteq s_D^*( F )$ is a minimal edge cover of $H$ by the first part of this proof. This contradicts $\size{ s_D^*( F' ) } = \size{ F' } < \size{ F } = \size{ s_D^*( F ) }$.
    \end{claimproof}
    
    From the above, it follows that
    \[
        \mincover\big( A_D^* \big) = \bigcup_{ C \in \minedgecover( H ) } \set{ F \subseteq F_D^* \with s_D^*( F ) = C \text{ and } \size{C} = \size{F}}
        \text.
    \]
    We emphasize that the union on the right-hand side is disjoint, since $s_D^*$ is a function. Next, we transform the bound from \eqref{eq:mincover} into a bound in terms of minimal edge covers.
    
    \begin{claim}
        The probability of representing $D$ is bounded as follows:
        \[
            \Pr_{ I \sim \pdb[I] }\big( \Phi(I) = D \big)
                \leq
            \sum_{ C \in \minedgecover( H ) }
                \bigg( \frac{1}{\size{C}} \sum_{ f \in F_D^* } p_f \bigg)^{ \size{C} }\text.
        \]
    \end{claim}
    
    \begin{claimproof}
        We continue our calculation of an upper bound for $\Pr_{ I \sim \pdb[I] }\big( \Phi(I) = D \big)$ from \eqref{eq:mincover} as follows:
        \[
            \Pr_{ I \sim \pdb[I] }\big( \Phi(I) = D \big)
                \mathrel{\overset{\eqref{eq:mincover}}{\leq}}
            \sum_{ F \in \mincover( A_D^* ) } \prod_{ f \in F } p_f
                =
            \sum_{ C \in \minedgecover( H ) }
                \sum_{ \substack{ F \in \mincover( A_D^* ) \\ s_D^*( F ) = C } }
                    \prod_{ f \in F } p_f
                =
            \sum_{ C \in \minedgecover( H ) }
                \sum_{ \substack{ F \subseteq F_D^* \\ s_D^*(F) = C  \\ \size{F} = \size{C}}}
                    \prod_{f \in F} p_f
            \text.
        \]
        Suppose now that $C = \set{ e_1, \dots, e_k }$ is a minimal edge cover of size $k$. Then
        \[
            \sum_{ \substack{ F \subseteq F_D^* \\ s_D^*(F) = C  \\ \size{F} = \size{C}}}
                \prod_{f \in F} p_f
                =
            \mkern-4mu
                \sum_{ \substack{ f_1 \in F_D^* \\ s_D^*( f_1 ) = e_1 } }
                \mkern-6mu\cdots\mkern-6mu
                \sum_{ \substack{ f_k \in F_D^* \\ s_D^*( f_k ) = e_k } }
                \mkern-4mu
                p_{f_1} \cdots p_{f_k}
                =
            \Bigg( \mkern-2mu\sum_{ \substack{ f_1 \in F_D^* \\ s_D^*( f_1 ) = e_1 } } \mkern-2mu p_{f_1} \Bigg)
                \cdots
                    \Bigg( \mkern-2mu\sum_{ \substack{ f_k \in F_D^* \\ s_D^*( f_k ) = e_k } } \mkern-2mu p_{f_k} \Bigg)
                =
            \prod_{ e \in C }
                \Bigg( \mkern-2mu \sum_{ \substack{ f \in F_D^* \\ s_D^*( f ) = e } } p_f \Bigg)
            \text.
        \]
        Applying the inequality of arithmetic and geometric means, we get that
        \[
            \prod_{ e \in C }\Bigg( \sum_{ \substack{ f \in F_D^* \\ s_D^*( f ) = e } } p_f \Bigg)
                \leq 
            \Bigg( \frac{1}{\size{ C }} \sum_{ e \in C } \sum_{ \substack{ f \in F_D^* \\ s_D^*( f ) = e } } p_f \Bigg)^{ \size{ C } }
                \leq
            \sum_{ C \in \minedgecover( H ) }
            \bigg( \frac{1}{\size{C}} \sum_{ f \in F_D^* } p_f \bigg)^{ \size{C} }
            \text.\qedhere
        \]
    \end{claimproof}

    Finally, we are ready to complete the proof of \cref{lemma:view-prob}. Let $r$ denote the maximum arity of any relation in the schema of $\pdb[I]$. Then each hyperedge in $H$ contains at most $r$ elements. Thus, the size of any minimal edge cover of $A_D^*$ is at least $\lceil\size{A_D^*}/r\rceil \geq \size{ A_D^* }/r$ and (trivially) at most $\size{A_D^*}$.
    
    The number of hyperedges in $H$ is at most $\sum_{ i = 1 }^r \binom{ \size{A_D^*} }{ i } \leq \sum_{ i = 1 }^{r } \size{A_D^*}^i \leq r \cdot \size{A_D^*}^r$. In particular, there are at most $\binom{ r \cdot \size{ A_D^* }^r }{ k } \leq \big( r \cdot \size{A_D^*}^r \big)^k$ minimal edge covers of size $k$ in $H$. Therefore,
    \begin{align*}
        \Pr_{ I \sim \pdb[I] }\big( \Phi(I) = D \big)
            \leq
        \sum_{ C \in \minedgecover( A_D^* ) }
        \bigg( \frac{1}{\size{C}} \sum_{ f \in F_D^* } p_f \bigg)^{ \size{C} }
            &=
        \sum_{ k = \lceil\size{A_D^*}/r\rceil }^{ \size{A_D^*} }
        \sum_{ \substack{ F \in \mincover( A_D^*) \\ \size{F} = k } }
            \bigg( \frac1k \sum_{ f \in F_D^* } p_f \bigg)^k\\
            &\leq 
        \sum_{ k = \lceil\size{A_D^*}/r\rceil }^{ \size{A_D^*} }
        \big( r \cdot \size{ A_D^* }^r \big)^k \cdot \bigg( \frac1k \sum_{ f \in F_D^* } p_f \bigg)^k\\
            &\leq 
        \sum_{ k = \lceil\size{A_D^*}/r\rceil }^{ \size{A_D^*} }
        \bigg( r^2 \cdot \size{ A_D^* }^{r-1} \sum_{ f \in F_D^* } p_f \bigg)^k
    \end{align*}
    The inequality remains intact if we decrease the value of the exponent on the right-hand side to a smaller non-negative number: if the value inside the parentheses is larger than $1$, the complete right-hand side will still be larger than $1$ as well. If the expression in the parentheses is at most one, then making the exponent smaller can only increase the value. Thus, we can infer that
    \[
        \Pr\big( \Phi(I) = D \big)
            \leq 
        \sum_{ k = \lceil\size{A_D^*}/r\rceil }^{ \size{A_D^*} }
        \bigg( r^2 \cdot \size{ A_D^* }^{r-1} \sum_{ f \in F_D^* } p_f \bigg)^{ \size{ A_D^* }/r}
            \leq
            \size{ A_D^* }\cdot \bigg( r^2 \cdot \size{A_D^*}^{r-1} \sum_{ f \in F_D^* } p_f \bigg)^{ \size{A_D^*}/ r }\text. 
    \]
    for all instances $D$ of $\Phi( \pdb[I] )$, concluding the proof of \cref{lemma:view-prob}.
\end{proof}

We can think of \cref{lemma:view-prob} as another necessary condition for representability in $\FOTI$. However, it is tied to a concrete representation $\Phi( \pdb[I] )$. In the next step, we turn this into a general necessary condition for the represented probabilistic database that holds regardless of the choice of representation. For this, we also include the convergence requirement of fact probabilities in any $\TI$-PDB that serves as a base for the representation. To cast this into requirement over the probabilities of the probabilities of possible world, we consider the class of domain disjoint PDBs: A probabilistic database $\pdb = (\DD,P)$ is called \emph{domain disjoint} if $\adom( D ) \cap \adom( D' ) = \emptyset$ for all $D, D' \in \DD$ with $D \neq D'$. Investigating this class allows us to identify a connection between the probabilities of possible worlds and the size of their active domains.

\begin{lemma}\label{lemma:disjoint-inequality}
    Let $\pdb = ( \DD, P ) \in \FOTI$ be domain-disjoint with instances $\DD = \set{ D_1, D_2, \dots }$. Let $d_n \coloneqq \size{ \adom(D_n) }$. Then there exists a constant $r \in \NN_+$ such that for every divergent series $\sum_{ n = 1 }^{ \infty } a_n = \infty$ of non-negative integers $a_n$ there are infinitely many $n \in \NN_+$ with
    \[
        \Pr_{ D \sim \pdb }\big( D = D_n \big) 
            <
        d_n \big( a_n \cdot d_n^{ r-1 } \big)^{d_n / r}
        \text.
    \]
\end{lemma}
\begin{proof}
    Suppose $\pdb = \Phi( \pdb[I] )$ where $\pdb[I]$ is a $\TI$-PDB and $\Phi$ is an $\FO$-view over the schema of $\pdb[I]$. Let $r$ be the maximum arity among the relations of the schema $\pdb[I]$. For all $n \in \NN_+$, we define $A_n \coloneqq \adom( D_n ) \setminus \adom( \Phi )$ and let $F_n \subseteq \FF = \facts(\pdb[I])$ denote the set of facts in $\pdb[I]$ that contain at least one domain element from $A_n$. Note that the $A_n$ are pairwise disjoint. 
    
    In the following, we abuse notation and write $a \in f$ if $a$ is a domain element appearing in the fact $f$. We know that for every $a \in \bigcup_{ I \in \worlds( \pdb[I] ) } \adom( I ) = \AA$ there exists at most one $n \in \NN_+$ such that $a \in A_n$. Thus, it holds that
    \begin{equation}\label{eq:sumFnfinite}
        \sum_{ n = 1 }^{ \infty } \sum_{ f \in F_n } p_f
            \leq
        \sum_{ n = 1 }^{ \infty } \sum_{ a \in A_n } \sum_{ \substack{ f \in \FF \\ a \in f} } p_f
            \leq 
        \sum_{ a \in \AA } \sum_{ \substack{ f \in \FF \\ a \in f } } p_f
            =
        \sum_{ f \in \FF } \sum_{ \substack{ a \in \AA \\ a \in f } } p_f
            \leq
        \sum_{ f \in \FF } r \cdot p_f
            =
        r \cdot \sum_{ f \in \FF } p_f
            <
        \infty
        \text,
    \end{equation}
    since $\pdb[I]$ is a well-defined $\TI$-PDB. We claim that for infinitely many $n \in \NN_+$ we have $\sum_{ f \in F_n } p_f < \frac{a_n}{r^2}$. Assume otherwise. Then for all but finitely many $n$ it holds that $\sum_{ f \in F_n } p_f \geq \frac{a_n}{r^2}$. Let $N_0$ denote the set of these $n$. Of course, $\sum_{ n \in N_0 } a_n$ still diverges. But then $\sum_{ n \in N_0 } \sum_{ f \in F_n } p_f \geq \frac{1}{r^2} \cdot \sum_{ n \in N_0 } a_n = \infty$, which means that $\sum_{ n = 1 }^{ \infty } \sum_{ f \in F_n } p_f$ contains a divergent subseries in contradiction to \eqref{eq:sumFnfinite}.
    
    Recall that the active domains of distinct possible worlds of $\pdb$ are disjoint. This means that for all large enough $n$ it holds that $\adom( D_n ) \cap \adom( \Phi ) = \emptyset$ and, in particular, $A_n = \adom( D_n )$ and $d_n = \size{ A_n }$. Let $N$ be the (infinite) set of indices $n$ for which $\sum_{ f \in F_n } p_f < \frac{ a_n }{ r^2 }$. Then, with \cref{lemma:view-prob}, for all $n \in N$ it holds that
    \[
        \Pr_{ D \sim \pdb }\big( D = D_n \big) 
            =
        \Pr_{ I \sim \pdb[I] }\big( \Phi( I ) = D_n \big)
            \leq
        \size{ A_n }\cdot \Big( r^2 \cdot \size{ A_n }^{ r-1 } \cdot \sum_{ f \in F_n } p_f \Big)^{ \size{ A_n } / r }
            <
        d_n \cdot \big( a_n \cdot d_n^{r-1} \big)^{ d_n / r }
        \text.\qedhere
    \]
\end{proof}

\begin{remark}
    The argument from the proof of \cref{lemma:disjoint-inequality} remains valid also in the case where the active domains of the instances $D_n$ are not disjoint, but there exists a constant bound on the number of worlds $D_n$ that contain a domain element $a \in \AA$.
\end{remark}

The condition from \cref{lemma:disjoint-inequality} is a necessary condition for the representability of domain-disjoint PDBs in $\FOTI$. We can now construct a domain-disjoint PDB that violates this condition, despite having the finite moments property.

\begin{example}\label{ex:no-rep-fin-moments}
    Let $\pdb = ( \DD, P )$ be a domain-disjoint PDB such that $\DD = \set{ D_1, D_2 \dots }$ with $d_n \coloneqq \size[\big]{ \adom(D_n) } = \ceil[\big]{ \log(n) }$ and $P\big( \set{D_n} \big) = \frac{Z}{n^2}$ for all $n \in \NN_+$, where $Z = \frac{6}{\pi^2}$ is the normalizing constant that makes $P$ is a probability distribution. Note that $d_n = \ceil[\big]{ \log(n) } \in o( n^\epsilon )$ for all $\epsilon > 0$. Thus, for all $\epsilon > 0$ there exists $n(\epsilon) \in \NN$ such that for all $n > n(\epsilon)$ it holds that $d_n \leq n^{\epsilon}$.
    
    First, note that this PDB has the finite moments property. To see this, let $k \in \NN_+$ be arbitrary but fixed and let $\epsilon \coloneqq \frac1{2k}$. Then
    \[
        \Expectation_{\pdb}\big( \size{\?}^k \big)
            =
        \sum_{ n = 1 }^{ \infty } \frac{ Z \cdot d_n^k }{ n^2 }
            \leq
        \sum_{ n = 1 }^{ n(\epsilon) } \frac{ Z \cdot d_n^k }{ n^2 } + \sum_{ n > n(\epsilon) } \frac{ Z \cdot n^{\epsilon k} }{ n^2 }
            =
        \sum_{ n = 1 }^{ n(\epsilon) } \frac{ Z \cdot d_n^k }{ n^2 } + Z \cdot \sum_{ n > n(\epsilon) } n^{-3/2}\text,
    \]
    which is finite, since the first sum has only finitely many terms and the second sum is well-known to converge.
    
    \medskip
    
    Now towards a contradiction, assume that $\pdb \in \FO( \TI )$. As $\sum_{ n = 1 }^{ \infty } \frac1n$ diverges, by \cref{lemma:disjoint-inequality} there exists a constant $r \in \NN_+$ such that infinitely many $n \in \NN_+$ satisfy
    \begin{equation}
        \frac{Z}{n^2} 
            =
        \Pr_{ D \sim \pdb }\big( D = D_n \big)
            <
        d_n \cdot \Big( \frac{d_n^{r-1}}{n} \Big)^{ d_n / r }
        \text.\label{eq:Zn2}
    \end{equation}
    
    For all $n > n\big( \frac{1}{r} \big)$, it holds that $d_n \leq n^{1/r}$. We fix such an $n$ that additionally satisfies $n > \frac{1}{Z}$ and $d_n \geq 3r^2+r$.
    Then it holds that
    \[
        d_n \cdot \Big( \frac{ d_n^{ r-1 } }{ n } \Big)^{ d_n / r }
            \leq
        n^{1/r} \big( (n^{1/r})^{r-1} \cdot n^{-1} \big)^{ d_n / r }
            =
        n^{1/r} \big( n^{-1/r} \big)^{ d_n / r }
            = 
        \big( n^{1/r} \big)^{ 1 - \frac{d_n}{r} }
            =
        n^{ (r - d_n)/r^2 }
            \leq
        n^{-3}
            <
        \frac{Z}{n^2}\text,
    \]
    contradicting \eqref{eq:Zn2}.
\end{example}

The PDB from \cref{ex:no-rep-fin-moments} is a PDB with the finite moments property that has no representation in $\FO(\TI)$. In particular, it proves the main theorem (\cref{thm:not-all-finite}) of this section.
\section{Conditional Views}\label{sec:conditions}

In this section, we define conditional representations: representations of PDBs as views over $\TI$-PDBs conditioned on an $\FO$-sentence. The possible worlds are restricted to those satisfying the condition, and the probability mass of the valid instances is scaled to add up to one. As the next sections demonstrate, when constructing a representation for a given PDB, conditional representations are often simpler to identify and explain compared to unconditional ones. We show the equivalence between the PDBs that admit an $\FO$-representation and those that admit a conditional $\FO$-representation. As a consequence, we obtain a tool for showing that a PDB has a representation as an $\FO$-view over a $\TI$-PDB: it is enough to identify such a representation that is conditioned on an $\FO$-sentence. 
\medskip

We start by defining conditional views.
Given a PDB $\pdb = (\DD, P)$ and an $\FO$-sentence $\phi$ such that $\Pr_{D\sim\pdb}\big( D\models\phi \big)>0$,
we denote by $\pdb \under \phi$ the PDB $(\DD_{\phi}, P')$ where 
\begin{align*}
    \DD_{\phi} 
    &\coloneqq \set[\big]{D \in \DD \mid D \models \phi}
    \shortintertext{and for all \( D \in \DD_{\phi} \), }
    P'\big(\set{D}\big) 
    &\coloneqq P\big( \set{D} \under \DD_{\phi}\big) 
    = \frac{ P\big( \set{D} \big) }{ P\big(\DD_{\phi}\big) }
    \text.
\end{align*}
Given a class $\classstyle{D}$ of PDBs, we denote by $\classstyle{D}\under\FO$ the class of all PDBs obtained by conditioning a PDB of $\classstyle{D}$ on an $\FO$-sentence. That is,
\begin{equation*}
    \classstyle{D}\under\FO 
    \coloneqq 
    \set[\big]{ 
        ( \pdb \under \phi ) \with
        \pdb\in\classstyle{D}\text{ and }\phi \in \FO \text{ sentence with }\Pr_{D\sim\pdb}(D\models\phi)>0
    }\text.
\end{equation*}

The following is the main result of this section, stating that the class of PDBs that can be represented as $\FO$-views over \emph{$\FO$-conditioned} $\TI$-PDBs coincides with the class of PDBs that can be represented by an $\FO$-view of a $\TI$-PDB alone.

\begin{theorem}\label{thm:conditional-views}
    $\FOTIFO = \FOTI$.
\end{theorem}

\begin{remark}
    Note that $\FOTIFO \subseteq \FOTI$ is by no means trivial. In particular, we cannot simply merge the condition given as an $\FO$-sentence into the $\FO$-view. Composing them cannot work, as the condition is a sentence, and so the composition can only result in two outcomes and cannot represent PDBs with more than two possible worlds. Intersecting them will not work either. The conditional view removes the possible worlds that do not meet the condition from the sample space and scales the probability mass of the remaining possible worlds up to one, while intersecting the $\FO$-view with the condition keeps the probability mass of the invalid possible worlds but renders these worlds empty.
\end{remark}

\begin{proof}[Proof (\cref{thm:conditional-views})]
    The direction $\FOTI \subseteq \FOTIFO$ is immediate. Instead of showing $\FOTIFO \subseteq \FOTI$, we prove the equivalent statement $\TI\under\FO \subseteq \FOTI$ (this is equivalent due to $\TI\under\FO \subseteq \FOTIFO$ and $\FOTI = \FO(\FOTI)$). Thus, let $\pdb[I] = \big( \II, P_{\pdb[I]} \big)$ be a $\TI$-PDB and let $\phi$ be an $\FO$-sentence for which $\pdb[I] \under \phi$ is well-defined (that is, $P_{\pdb[I]}( \phi ) > 0$). Our goal is to show that $\pdb[I]_{\phi} \coloneqq \pdb[I] \under \phi \in \FOTI$. 
    If $P_{\pdb[I]}( \phi ) = 1$, then we are done, since then $\pdb[I]_\phi = \pdb[I] \in \TI \subseteq \FOTI$. Therefore, for the remainder of the proof, we assume that 
    \[
        p_{\phi} \coloneqq P_{\pdb[I]}( \phi ) \in (0,1)
        \text.
    \]
    We fix some instance $I_0 \in \II$ of positive probability in $\pdb[I]_{\phi}$. If $\Pr_{ I \sim \pdb[I]_{\phi} }( I = I_0 ) = 1$, then $I_0$ is the only instance of positive probability in $\pdb[I]_{\phi}$. This means that $\pdb[I]_{\phi}$ is a $\TI$-PDB where any fact $R(\tup t)$ has marginal probability $1$ if $R(\tup t) \in I_0$, and marginal probability $0$ otherwise. Consequentially, $\pdb[I]_{\phi} \in \TI \subseteq \FOTI$ and we are done. Thus, from now on, assume that 
    \[
        p_0 \coloneqq \Pr_{ I \sim \pdb[I]_\phi }\big( I = I_0 ) \in (0,1)
        \text.
    \]
    
    Before we continue we want to convey the main idea of the proof. We are going to construct a new $\TI$-PDB consisting of multiple independent copies of the $\TI$-PDB $\I$. We define how instances of the new PDB represent instances of the PDB $\pdb[I]_{\phi}$ other than $I_0$. This representation will change instance probabilities by a constant factor greater than $1$. All remaining instances of the new PDB are taken as the representation of $I_0$. Since the other probabilities are scaled up, the probability mass of these instances is smaller than the probability of $I_0$ in $\pdb[I]_{\phi}$. Therefore, we install a gadget that acts like a switch between the two cases and carefully configure its probability to reconstruct the probability distribution of $\pdb[I]_{\phi}$. Finally, we show that the original instances can be reobtained from our representation by the means of an $\FO$-view.\par\bigskip
    
    Assume that $I_0 = \bigcup_{ n = 1 }^{ N } \set{ R_n( \tup a_{ni} ) \with i = 1,\dots, r_n }$ where $N$ is the number of relations ($R_1,\dots,R_N$) in the schema of $\pdb[I]$ and for each $n=1,\dots, N$, the number $r_n$ is the arity of the relation symbol $R_n$. Then $I_0$ can be exactly characterized by an $\FO$-sentence $\phi_0$ as follows:
    \[
        \phi_0 \coloneqq \bigwedge_{ n = 1 }^{ N } \bigg( \forall \tup x \colon R_n( \tup x ) \leftrightarrow \bigvee_{ i = 1 }^{ r_n } \tup x = \tup a_{ni} \bigg)
        \text.
    \]
    That is, for all $I \in \II$ it holds that $I \models \phi_0$ if and only if $I = I_0$. We let
    \begin{align*}
        \psi &\coloneqq \phi \wedge \neg \phi_0
        \shortintertext{and}
        p_{\psi} &\coloneqq P_{\pdb[I]}( \psi )
        \text.
    \end{align*}

    \begin{claim}[Separating $I_0$]\label{cla:psi-probs}
        It holds that
        \begin{align}
            p_{\psi}        &= \big( 1- p_0 \big) \cdot p_{\phi} \in (0,1)
            \shortintertext{and for all \(I \in \II\setminus \set{ I_0 }\) it holds that}
            P_{\pdb[I]}\big( \set{ I } \under \psi \big)
                            &= \frac{ P_{\pdb[I]_{\phi}}\big( \set{I} \big) }{ 1 - p_0 }\label{eq:I-under-psi}
            \text.
        \end{align}
    \end{claim}
    
    \begin{claimproof}
        We start by calculating the probability of $\psi$ being true in $\pdb[I]$, given that $\phi$ holds:
        \[
            P_{\pdb[I]}\big( \psi \under \phi \big) 
                =
            P_{\pdb[I]}\big( \neg\phi_0 \under \phi \big)
                =
            1 - P_{\pdb[I]}\big( \phi_0 \under \phi \big)
                =
            1 - p_0
            \text.
        \]
        Moreover, for all $I \in \II\setminus \set{I_0}$ it holds that
        \begin{equation}
            P_{ \pdb[I] }\big( \set{ I } \under \psi \big)
                =
            \frac{ P_{\pdb[I]}\big( \set{I}, \psi \big) }{ P_{\pdb[I]}(\psi) }
                =
            \frac{ P_{\pdb[I]}\big( \set{I}, \phi \big) }{ P_{\pdb[I]}(\psi,\phi) }
                =
            \frac{ P_{\pdb[I]}\big( \set{I} \under \phi \big) }{ P_{\pdb[I]}\big( \psi \under \phi \big) }
                =
            \frac{ P_{\pdb[I]_{\phi}}\big( \set{I} \big) }{ 1 - p_0 }
            \text.
        \end{equation}
        Therein, for the second equality we used that $\psi \equiv \psi \wedge \phi$ in the denominator. For the numerator, we used the fact that if $I \not\models \phi_0$ (which is the case since $I\neq I_0$), then $I \models \psi$ if and only if $I \models \phi$.
        The probability of the condition $\psi$ being satisfied in $\pdb[I]$ is thus calculated as follows:
        \[
            p_{ \psi }
                \coloneqq
            P_{ \pdb[I] }\big( \psi \big) 
                =
            P_{ \pdb[I] }\big( \psi \under \phi \big) \cdot P_{\pdb[I]}\big( \phi \big)
                =
            \big( 1 - p_0 \big) \cdot p_{ \phi }
            \text.\qedhere
        \]
    \end{claimproof}

    Note that since $p_0,p_{\phi} \in (0,1)$, it follows that $p_{\psi} = \big( 1- p_0 \big) \cdot p_{\phi}\in (0,1)$. Thus, we can choose $k \in \NN$ large enough such that
    \[
        \big( 1 - p_{\psi} \big)^k 
            = 
        \big( 1 - P_{\pdb[I]}\big( \psi \big) \big)^k 
            < 
        p_0\text.
    \]
    We fix such a number $k$ and construct a new $\TI$-PDB $\pdb[J] = ( \JJ, P_{\pdb[J]} )$ with schema and marginal probabilities as follows:
    \newcommand*{\LEQ}{\mathrm{LEQ}}
    \begin{enumerate}
        \item \textbf{Schema of $\pdb[J]$:}
            \begin{itemize}
                \item Per relation symbol $R$ of arity $r$ of the schema of $\pdb[I]$, the schema of $\pdb[J]$ contains a distinguished relation symbol $R'$ of arity $r+1$.
                \item We add the set $\set{1,\dots,k}$ to the domain of $\pdb[J]$ (assuming, without loss of generality, that these are \emph{new} domain elements that do not exist in $\pdb[I]$.
            \end{itemize}
        \item \textbf{Marginal probabilities in $\pdb[J]$:}
            \begin{itemize}
                \item For every fact $R( \tup a )$ appearing in $\pdb[I]$ with marginal probability $p$, $\pdb[J]$ contains the facts $R'(1,\tup a), \dots, R'(k,\tup a)$, each with marginal probability $p$.
                \item All other facts have marginal probability $0$.
            \end{itemize}
    \end{enumerate}
    Since $\pdb[I]$ is well-defined, so is $\pdb[J]$, as the sum of all fact probabilities in $\pdb[J]$ is bounded from above by $k\cdot \Expectation_{\pdb[I]}\big( \size{\?} \big) < \infty$.
    The $\TI$-PDB $\pdb[J]$ can be thought of as consisting of $k$ independent copies of the original $\TI$-PDB $\pdb[I]$. The copies in $\pdb[J]$ can be distinguished by the identifier $i \in \set{1,\dots,k}$ in its facts. 
    
    If $J$ is an instance of $\pdb[J]$, and $i \in \set{1,\dots,k}$, we let
    \[
        J[i] \coloneqq \set[\big]{ R( \tup a ) \with R'(i ,\tup a) \in J }
        \text.
    \]
    This is then an instance in $\II$. Intuitively, $J[i]$ is obtained from $J$ by selecting the facts with copy identifier $i$, and projecting the identifier out. We introduce the following terminology:
    \begin{itemize}
        \item We call $J$ a \emph{representation} if there exists some $i \in \set{1,\dots,k}$ such that $J[i] \models \psi$.
        \item We say $J$ \emph{represents} $I \in \II \setminus \set{ I_0 }$ if $J$ is a representation such that for $i^* = \min \set{ i \with J[i] \models \psi }$ we have $J[i^*] = I$.
    \end{itemize}
    Note that we excluded the special instance $I_0$ from the definition above, and that the set of instances $I$ we can represent in $\pdb[J]$ is exactly the set of instances in $\pdb[I]$ that satisfy $\phi$ (minus the instance $I_0$).
    The reason for this is that we tie representations to the sentence $\psi = \phi \wedge \neg \phi_0$, and that $J[i] \models \psi$ entails that $J[i] \neq I_0$. Observe that for all $i \in \set{1,\dots,k}$ it holds that 
    \begin{align}
        \Pr_{ J \sim \pdb[J] }\big( J[i] \models \psi \big)
            = {} 
        & \Pr_{ I \sim \pdb[I] }\big( I \models \psi \big)
            = {}
        p_{\psi}
        \label{eq:Ji-models-psi}
    \shortintertext{and}
        \Pr_{ J \sim \pdb[J] }\big( J[i] = I \under J[i] \models \psi \big) 
            = {}
        & P_{ \pdb[I] }\big( \set{ I } \under \psi \big)
            \mathrel{\overset{\text{\eqref{eq:I-under-psi}}}{=}} {}
            \frac{ P_{ \pdb[I]_{\phi} }\big( \set{ I } \big) }{ 1 - p_0 }
        \label{eq:Ji-rep-I-under-Ji-models-psi}
    \end{align}
    for all $I \in \II \setminus \set{ I_0 }$. The probability that an instance $J \sim \pdb[J]$ is a representation is given by
    \newcommand*{\rep}{\mathrm{rep}}%
    \begin{align}
        p_{\rep}
            \coloneqq 
        \Pr_{ J \sim \pdb[J] }\big( J \text{ is a representation} \big)
            ={}&
        1 - \Pr_{ J \sim \pdb[J] }\big( J[i] \not\models \psi \text{ for all } i = 1, \dots, k \big)\notag{}\\
            ={}&
        1 - \big( 1-p_{\psi} \big)^k
            >
        1 - p_0\text,
        \label{eq:J-representation}
    \end{align}
    by the choice of $k$. Note that $p_{\rep} < 1$, as $p_{\psi} < 1$.

    \begin{claim}[Representations in {$\pdb[J]$}]
        For all $I \in \II \setminus \set{ I_0 }$ it holds that
        \[
            \Pr_{ J \sim \pdb[J] }\big( J \text{\textup{ represents }} I \big) 
                = 
            \frac{ P_{ \pdb[I]_{\phi} }\big(\set{ I }\big) }{ 1 - p_0 } \cdot p_{\rep} \in \big( P_{\pdb[I]_{\phi}}(\set{I}),1\big)
        \]
    \end{claim}
    
    \begin{claimproof}
    Let $I \in \II\setminus \set{ I_0 }$, and first consider the probability of the event \enquote{$J \text{ represents } I$}, conditioned on $J$ being a representation:
    \begin{align*}
        &
        \Pr_{ J \sim \pdb[J] }\big( J \text{ represents } I \under J \text{ is a representation} \big)\\
            = {}
        & \sum_{ i^* = 1 }^{ k } \begin{multlined}[t]
            \Pr_{ J \sim \pdb[J] }\big( J[i^*] = I \under J[i^*] \models \psi \text{\textbf{ and }} J[i] \not\models \psi \text{ for all } i < i^* \big)\\
            {}\cdot{}\Pr_{ J \sim \pdb[J] }\big( J[i^*] \models \psi \text{\textbf{ and }} J[i] \not\models \psi \text{ for all } i < i^* \under J \text{ is a representation}\big)
        \end{multlined}\\
            = {}
        & \begin{multlined}[t]
            \sum_{ i^* = 1 }^{ k } \Pr_{ J \sim \pdb[J] }\big( J[i^*] = I \under J[i^*] \models \psi \big)\\
                \cdot \Pr_{ J \sim \pdb[J] }\big( J[i^*] \models \psi \text{\textbf{ and }} J[i] \not\models \psi \text{ for all } i < i^* \under J \text{ is a representation}\big)
        \end{multlined}\\
            \mathrel{\overset{\makebox[0pt]{\scriptsize\eqref{eq:Ji-rep-I-under-Ji-models-psi}}}{=}}{}
        & \frac{ P_{ \pdb[I]_\phi }\big( \set{ I } \big) }{ 1 - p_0 } \cdot \sum_{ J \sim \pdb[J] }\big( J[i^*] \models \psi \text{\textbf{ and }} J[i] \not\models \psi \text{ for all }i < i^* \under J \text{ is a representation}\big)
        \text.
    \end{align*}
    In the above, in the first step, we used that the events $\big(  J[i^*] \models \psi \text{\textbf{ and }} J[i] \not\models \psi \text{ for all } i<i^* \big)_{i^*}$ form a partition of $I$. In the second step, in addition to \eqref{eq:Ji-rep-I-under-Ji-models-psi}, we exploited that statements about $J[i]$ are independent of statements about the $J[i']$ with $i' \neq i$. 
    Furthermore, it follows from $\sum_{ J \sim \pdb[J] }\big( J[i^*] \models \psi \text{\textbf{ and }} J[i] \not\models \psi \text{ for all }i < i^* \under J \text{ is a representation}\big) = 1$ that
    \[
        \Pr_{ J \sim \pdb[J] }\big( J \text{ represents } I \under J \text{ is a representation} \big)
        =
        \frac{ P_{ \pdb[I]_{\phi} } \big( \set{ I } \big) }{ 1 - p_0 }
        \text.
    \]
    Combining this with \eqref{eq:J-representation}, we conclude that for $I \in \II \setminus \set{ I_0 }$ it holds that
    \begin{align}
        \Pr_{ J \sim \pdb[J] }\big( J \text{ represents } I\big)
            = {}
        & \Pr_{ J \sim \pdb[J] }\big( J \text{ represents } I \under J \text{ is a representation} \big) \cdot \Pr_{ J \sim \pdb[J] }\big( J \text{ is a representation} \big) 
        \notag{}\\
            = {}
        & \frac{ P_{ \pdb[I]_\phi }\big( \set{ I } \big) }{ 1 - p_0 } \cdot p_{\rep} \in \big( P_{\pdb[I]_{\phi}}(\set{I}),1\big)
        \text.
        \qedhere
    \end{align}%
    \end{claimproof}

    Starting from $\pdb[J]$, we construct yet another $\TI$-PDB $\pdb[J]_{\bot}$ that emerges from $\pdb[J]$. We add a new, independent dummy fact $f_\bot = R_{\bot}(\bot)$ to $\pdb[J]$ (where $R_{\bot}$ is a new unary relation symbol and $\bot$ is a new domain element), with marginal probability
    \[
        p_{\bot} \coloneqq \frac{ p_{ \rep } - ( 1 - p_0 ) }{ p_{\rep} }\text,
    \]
    leaving all other facts and their marginal probabilities unchanged. Note that it follows from \eqref{eq:J-representation} that $0 < p_{\bot} < 1$. Note that for all $J \sim \pdb[J]_{\bot}$ it holds that $J \setminus \set{f_\bot}$ is an instance of $\pdb[J]$. With the new $\TI$-PDB $\pdb[J]_{\bot}$, we now incorporate the representation of the special instance $I_0$ as follows: we say that $J \sim \pdb[J]_{\bot}$
    \begin{itemize}
        \item \emph{represents} $I_0$ if either $J \setminus \set{ f_\bot }$ is \emph{not} a representation in $\pdb[J]$ \textbf{or} $f_\bot \in J$ and $J \setminus \set{ f_\bot } $ \emph{is} a representation in $\pdb[J]$,
        \item \emph{represents} $I \neq I_0$ if $J$ does not represent $I_0$ \textbf{and} $J \setminus\set{ f_\bot }$ represents $I$ in $\pdb[J]$.
    \end{itemize}
    Observe that every instance $J$ of $\pdb[J]_{\bot}$ represents exactly one instance of $\pdb[I]$. A sketch of the overall situation that we have so far is shown in \cref{fig:independent-copies-sketch}.
    We now show that in $\pdb[J]_{\bot}$, the probability of the set of instances representing $I \in \II$ corresponds to the probability of $I$ in $\pdb[I]_{\phi}$:

    \begin{claim}[Representations in {$\pdb[J]_{\bot}$}]\label{cla:J-bot-probs}
        For all $I \in \II$ it holds that
        \[
            \Pr_{ J \sim \pdb[J]_{ \bot } }\big( J \text{\normalfont{} represents } I \big) 
                = 
            P_{\pdb[I]_{\phi}}\big( \set{ I } \big)
                =
            P_{ \pdb[I] }\big( \set{I} \under \phi \big)
            \text.
        \]
    \end{claim}
    
    \begin{claimproof}
        This is easily verified by direct calculation. It holds that
        \begin{align*}
            & \Pr_{ J \sim \pdb[J]_{\bot} }\big( J \text{ represents } I_0 \text{ (in }\pdb[J]_{\bot} \text)\big) \\
                = {}
            & \Pr_{ J \sim \pdb[J]_{\bot} }\big( J \setminus \set{ f_\bot } \text{ is no representation in } \pdb[J] \big) +
            \Pr_{ J \sim \pdb[J]_{\bot} }\big( f_\bot \in J \text{\textbf{ and }} J \setminus \set{ f_\bot } \text{ is a representation in } \pdb[J] \big)\\
                = {}
            & \big( 1 - p_{\rep} \big) + \big( p_{\bot} \cdot p_{ \rep } \big)
                =
            \big( 1 - p_{\rep} \big) + \big( p_{\rep} - ( 1 - p_0 ) \big)
                = {}
            p_0
                = {}
            P_{ \pdb[I]_{\phi} }\big( \set{ I_0 } \big)\text.
        \intertext{Moreover, for all $I \in \II \setminus \set{ I_0 }$ it holds that}
            & \Pr_{ J \sim \pdb[J]_{ f_\bot} }\big( J \text{ represents } I \text{ (in }\pdb[J]_{\bot}\text) \big) \\
                = {}
            & \Pr_{ J \sim \pdb[J]_{ \bot} }\big( f_\bot \notin J \text{\textbf{ and }} J\setminus \set{ f_\bot } \text{ represents } I \text{ (in }\pdb[J]\text) \big)\\
                = {}
            & \big( 1 - p_{\bot} \big) \cdot \frac{ P_{ \pdb[I]_{\phi} }\big( \set{ I }\big) }{ 1 - p_0 } \cdot p_{\rep}
                = {}
            \frac{ (1-p_0) - p_{\rep} }{ p_{\rep} } \cdot \frac{ P_{ \pdb[I]_{\phi} }\big( \set{ I }\big) }{ 1 - p_0 } \cdot p_{\rep}
                = {}
            P_{\pdb[I]_{\phi}}\big( \set{ I } \big)\text.\qedhere
        \end{align*}
    \end{claimproof}
    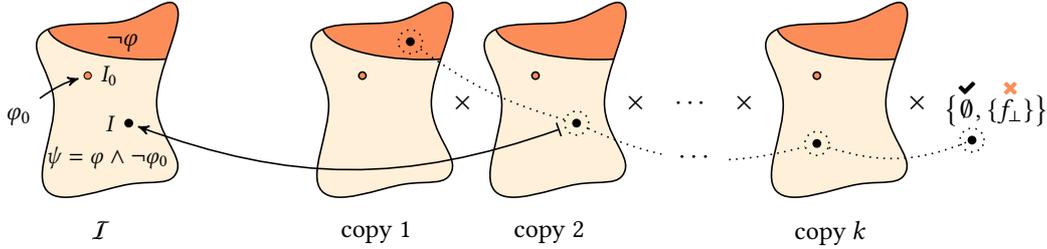
\begin{figure}[h]
        \centering
        \begin{tikzpicture}[scale=0.9, transform shape]
        
            \def\pdbshape#1{plot [smooth cycle,tension=.7] coordinates { 
                ($#1 + (-0.2,0)$)
                ($#1 + (0.1,1.4)$)
                ($#1 + (-0.7,1.2)$)
                ($#1 + (-1.9,1)$)
                ($#1 + (-1.7,0)$)
                ($#1 + (-1.8,-1.3)$)
                ($#1 + (-0.6,-1.2)$)
                ($#1 + (0.1,-1)$)
            }}
            
            \def\notphi#1{
                #1 ellipse (1.5cm and 0.5cm)
            }
        
            \tikzset{pdb/.style={draw=none,ellipse,minimum width=2cm,minimum height=3.2cm,inner sep=0}}
            \tikzset{dot/.style={fill=black,draw=black,minimum width=3pt, inner sep=0pt,circle}}
            \node[pdb,label={below:$\pdb[I]$}] (I) at (0,0) {};
            \node[pdb,label={below:copy $1$},right=2cm of I] (J1) {};
            \node[right=0cm of J1.east] (left) {$\times$};
            \node[pdb,label={below:copy $2$},anchor=west,right=0cm of left] (J2) {};
            \node[right=0cm of J2.east] (middle) {$\times\quad\dotsm\quad\times $};
            \node[pdb,label={below:copy $k$},anchor=west,right=0cm of middle.east] (Jk) {};
            \node[right=0cm of Jk.east] (right) {$\times$};
            \node[rectangle,anchor=west,right=0cm of right] (bot) {$\set[\big]{ \overset{\text{\faCheck}}{\emptyset}, \overset{\text{\textcolor{ShadingDark}{\faClose}}}{\set{f_{\bot}}} }$};

            \begin{scope}
                \clip \pdbshape{(I.east)};
                \fill[ShadingLight] \pdbshape{(I.east)};
                \clip \notphi{(I.north east)};
                \fill[ShadingDark] \notphi{(I.north east)};
            \end{scope}
            \node[dot,fill=ShadingDark,shift={(-.2,.4)},label={[font=\small]right:$I_0$}] (I0) at (I) {};
            \draw[semithick] \pdbshape{(I.east)};
            \begin{scope}
                \clip \pdbshape{(I.east)};
                \draw[semithick] \notphi{(I.north east)};
            \end{scope}
            \node[xshift=-.4cm,yshift=-.25cm] at (I.north east) (phi) {$\neg\phi$};
            \node[xshift=-.2cm,yshift=-.2cm] (phi0) at (I.west) {$\phi_0$};
            \draw[-stealth',shorten >=2pt,semithick] (phi0) to[bend left=10] (I0) {};
            \node[above=.5cm of I.south,font=\small,xshift=.1cm] {$\psi = \phi \wedge \neg \phi_0$};            
            \node[dot,semithick,shift={(.4,-.3)},label={[font=\small]left:$I$}] (Iinst) at (I) {};

            \begin{scope}
                \clip \pdbshape{(J1.east)};
                \fill[ShadingLight] \pdbshape{(J1.east)};
                \clip \notphi{(J1.north east)};
                \fill[ShadingDark] \notphi{(J1.north east)};
            \end{scope}
            \draw[semithick] \pdbshape{(J1.east)};
            \node[dot,semithick,fill=ShadingDark,shift={(-.2,.4)}] (I0) at (J1) {};
            \begin{scope}
                \clip \pdbshape{(J1.east)};
                \draw[semithick] \notphi{(J1.north east)};
            \end{scope}
            \node[dot,semithick,shift={(.5,.9)}] (Jinst1) at (J1) {};

            \begin{scope}
                \clip \pdbshape{(J2.east)};
                \fill[ShadingLight] \pdbshape{(J2.east)};
                \clip \notphi{(J2.north east)};
                \fill[ShadingDark] \notphi{(J2.north east)};
            \end{scope}
            \draw[semithick] \pdbshape{(J2.east)};
            \node[dot,semithick,fill=ShadingDark,shift={(-.2,.4)}] (I0) at (J2) {};
            \begin{scope}
                \clip \pdbshape{(J2.east)};
                \draw[semithick] \notphi{(J2.north east)};
            \end{scope}
            \node[dot,semithick,shift={(.4,-.3)}] (Jinst2) at (J2) {};

            \begin{scope}
                \clip \pdbshape{(Jk.east)};
                \fill[ShadingLight] \pdbshape{(Jk.east)};
                \clip \notphi{(Jk.north east)};
                \fill[ShadingDark] \notphi{(Jk.north east)};
            \end{scope}
            \draw[semithick] \pdbshape{(Jk.east)};
            \node[dot,semithick,fill=ShadingDark,shift={(-.2,.4)}] (I0) at (Jk) {};
            \begin{scope}
                \clip \pdbshape{(Jk.east)};
                \draw[semithick] \notphi{(Jk.north east)};
            \end{scope}
            \node[dot,semithick,shift={(-.2,-.6)}] (Jinstk) at (Jk) {};
            
            \node[circle,fill=none,draw=black,dotted,semithick,inner sep=0pt,minimum width=10pt] (Jinst1') at (Jinst1) {};
            \node[circle,fill=none,draw=black,dotted,semithick,inner sep=0pt,minimum width=10pt] (Jinst2') at (Jinst2) {};
            \node[circle,fill=none,draw=black,dotted,semithick,inner sep=0pt,minimum width=10pt] (Jinstk') at (Jinstk) {};
            \node[below=0cm of bot,xshift=-10pt,dot,semithick] (botinst) {};
            \node[circle,fill=none,draw=black,dotted,semithick,inner sep=0pt,minimum width=10pt] (botinst') at (botinst) {};
            
            \draw[dotted,semithick] (Jinst1') to[bend right=10] (Jinst2');
            \draw[dotted,semithick] (Jinst2') to[bend right=20] 
                node[circle,fill=white,inner sep=4pt]{ $\mkern2mu\cdots$ } (Jinstk');
            \draw[dotted,semithick] (Jinstk') to[bend right=20] (botinst);
            
            \draw[|-stealth',semithick,shorten <= 2pt, shorten >= 2pt] (Jinst2') to[bend left=20] (Iinst);
        \end{tikzpicture}
        \caption{The above figure shows a sketch of the situation. On the left-hand side it shows the original PDB $\pdb[I]$ with its restriction to $\phi$ and the special instance $I_0$. On the right-hand side it shows the structure of the PDB $\pdb[J]_{\bot}$ we constructed. Every instance $J \sim \pdb[J]_{\bot}$ is essentially a tuple of $k$ instances followed by a flag. The instance $J$ represents $I_0$ if either all of its parts are taken from the darker part of the copies, or it contains the flag fact $f_{\bot}$. Otherwise, $J$ represents the instance corresponding to the first part residing in the lighter shaded part of the copies.}
        \label{fig:independent-copies-sketch}
    \end{figure}

    At this point, all that remains to show is that there exists an $\FO$-view $\Phi$ that maps each $J \sim \pdb[J]_{\bot}$ to the instance $I \in \II$ (with $I \models \phi$) it represents. If $\Phi$ has this property, then $\Phi( \pdb[J]_{\bot} ) = \pdb[I] \under \phi$, proving our claim. Suppose that the relations of $\pdb[I]$ are $R_1,\dots,R_N$. Let $\tilde\phi$ be any $\FO$-formula over the schema of $\pdb[I]$. We assume that $\tilde\phi[i]$ is rewritten so that it contains neither $\vee$, nor universal quantification (this can be done in an equivalence preserving and domain-independent way, cf. \cite[p.~82]{Abiteboul+1995}). Then, for all $i = 1,\dots,k$, the formula $\tilde\phi[i]$ is defined from $\tilde\phi$ by structural induction as follows:
    \[
        \tilde\phi[i] \coloneqq
        \begin{cases}
            R_n'( i, \tup u )                           & \text{if }\tilde\phi = R_n( \tup u )\\
            x = a                                       & \text{if }\tilde\phi = ( x = a )\\
            (x = y) \wedge \big( x \neq \bot \wedge \bigwedge_{ j = 1 }^{ k } x \neq j \big)
                                                        & \text{if }\tilde\phi = ( x = y )\\
            \neg\tilde\phi'[i] \wedge \bigwedge_{ x \in \free(\phi') } \big( x \neq \bot \wedge \bigwedge_{ j = 1 }^{ k } x \neq j \big)
                                                        & \text{if }\tilde\phi = \neg\tilde\phi'\\
            \tilde\phi_1[i] \wedge \tilde\phi_2[i]      & \text{if }\tilde\phi = \tilde\phi_1 \wedge \tilde\phi_2\\
            \exists x \with \tilde\phi'[i]              & \text{if }\tilde\phi = \exists x \with \tilde\phi'\text.
        \end{cases}
    \]
    That is, in essence we change the atoms $R_n( \tup u )$ into $R_n'( i, \tup u )$ and guard all free variables in negated subformulae and in equalities between variables. Then, for all $i = 1,\dots,k$ and all $J \sim \pdb[J]_{\bot}$ it holds that $\tilde\phi( J[i] ) = \tilde\phi[i]( J )$. This can easily be verified over the structure of $\tilde\phi$.
    
    Applying this construction to the $\FO$-sentence $\psi$ leaves us with $\FO$-sentences $\psi[i]$ over the schema of $\pdb[J]_{\bot}$ such that for all $J \sim \pdb[J]_{\bot}$ and all $i = 1,\dots,k$ it holds that
    \[
        J[i] \models \psi
            \Leftrightarrow
        J \models \psi[i]
        \text.
    \]
    Now for all $n = 1, \dots, N$ we define the query
    \[
        \Phi_n( \tup x ) \coloneqq 
            \Bigg\lbrack
                \bigg\lparen f_\bot \vee \bigwedge_{ i^* = 1 }^{ k } \neg \psi[i^*]\bigg\rparen 
                    \wedge 
                \bigg\lparen \bigvee_{R_n( \tup a ) \in I_0 } \tup x = \tup a \bigg\rparen
            \Bigg\rbrack
                \wedge 
            \Bigg\lbrack 
                \bigg\lparen 
                    \bigvee_{ i^* = 1 }^{ k } \neg f_\bot \wedge \psi[i^*] \wedge \bigwedge_{ i = 1 }^{ i^* - 1 } \neg\psi[i^*] 
                \bigg\rparen 
                    \wedge 
                \bigg\lparen 
                    R_n'( i^*, \tup x ) 
                \bigg\rparen 
            \Bigg\rbrack\text.
    \]
    Then for all $n = 1,\dots,N$ and all $J \sim \pdb[J]_{ \bot }$, it holds that $\Phi_n( J )$ is equal to the restriction of $I$ to the relation $R_n$ where $I$ is the instance from $\pdb[I]$ that is represented by $J$. By \cref{cla:J-bot-probs}, it follows that $\Phi( \pdb[J]_{\bot} ) = \pdb[I] \under \phi$ where $\Phi = \set{ \Phi_1,\dots,\Phi_N }$. That is, $\pdb[I] \under \phi \in \FO(\TI)$, concluding the proof of \cref{thm:conditional-views}.
\end{proof}

\begin{remark}\label{remark:fo-everywhere}
An implication of \Cref{thm:conditional-views} is that $\FOTI$ is closed under $\FO$-conditioning, regardless of whether the conditioning is done before or after applying the view. This holds since
$\FOTI \subseteq \FOTI\under\FO \subseteq \FOTIFO \subseteq \FOTI$.
To prove the middle containment, an $\FO$-sentence over the view's output schema can be translated to an $\FO$-sentence over the input schema by replacing every relation with the $\FO$-view defining it.
\end{remark}

\section{Power of Tuple-Independent Representations}\label{sec:power}

In this section, we establish some positive results regarding representability. In \cref{ssec:sizesandprobs}, we prove a sufficient criterion on the growth rate of the probabilities and conclude also that PDBs of bounded instances size are in $\FOTI$. \Cref{ssec:bid} is devoted to $\BID$-PDBs.

\subsection{A Condition on Sizes and Probabilities}\label{ssec:sizesandprobs}

Towards characterizing representability of PDBs using $\FO$-views over $\TI$-PDBs, so far we have from \cref{pro:finite-moments} that the finite moments property is a \emph{necessary} condition. In this subsection, we present a \emph{sufficient} condition for membership in $\FOTI$, taking the detour over membership in $\FOTIFO$. At the end of this subsection, we discuss the implications of this condition.
 
\begin{lemma}\label{lem:condition}
    Let $\pdb = (\DD,P)$ be a PDB. If there exists some constant $c \in \NN_+$ such that
    \begin{equation}\label{eq:nameless}
        \sum_{D \in \DD\setminus\set{\emptyset}} \size{D} \cdot P\big(\set{D}\big)^{\frac{c}{\size{D}}} < \infty\text,
        \tag{$\dagger$}
    \end{equation}
    then $\pdb \in \FO(\TI\under\FO)$.
\end{lemma}

Note that \eqref{eq:nameless} looks quite similar to the requirement that the expected instance size of $\pdb$ be finite. However, instead of weighting with the original probabilities of the possible worlds, in \eqref{eq:nameless} the probability of a world $D$ is raised to the power of $\frac{ c }{ \size{ D } }$. This makes \eqref{eq:nameless} more restrictive than the finite moments property.

The statement of the lemma was chosen to involve $\FOTIFO$ rather than the equivalent $\FOTI$, since in $\FOTIFO$ we can also use conditioning which facilitates finding a representation. Before starting with the proof, let us introduce the main idea. So suppose that $\pdb = ( \DD, P )$ is the probabilistic database we are given. We want to show that if there exists a constant $c$ such that \eqref{eq:nameless} holds, then $\pdb \in \FO(\TI\under\FO)$. For the moment, assume that $c = 1$, so that the condition \eqref{eq:nameless} becomes $\sum_{ D \in \DD \setminus \set{ \emptyset } } \size{D} \cdot P( \set{ D } )^{ 1 / \size{D} } < \infty$. We first want to construct a $\TI$-PDB that serves as a base for the $\FOTIFO$-representation. For this, whenever $D \in \DD$ and $f \in D$, then we create a new fact that is a copy of $f$, tagged with an identifier for the instance $D$. In particular, this creates a copy of $f$ for every instance in which it appears. We say that a subset of the constructed facts is a representation if there exists exactly one instance for which all tagged facts appear. We construct the $\FO$-condition in such a way that it filters out the possible worlds that are representations. Our $\FO$-view then reconstructs the original instance from the tagged facts. The probabilities of the new facts have to be tuned in such a way that every possible world of the representation has the correct probability. The facts we constructed, together with their probabilities, have to adhere to the convergence condition from \cref{thm:well-def-ti}. Considering the representation mechanism, we end up with the condition \eqref{eq:nameless} (for $c=1$). The condition can be relaxed to the general form with a constant $c$ by encoding $c$ original facts into a single new fact.

\begin{proof}[{Proof of \cref{lem:condition}}]

    Let $\pdb = (\DD,P)$ be a PDB with instances $\DD = \set[\big]{D_0,D_1, D_2, \dots}$ where $D_0 = \emptyset$ is the empty instance.
    Let $s_i \coloneqq \size{D_i}$ denote the size of $D_i$ and $p_i \coloneqq P\big(\set{D_i}\big)$ its probability in $\pdb$ for all $i \in \NN$. Without loss of generality,
    we assume that $p_i > 0$ for all $i \in \NN$. (If $p_0 = 0$, we just need to consider $i \in \NN_+$. The proof then works the same way without the special treatment of $i=0$.)
    For simplicity, we assume that the schema of $\pdb$ consists of a single $r$-ary relation symbol $R$. The proof can easily be generalized to arbitrary schemas.
    
    Let $c \in \NN_+$ be some fixed constant (we will determine the value of $c$ later). We now construct the $\TI$-PDB $\pdb[I]$ used in our representation. The schema of $\pdb[I]$ consists of a single relation symbol $\widehat{R}$ of arity $3 + c \cdot r$. We let $\widehat{ \UU } \coloneqq \UU \cup \set{ \bot }$ be the underlying universe of $\pdb[I]$, which is the universe of $\pdb$, (disjointly) augmented with a new dummy symbol $\bot$.
    
    For all $i \in \NN$, suppose that $D_i = \set[\big]{ R( \tup a_{i,1} ), \dots, R( \tup a_{i,s_i} ) }$ where $\tup a_{i,j} \in \UU^r$ for all $1 \leq j \leq s_i$. For $j > s_i$, we set $\tup a_{i,j} \coloneqq \big( \bot, \dots, \bot \big)$. We define
    \[
        \NEXT_{i,j} \coloneqq 
        \begin{cases}
            j+1     & \text{if } (j+1)\cdot c < s_i\text{ and}\\
            \bot    & \text{if } (j+1)\cdot c \geq s_i\text.
        \end{cases}
    \]
    Using this notation we define a set $F$ of facts over the new schema $\set[\big]{ \widehat{R} }$ and the new universe $\widehat{\UU}$ by letting
    \[
        F \coloneqq \set[\Big]{ 
            \widehat{R}\big( i, j, N_{i,j}, \tup a_{i,jc+1}, \dots, \tup a_{i,jc+c} \big)\with
            i,j \in \NN \text{ and } j \leq \max\big( \big\lceil \tfrac{s_i}{c} \big\rceil - 1 , 0 \big)
        }
    \]
    Essentially, we chop up every instance $D_i$ into segments containing $c$ facts each (where the last segment may contain $<c$ facts). Every individual fact of $F$ corresponds to such a segment of (up to) $c$ facts in a database instance of $\DD$. The structure of these facts is shown in \cref{fig:segmentation}. The first entry of such a fact is the index of the original instance in which the encoded facts appear (\emph{instance identifier}). The second entry contains the index of the segment (\emph{segment identifier}) whereas the third entry contains the index of the next segment (\emph{next segment identifier}). The $r\cdot c$ remaining entries contain the $c$ facts of the segment from the original instance. If there are not enough facts left to fill the $c$ fact slots, they are padded with dummy entries $(\bot, \dots, \bot)$.
    
    \begin{figure}[htb]
        \centering
        \begin{tikzpicture}
            \matrix [nodes={draw,thick,minimum size=7mm}, column sep=-0.8pt]%
            {
                \node[fill=ShadingLight] (i) {$i$}; &
                \node[fill=ShadingMedium] (j) {$j$}; & 
                \node[fill=ShadingMedium,inner ysep=0pt] (N) {$\NEXT_{i,j}$}; & 
                \node[fill=ShadingDark] (t) [minimum width=2cm]{$\tup a_{i,jc+1}$}; & 
                \node[fill=ShadingDark] (dots) [minimum width=2cm]{$\cdots$}; & 
                \node[fill=ShadingDark] (t') [minimum width=2cm]{ $\tup a_{i,jc+ c}$ }; \\
            };
            \draw[decoration={calligraphic brace,amplitude=5pt}, decorate, line width=1.25pt] ([yshift=-2pt]t'.south east) to node[below,yshift=-10pt] (segment-label) {segment of $\leq c$ facts from $D_i$} ([yshift=-2pt]t.south west);
            
            \node[above=26pt of dots.center,anchor=base] (nsp-label) {next segment identifier};
            \node[left=4.5cm of segment-label.base,anchor=base] (si-label) {segment identifier};
            \node[left=4.2cm of nsp-label.base,anchor=base] (ii-label) {instance identifier};

            \draw[-stealth',shorten >= 6pt, shorten <= 2pt] (nsp-label.west) to[bend right=5] (N.55);
            \draw[-stealth',shorten >= 6pt, shorten <= 2pt] (si-label.120) to[bend right=5] (j.-100);
            \draw[-stealth',shorten >= 6pt, shorten <= 2pt] (ii-label.-100) to[bend right=5] (i.80);
            \end{tikzpicture}
        \caption{Structure of the new facts, and meaning of their individual components.}
        \label{fig:segmentation}
    \end{figure}
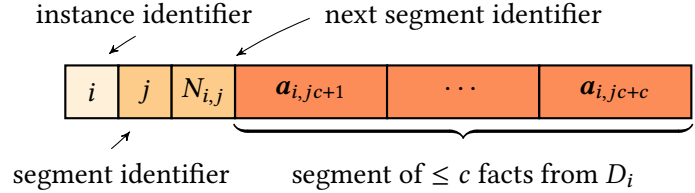
    
    We now define the $\TI$-PDB $\pdb[I] = \big(\II, P_{\pdb[I]} \big)$ such that $\facts(\pdb[I]) = F$, giving every fact from $F$ a positive probability (that we specify later) so that $\pdb[I]$ is well-defined. Let us first finalize the representation mechanism. For all $i \in \NN$ we let $\widehat{D}_i \in \II$ be the instance from $\pdb[I]$ that contains precisely the facts that have instance identifier $i$. We say that an instance $I \in \II$ represents $D_i \in \DD$ if \begin{enumerate*} \item it contains all the facts from $\widehat{D}_i$ (that is, $\widehat{D}_i \subseteq I$), and \item it does not contain all facts of some other $\widehat{D}_j$ (that is, $\widehat{D}_j \not\subseteq I$ for all $j \neq i$).\end{enumerate*} In general, we call an instance $I\in\II$ a \emph{representation}, if there exists some $i \in \NN$ such that $I$ represents $D_i$.
    
    Next, we define the marginal probabilities for $\pdb[I]$. For this, first note that for $i \in \NN$, the instances $\widehat{D}_i$ are pairwise disjoint as their facts have different instance identifiers. Thus, by the construction of $F$, for all $f \in F$, there exists a unique $i = i(f)$ such that $f \in \widehat{D}_i$. For the sizes $\widehat{s}_i = \size[\big]{ \widehat{D}_i }$ of these instances, it holds that $\widehat{s}_i = \big\lceil \frac{s_i}{c} \big\rceil \geq 1$ for all $i > 0$ and $\widehat{s}_0 = 1$. Recall that $p_i$ is the probability of $D_i$ in $\pdb$  for all $i \in \NN$. For every fact $f \in F$, we define the marginal probability $q_f$ of $f$ by setting
    \[
        q_f \coloneqq \Big( \frac{ p_{i(t)} }{ 1 + p_{i(f)} } \Big)^{ 1 / \widehat{s}_{i(f)} }\text.
    \]
    Recall that everything we have done so far depends on the constant $c$.
    
    \begin{claim}\label{cla:uglyconvergence}
        If $c$ fulfills the condition that
        \begin{equation}\label{eq:namelessugly}
            \sum_{ i = 1 }^{ \infty } \ceil[\big]{\frac{s_i}{c}} \cdot p_i^{ \frac{ 1 }{ \ceil{ s_i/c } } } < \infty
            \tag{$\ddagger$}
        \end{equation}
        then $\pdb[I]$ is a well-defined $\TI$-PDB.
    \end{claim}
    
    \begin{claimproof}
        Suppose that \eqref{eq:namelessugly} holds for $c$. It holds that
        \[
            \sum_{ f \in F } q_f 
                =
            \frac{ p_0 }{ 1 + p_0 } + \sum_{ i = 1 }^{ \infty }\Big( \frac{ p_i }{ 1 + p_i } \Big)^{ \frac{ 1 }{ \ceil{ s_i / c } } }
                \leq
            1 + \sum_{ i = 1 }^{ \infty } p_i^{ \frac{ 1 }{ \ceil{ s_i / c } } }
                \leq
            1 + \sum_{ i = 1 }^{ \infty } \ceil[\big]{ \frac{s_i}{c} } \cdot p_i^{ \frac{ 1 }{ \ceil{ s_i / c } } }
                <
            \infty\text.
        \]
        By \cref{thm:well-def-ti}, the $\TI$-PDB $\pdb[I]$ is well-defined.
    \end{claimproof}
    
    Although \eqref{eq:namelessugly} differs slightly from \eqref{eq:nameless}, we can work with \eqref{eq:namelessugly} instead of \eqref{eq:nameless} by the following claim.
    
    \begin{claim}\label{cla:namelessvsugly}
        The following statements are equivalent:
        \begin{enumerate}
            \item\label{itm:ceiling} There exists $c \in \NN_+$ such that \eqref{eq:namelessugly} holds.
            \item\label{itm:original} There exists $c \in \NN_+$ such that \eqref{eq:nameless} holds.
        \end{enumerate}
    \end{claim}
    
    \begin{claimproof}
        We start with the direction (\labelcref{itm:ceiling}) $\Rightarrow$ (\labelcref{itm:original}). Suppose that $c$ is a constant for which \eqref{eq:namelessugly} holds.  For all $D_i$ with $i > 0$ it holds that $\frac{ s_i }{ c } \leq \ceil[\big]{ \frac{ s_i }{ c } }$. In particular, $s_i \leq c \cdot \ceil[\big]{ \frac{ s_i }{c} }$ and $\frac{ c }{ s_i } \geq 1 / \ceil[\big]{ \frac{ s_i }{ c } }$. As $p_i \in [0,1]$, it follows from the latter inequality that 
        \(
            p_i^{ c / s_i } \leq p_i^{ 1 / \ceil{ s_i / c } }
            \text.
        \)
        Hence,
        \[
            \sum_{ i = 1 }^{ \infty } s_i \cdot p_i^{ \frac{ c }{ s_i } }
                \leq 
            c \cdot \sum_{ i = 1 }^{ \infty } 
                \ceil[\Big]{ \frac{ s_i }{ c } } \cdot p_i^{ \frac{1}{ \ceil{ s_i / c} } } < \infty\text,
        \]
        so \eqref{eq:nameless} also holds for $c$.
        
        For the other direction, suppose \eqref{eq:nameless} holds for the constant $c$. For all $D_i$ with $s_i > 2c$ we have
        \[
            \ceil[\big]{ \frac{ s_i }{2c} } 
                < 
            \frac{ s_i }{2c} + 1 
                =
            \frac{2c+s_i}{2c}
                <
            \frac{2s_i}{2c}
                =
            \frac{s_i}{c}\text.
        \]
        Thus, $1 / \ceil[\big]{ \frac{ s_i}{ 2c } } \geq \frac{c}{ s_i }$ and therefore $p_i^{1 / \ceil{ s_i / 2c } } \leq p_i^{ c / s_i }$ whenever $s_i > 2c$. For the instances $D_i$ with $0 < s_i \leq 2c$, it holds that $\ceil[\big]{ \frac{ s_i }{ 2c }} = 1$ which yields
        \[
            \sum_{ i = 1 }^{ \infty } \ceil[\big]{ \frac{ s_i }{ 2c } } \cdot p_i^{ \frac{1}{ \ceil{ s_i / 2c } } }
                \leq 
            \sum_{ \substack{ i > 0\with\\s_i \leq 2c } } p_i
                +
            \sum_{ \substack{ i > 0\with\\s_i > 2c } } \tfrac{ s_i }{ c } p_i^{ \frac{c}{ s_i } }
                \leq
            1 + \tfrac{ 1 }{ c } \cdot \sum_{ i = 1 }^{ \infty } s_i \cdot p_i^{ \frac{c}{ s_i } }
                < \infty\text,
        \]
        so \eqref{eq:namelessugly} holds for $2c$.
    \end{claimproof}

    Given \cref{cla:namelessvsugly}, we fix the value of $c$ such that \eqref{eq:namelessugly} holds. Recall from \cref{cla:uglyconvergence} that $\pdb[I]$ is then well-defined. 
    
    Let $q_i \coloneqq \frac{p_i}{1+p_i}$ and observe that for all $i > 0$ we have
    \[
        \Pr_{I \sim \pdb[I]}\big( \widehat{D}_i \subseteq I \big) 
            = 
        \prod_{f \in \widehat{D}_i} q_f 
            = 
        \big( q_i^{ 1 / \widehat{s}_i } \big)^{ \widehat{s}_i } = q_i
    \]
    as well as 
    \[
        \Pr_{I \sim \pdb[I]}\big( \widehat{D}_0 \subseteq I \big) = q_{f_0} = q_0
    \]
    where $f_0$ is the unique fact in $\widehat{D}_0$. Moreover, note that $0 < q_i < 1$ for all $i \in \NN$. Since the sum over all $q_i$ is finite because of $\sum_{i=0}^{\infty} q_i \leq \sum_{i=0}^{\infty} p_i < \infty$, it holds that
    \[
        Z \coloneqq \prod_{i = 0}^{\infty} (1-q_i) \in (0,1]\text.
    \]
    Then for all $i \in \NN$ we have
    \begin{align*}
        \Pr_{I \sim \pdb[I]}\big( I \text{ represents } D_i ) 
            &= 
        q_i \cdot \prod_{ j \neq i }( 1 - q_j ) 
            = 
        Z \cdot \frac{q_i}{1-q_i} 
            =
        Z \cdot p_i\text,
    \shortintertext{which yields}
        \Pr_{I \sim \pdb[I]}\big( I \text{ is a representation}) 
            &=
        \sum_{i = 0}^{\infty} Z\cdot p_i 
            = 
            Z\text,
    \end{align*}
    and, moreover,
    \begin{equation}\label{eq:namelessrep}
        \Pr_{I \sim \pdb[I]}\big( I \text{ represents } D_i \under I \text{ is a representation})
            = 
        \frac{\Pr_{I \sim \pdb[I]}\big( I \text{ represents } D_i \big)}{\Pr_{I \sim \pdb[I]}\big( I \text{ is a representation}\big)} = p_i \text.
    \end{equation}
    
    In order to establish $\pdb \in \FOTIFO$, it now suffices to prove the following claim.
    
    \begin{claim}\label{cla:namelessviews}
        \begin{enumerate}
            \item There exists an $\FO$-sentence $\phi$ such that for all $I \in \II$ it holds that $I \models \phi$ if and only if $I$ is a representation.
            \item There exists an $\FO$-view $\Phi$ that, for all $I \in \II$ that are representations, maps $I$ to the instance of $\DD$ it represents.
        \end{enumerate}
    \end{claim}
    
    \begin{claimproof}
        \begin{enumerate}
            \item In order to check whether $I$ is a representation, we need to check whether there exists some $i$ such that $I$ represents $i$. Now $I$ represents $i$ if and only if $I$ contains all facts of $\widehat{D}_i$ but does not contain all facts of $\widehat{D}_j$ for all $j \neq i$. One can check if $\widehat{D}_i \subseteq I$ (and thus if $\widehat{D}_j \not\subseteq I$ for $j\neq i$) by checking the following:
            \begin{itemize}
                \item $I$ needs to contain a fact starting with instance identifier $i$ and segment identifier $0$; and
                \item whenever $I$ contains a fact with instance identifier $i$, segment identifier $j$ and next segment identifier $j' \neq \bot$, then $I$ contains a fact with instance identifier $i$ and segment identifier $j'$ ($j'$ will be $j+1$ by the definition of $F$).
            \end{itemize}
            \item Recall that the facts of $F$ contain up to $c$ original facts in succession. These can be recovered by a union of $c$ projections, under omitting the $\bot$ entries.\qedhere
        \end{enumerate}
    \end{claimproof}
    
    Letting $\Phi$ and $\phi$ be as in \cref{cla:namelessviews}, the equality \cref{eq:namelessrep} implies that
    \[
        \Pr_{I \sim \pdb[I]}\big( \Phi(I) = D_i \under I \models \phi \big) = p_i = P\big(\set{D_i}\big)
    \]
    for all $i \in \NN$ and hence, $\pdb \in \FOTIFO$. This concludes the proof of \cref{lem:condition}.
\end{proof}

As already mentioned, due to \cref{thm:conditional-views} we directly obtain the following from \cref{lem:condition}.

\begin{theorem}\label{thm:condition}
    Let $\pdb = (\DD,P)$ be a PDB. If there exists $c \in \NN_+$ such that \eqref{eq:nameless} holds, then $\pdb \in \FOTI$.
\end{theorem}

\Cref{thm:condition} easily yields a representability result for the class of PDBs of bounded instance size: We say that a PDB is of \emph{bounded instance size} if there exists some fixed bound $c$ such that all possible worlds have size at most $c$. (Note that carrying out the construction from the proof of \cref{lem:condition}, every instance of the size-bounded PDB can be encoded by a single fact.)

\begin{corollary}\label{cor:bounded-size}
    Every PDB of bounded instance size is in $\FOTI$.
\end{corollary}

\begin{proof}
    Let $\pdb = (\DD, P)$ be a PDB and take $c$ to be the bound on its instance sizes. Then, by \cref{thm:condition}, $\pdb \in \FOTI$ since
    \[
        \sum_{D \in \DD \setminus \set{\emptyset}} \size{D} \cdot P\big(\set{D}\big)^\frac{c}{\size{D}}
            \leq 
        c \cdot \sum_{D \in \DD} P\big(\set{D}\big) 
            = 
        c 
            < 
        \infty\text.\qedhere
    \]
\end{proof}

\begin{remark}
    PDBs of bounded instance size are not to be confused with finite PDBs. In particular, they can have infinite domains. As an example, consider the PDB $\pdb = (\DD,P)$ over a schema consisting of a single unary relation symbol $R$ with $\DD = \set{D_1,D_2,\dots}$, where $D_n$ contains a single fact $R(n)$ and $P\left(\set{D_n}\right) =\frac{6}{n^2\pi^2}$. In this example, the instance size is bounded by $1$ although the domain and, in particular, the number of possible worlds is infinite.
\end{remark}

The following example shows that the condition from \cref{thm:condition} can also be applicable to some PDBs with unbounded instance size.

\begin{example}
    Let $x \coloneqq \sum_{i=1}^{\infty} 2^{-i^2}$. Then $0 < x < \sum_{i=1}^{\infty} 2^{-i} = 1$. Now define $\pdb = (\DD,P)$ with $\DD = \set{ D_1, D_2, \dots }$ with $\size{D_i} = i$ and let $P\big(\set{D_i}\big) = \frac{1}{x}\cdot 2^{-i^2}$. Since $\sum_{D\in\DD} P\big(\set{D}\big) = 1$, $\pdb$ is a PDB. Note that $\big(\frac{1}{x}\big)^{\alpha} \leq \frac{1}{x}$ for all $\alpha \in (0,1]$ because $\frac{1}{x} > 1$. Then
    \begin{equation*}
        \sum_{D \in \DD} \size{D} \cdot P\big(\set{D}\big)^{\frac{1}{\size{D}}}
        = \sum_{i=1}^{\infty} i \cdot \big( \tfrac{1}{x} \cdot 2^{-i^2} \big)^{\frac{1}{i}}
        \leq \tfrac{1}{x} \sum_{i=1}^{\infty} i \cdot 2^{-i} = \tfrac{2}{x} < \infty\text.
    \end{equation*}
    Therefore, $\pdb$ satisfies the condition of \cref{thm:condition} for $c = 1$, so $\pdb \in \FOTI$. Yet, $\pdb$ is of unbounded instance size.
\end{example}
Unfortunately, we have no full characterization of $\FOTI$ yet. We know that the condition from \cref{thm:condition} implies membership in $\FOTI$ and that membership in $\FOTI$ implies the finite moments property (\cref{pro:finite-moments-ti}). As we have seen in \cref{ex:no-rep-fin-moments}, the finite moments property does not imply membership in $\FOTI$. Our next example shows that membership in $\FOTI$ does not imply that the condition from \cref{thm:condition} is satisfied. That is, the converse of \cref{thm:condition} does not hold. In fact, even though every $\TI$-PDB is trivially in $\FOTI$, some $\TI$-PDBs violate the condition of the theorem:

\begin{example}\label{exa:ti-not-cond}
    Consider the $\TI$-PDB $\pdb[I] = (\II,P)$ with fact set $\facts( \pdb[I] ) = \set[\big]{ R(i) \with i \in \NN_+ }$ and marginal probabilities 
    \[
        p_i = \Pr_{ I \sim \pdb[I] }\big( R(i) \in I \big) = \frac{1}{i^2+1} \in (0,1)\text. 
    \]
    First, note that $\pdb[I]$ is indeed well-defined as $\sum_{ i = 1 }^{ \infty } \frac{1}{i^2+1} < \infty$. Let $Z \coloneqq \prod_{ i = 1 }^{ \infty } \big( 1 - p_i \big)$. Since $\sum_{ i = 1 }^{ \infty } p_i < \infty$, it holds that $0 < Z < 1$. By the definition of tuple-independence, for all $I \in \II$ it holds that
    \[
        P\big( \set{ I } \big) 
            =
        \prod_{ R(i) \in I } p_i \cdot \prod_{ R(i) \notin I } \big( 1 - p_i \big) 
            =
        Z \cdot \prod_{ R(i) \in I } \frac{ p_i }{ 1 - p_i }\text.
    \]
    Let $I \neq \emptyset$. Recall that the geometric mean of a finite set of numbers is at least the minimum of the numbers. Thus,
    \[
        P\big( \set{ I } \big)^{ \frac{1}{\size{I} } }
            =
        Z^{ \frac{1}{\size{I}} } \cdot \bigg( \prod_{ R(i) \in I } \frac{ p_i }{ 1 - p_i } \bigg)^{ \frac{1}{\size{I}} }
            \geq
        Z^{ \frac{1}{\size{I}} } \cdot \min_{ R(i) \in I } \frac{ p_i }{ 1 - p_i }
    \]
    As our sequence $(p_i)_{ i \in \NN_+ }$ is monotonically decreasing, it holds that $\min_{ R(i) \in I } p_i = p_{i(I)}$ where $i(I) = \max\set[\big]{ i \with R(i) \in I }$. Because the function $x \mapsto \frac{x}{1-x}$ is monotonically increasing on $(0,1)$, we get
    \[
        \min_{ R(i) \in I }\frac{ p_i }{ 1-p_i } 
            =
        \frac{ \min_{ R(i) \in I} p_i }{ 1 - \min_{ R(i) \in I } p_i }
            =
        \frac{ p_{i(I)} }{ 1 - p_{i(I)} }\text.
    \]
    Using $Z^{ \frac{ 1 }{ \size{I} } } \geq \min\set{1,Z}$, we conclude
    \[
        P\big( \set{I} \big)^{\frac{1}{\size{I}}} 
            \geq 
        Z^{\frac{1}{\size{I}}} \min_{R(i) \in I} \frac{p_i}{1 - p_i}
            \geq
        \min\set{1,Z} \cdot \frac{p_{i(I)}}{1-p_{i(I)}}.
    \]
    Partitioning the instances according to their value of $i(I)$, we get the following for all $c \in \NN_+$:
    \begin{align*}
        \sum_{ I \neq \emptyset } \size{ I } \cdot P\big( \set{ I } \big)^{ \frac{ c }{ \size{ I } } }
            &\geq
        \sum_{ I \neq \emptyset } P\big( \set{ I } \big)^{ \frac{c}{\size{I}} }\\
            &=
        \sum_{ i = 1 }^{ \infty } \mkern4mu \sum_{ I \with\mkern-2mu i(I) = i } \Big( P\big( \set{ I } \big)^{ \frac{ 1}{\size{I}}} \Big)^c\\
            &\geq
        \sum_{ i = 1 }^{ \infty } \mkern4mu \sum_{ I \with\mkern-2mu i(I) = i } \Big( \min\set{1,Z} \cdot \frac{ p_i }{ 1-p_i } \Big)^c\\
            &=
        \min\set{1,Z}^c \cdot \sum_{ i = 1 }^{ \infty } \Big( \frac{ p_i }{ 1 - p_i } \Big)^{ c } \cdot 2^{i-1}\text,
    \end{align*}
    because $\size[\big]{ I \in \II \with i(I) = i } = 2^{i-1}$. Note that $\frac{ p_i }{ 1-p_i } = \frac{1}{i^2}$, since $p_i = \frac{1}{i^2+1}$. Thus,
    \[
        \sum_{ I \neq \emptyset } \size{ I } \cdot P\big( \set{ I } \big)^{ \frac{ c }{ \size{ I } } }
            \geq
        \min\set{ 1, Z }^c \cdot \sum_{ i = 1 }^{ \infty } \Big( \frac{ p_i }{ 1 - p_i } \Big)^c \cdot 2^{i-1}
            =
        \frac{ \min\set{ 1, Z }^c }{ 2 } \cdot \sum_{ i = 1 }^{ \infty } \frac{ 2^i }{ i^{2c} }
            = \infty\text,
    \]
    since $2^i$ grows faster than any polynomial. This shows that for $\pdb[I]$ there exists no $c \in \NN_+$ for which \eqref{eq:nameless} holds, so \cref{thm:condition} is not applicable. Nevertheless, $\pdb[I]$ is in $\TI$, so in particular it is in $\FOTI$.
\end{example}

Thus, we still exhibit a proper gap between our conditions for containment in $\FOTI$.

\subsection{Block-Independent Disjoint Databases}\label{ssec:bid}

In this section, we prove that every $\BID$-PDB can be represented as an $\FO$ view over a $\TI$-PDB. First, note that this does not follow from the previous section, as there are $\BID$-PDBs that are not naturally tuple-independent and such that \cref{thm:condition} does not apply to them.

\begin{example}\label{exa:bid-not-cond}
    Consider the $\BID$-PDB $\pdb[I] = (\II,P)$ with blocks $B_1, B_2, \dots$ such that for all $i$, the block $B_i$ contains exactly the two facts $R_1(i)$ and $R_2(i)$, both with marginal probability $p_i = \frac{ 1 }{ 2(i^2+1) }$. As $\pdb[I]$ is $\BID$-PDB, it is also contained in $\FO( \BID )$. However, we now show that $\pdb[I]$ does not satisfy the condition from \cref{thm:condition}.
    
    Let $\pdb[J]=(\JJ,P')$ be the $\TI$-PDB from \cref{exa:ti-not-cond}. Recall that the facts of $\pdb[J]$ were of the shape $R(i)$ with $i \in \NN_+$. We say an instance $J$ of $\pdb[J]$ and an instance $I$ of $\pdb[I]$ are \emph{similar}, denoted $J \simeq I$ if for all $i \in \NN$ it holds that
    \[
        R(i) \in J
        \quad\Leftrightarrow\quad
        I \cap B_i = I \cap \set[\big]{ R_1(i), R_2(i) } \neq \emptyset\text.
    \]
    For every instance $J \in \JJ$, there are exactly $2^{ \size{ J } }$ instances $I \in \II$ with $J \simeq I$. Each such instance $I$ has probability $P\big( \set{ I } \big) = 2^{ - \size{ J } }P'\big( \set{ J } \big)$. Then it holds that
    \begin{align*}
        \sum_{ I \neq \emptyset } \size{ I } \cdot P\big( \set{ I } \big)^{ \frac{ c }{ \size{ I } } }
            &=
        \sum_{ J \neq \emptyset } \sum_{ I \with J \simeq I } \size{ J } \cdot P'\big( \set{ J } \big)^{ \frac{ c }{ \size{ J } } }\\
            &=
        \sum_{ J \neq \emptyset } \sum_{ I \with J \simeq I } \size{ J } \cdot \Big( 2^{ -\size{J} } \cdot P'\big( \set{ J } \big) \Big)^{ \frac{ c }{ \size{ J } } }\\
            &=
        \sum_{ J \neq \emptyset } 2^{ \size{ J } } \cdot \size{ J } \cdot 2^{ -c } \cdot P'\big( \set{ J } \big)^{ \frac{ c }{ \size{ J } } }\\
            &\geq
        2^{-c} \cdot \sum_{ J \neq \emptyset } \size{ J } \cdot P'\big( \set{ J } \big)^{ \frac{ c }{ \size{ J } } }\text.
        \end{align*}
    It follows from \cref{exa:ti-not-cond} that this sum diverges for all $c \in \NN_+$. Therefore, $\pdb[I]$ is a $\BID$-PDB for which there exists no $c \in \NN_+$ such that \eqref{eq:nameless} holds.
\end{example}

Again, we take the detour via conditional representations and show that every $\BID$-PDB can be represented as an $\FO$-view of an $\FO$-conditioned $\TI$-PDB.

\begin{lemma}\label{lem:bid}
     $\BID \subseteq \FO(\TI\under\FO)$.
\end{lemma}

The basic idea of the proof is as follows. We start with a $\BID$-PDB and forget about its block structure. To compensate, every fact is equipped with an identifier indicating to which block it belongs. The resulting PDB is then treated as a $\TI$-PDB. We define the condition to reject the instances that violate the intended block structure (using the block identifiers to find violations), and then the view can simply project out the block identifiers. The marginal probabilities are carefully chosen to guarantee that this process results in the same probability distribution as in the original PDB. 

\begin{proof}
    Let $\pdb[I] = (\II, P_{\pdb[I]})$ be a $\BID$-PDB with blocks $\set{B_1, B_2, \dots}$ such that block $B_i$ consists of the facts $\set{f_{i,1},f_{i,2},\dots}$. To simplify notation, assume that $B_i$ is countably infinite for all $i \in \NN_+$. This is without loss of generality, since we can add infinitely many artificial facts of marginal probability $0$ to every block of $\pdb$. For all $i,j \in \NN_+$, let $p_{i,j} \coloneqq \Pr_{ I \sim \pdb[I] }( f_{i,j}\in I)$. By \cref{thm:well-def-bid}, $\sum_{j=1}^{\infty}\sum_{i=1}^{\infty} p_{i,j}  < \infty$. 
    
    The probability to choose no fact from block $B_i$ is $r_i = 1-\sum_{j=1}^{\infty} p_{i,j}$, called the \emph{residual} (probability mass) in $B_i$. for all $\epsilon \in (0,1)$, there are only finitely many residuals $r_i$ with $r_i < \epsilon$ (otherwise, the sum over all fact probabilities in $\pdb[I]$ would diverge in contradiction to \cref{thm:well-def-bid}). We assume that the blocks of $\pdb[I]$ are indexed in increasing order of the $r_i$. That is, $r_i \leq r_j$ whenever $j \geq i$. Let $m$ be the non-negative integer with $r_i = 0$ if and only if $1\leq i \leq m$.
    
    We construct a conditional, tuple-independent representation of $\pdb[I]$ by altering every relation in the schema of $\pdb[I]$ to contain an additional \emph{block identifier} attribute. The facts of our new $\TI$-PDB $\pdb[J] = (\JJ, P_{\J})$ are the facts from $\pdb[I]$, augmented by the number of their block in the additional attribute. That is, for every fact $f = R(\tup a)$ from a block $B_i$, we construct a fact $\tilde f = R'( \tup a, i )$. Note that for $f \neq f'$ we have $\tilde f \neq \tilde f'$. We let $\widetilde{B}_i = \set{ \tilde f_{i,j} \with j \in \NN_+ }$. Then adding, respectively removing the block identifier $i$ in the facts yields a one-to-one correspondence between $B_i$ and $\widetilde{B}_i$.
    
    The marginal probabilities $q_{i,j} = \Pr_{ J \sim \pdb[J] }( f_{i,j} \in J)$ of $\tilde f_{i,j}$ in $\pdb[J]$ are defined in the following way:
    \[
        q_{i,j} \coloneqq
        \begin{dcases}
            \frac{p_{i,j}}{1+p_{i,j}}   & \text{if }r_i = 0\text{ and}\\
            \frac{p_{i,j}}{r_i+p_{i,j}} & \text{if }r_i > 0\text.
        \end{dcases}
    \]
    We show that these marginal probabilities indeed span a well-defined $\TI$-PDB. Consider $\tilde f_{i,j}$ with $i$ and $j$ arbitrary. Then,
    \begin{itemize}
        \item if $r_i > 0$, it holds that 
            \(\displaystyle q_{i,j} =  \frac{p_{i,j}}{r_i+p_{i,j}} \leq  \frac{p_{i,j}}{r_{i}} \leq  \frac{p_{i,j}}{r_{m+1}} \), and
        \item if $r_i = 0$, it holds that
            \(\displaystyle q_{i,j} =  \frac{p_{i,j}}{1+p_{i,j}} \leq  \frac{p_{i,j}}{1} \leq  \frac{p_{i,j}}{r_{m+1}} \).
    \end{itemize}
    Since
    \[
        \sum_{ i = 1 }^{ \infty } \sum_{ j = 1 }^{ \infty }{ q_{i,j} } 
            \leq 
        \sum_{ i = 1 }^{ \infty } \sum_{ j = 1 }^{ \infty }{\frac{p_{i,j}}{r_{m+1}}}
            =
        \frac{1}{r_{m+1}} \cdot \sum_{ i = 1 }^{ \infty } \sum_{ j = 1 }^{ \infty } p_{i,j}
            < 
        \infty
        \text,
    \]
    it follows that $\pdb[J]$ is well-defined. The instances of interest in $\pdb[J]$ are those that obey the block structure of $\pdb[I]$. The following claim asserts that this property is $\FO$-definable.
    
    \begin{claim}
        There exists an $\FO$-sentence $\phi$ such that for all $J \in \JJ$ it holds that $J \models \phi$ if and only if \begin{enumerate*} \item $J$ contains at most one fact with block identifier $i$ for all $i \in \NN_+$; and \item $J$ contains exactly one fact with block identifier $i$ for all $i \leq m$. \end{enumerate*}
    \end{claim}
    
    \begin{claimproof}
        For simplicity, assume that the schema of $\pdb[I]$ consists of a single, unary relation symbol $R$. Let
        \[
            \phi \coloneqq 
            \big( \forall\mkern-1mu j\mkern2mu \exists^{\leq 1} x \colon R'(x,j) \big) \wedge\mkern2mu
            \bigwedge_{ i = 1 }^{ m } \exists^{ = 1 }x \with R'(x,i)
        \]
        where $R'$ is the augmented version of $R$. Note that the $\exists^{\leq 1}$ and $\exists^{=1}$ quantifiers are expressible in plain $\FO$. The formula above can easily be generalized to arbitrary schemas.
    \end{claimproof}
    
    Let $\Phi$ be the view that projects the block identifier out of the facts of $\JJ$. The PDB $\pdb[J]$ together with the condition $\phi$ and the $\FO$-view $\Phi$ is our representation of $\pdb[I] = (\II, P_{\pdb[I]})$. It is left to show that the probability distribution we obtain this way is the same as in the original PDB. That is, we need to show that for all $I \in \II$ it holds that
    \[
    	\Pr_{ J \sim \pdb[J]} \big( \Phi(J) = I \under J \models \phi\big)
    	    =
    	P_{\pdb[I]}\big(\set{I}\big)\text.
    \]
    
    We first show that it suffices to check every block independently. Given an instance $I$ of $\pdb[I]$, we denote by $\rest{I}{i}$ the restriction of $I$ to $B_i$. For the instances $J$ of $\pdb[J]$, we similarly let $\rest{J}{i}$ denote the restriction of $J$ to $\widetilde{B}_i$. Since the blocks are independent in $\pdb[I]$, for all $I \in \II$ it holds that
    \[
        P_{\pdb[I]}\big(\set{I}\big) 
            = 
        \Pr_{I' \sim \pdb[I]} \big(\rest{I'}{i} = \rest{I}{i}\text{ for all }i \in \NN_+\big)
            = 
        \prod_{ i = 1 }^{ \infty } \Pr_{I' \sim \pdb[I]}\big( \rest{I'}{i} = \rest{I}{i} \big)
        \text.
    \]
    Let us now inspect the conditional representation. We denote by $\phi_{i}$ the condition that if $r_i > 0$, there is at most one element of $B_i$, and if $r_i = 0$, there is exactly one element of $B_i$. It holds that $J \sim \pdb[J]$ satisfies $\phi$ if and only if $J \models \phi_i$ for all $i \in \NN_+$.\footnote{Although $\phi_i$ is expressible as an $\FO$-sentence, this is not required here. We use the name $\phi_i$ and the notation $\models$ for convenience.}
    
    Let $I \in \II$ be an instance with positive probability in $\pdb[I]$. Note that for all $J \in \JJ$ with $J \not\models \phi$, the instance $\Phi(J)$ can not be an instance of positive probability in $\pdb[I]$ since it either contains two distinct facts from the same block, or does not contain a fact from one of the blocks with residual $0$. Thus, if $\Phi( J ) = I$ for any instance $I$ of positive probability in $\pdb[I]$, then $J \models \phi$. Since the facts and therefore also the blocks are independent in $\pdb[J]$, it holds that
    \begin{align*}
    	\Pr_{ J \sim \pdb[J] }\big(\Phi(J) = I \under J \models \phi\big)
    	    &=
    	\frac{\Pr_{ J \sim \pdb[J] }( \Phi(J) = I \text{ and } J \models \phi )}{\Pr_{J \sim \pdb[J]}(J \models \phi)}\\
            &=
        \frac{\Pr_{ J \sim \pdb[J] }(\rest{\Phi(J)}{i} = \rest{I}{i} \text{ for all }i \in \NN_+)}{\Pr_{J \sim \pdb[J]}(\rest{J}{i} \models \phi_{i} \text{ for all }i \in \NN_+)}\\
            &=
        \frac{ 
            \prod_{ i = 1 }^{ \infty }
            \Pr_{J \sim \pdb[J]}(\rest{\Phi(J)}{i} = \rest{I}{i})
        }{
            \prod_{ i = 1 }^{ \infty }
            \Pr_{J \sim \pdb[J]}(\rest{J}{i} \models \phi_{i})
        }
        \text.
    \end{align*}
    Therefore, in order to show that the probability distribution over the conditional representation is the same as the original PDB, it is enough to show that for all $I \in \II$ and every block $B_i$ it holds that
    \begin{equation}\label{eq:bid-probs}
        \frac{
            \Pr_{J \sim \pdb[J]}( \rest{\Phi(J)}{i} = \rest{I}{i} )
        }{
            \Pr_{J \sim \pdb[J]}( \rest{J}{i} \models \phi_{i} )
        } 
            =
        \Pr_{I' \sim \pdb[I]}(\rest{I'}{i} = \rest{I}{i})\text.
    \end{equation}
    
    Thus, fix an arbitrary instance $I \in \II$ with positive probability in $\pdb[I]$ and let $i \geq 1$. We denote $Z_i \coloneqq \prod_{ f_{i,j} \in B_i }\big( 1-q_{i,j}\big)$. Recall that $\widetilde{B}_i$ is identical to $B_i$ apart from the block identifier $i$ that has been appended to the facts. Note that in $\pdb[J]$, the probability of selecting only $\tilde f_{i,j}$ among the facts of $\widetilde{B}_i$ is 
    \[
        q_{ i,j } \cdot 
        \prod_{ \substack{ f_{i,k}\in B_i \\ k\neq j } }
        \big( 1-q_{i,k} \big)
            =
        \frac{q_{i,j}}{1-q_{i,j}}{Z_i}\text.
    \]
    We distinguish cases according to whether $B_i$ has a positive residual and, if so, whether $I$ contains a fact from $B_i$.
    
    \begin{enumerate}
        \item If $r_i=0$, then the probability of having only fact $\tilde f_{i,j}$ from $\widetilde{B}_i$ is
        \[
            \frac{q_{i,j}}{1-q_{i,j}} \cdot Z_i
                =
            \frac{p_{i,j}}{1+p_{i,j}} \cdot \Big(1-\frac{p_{i,j}}{1+p_{i,j}}\Big)^{-1} \cdot Z_i
                =
            p_{i,j}\cdot Z_i
            \text.
        \]
        Suppose that $f_{i,k}$ is the unique fact from $B_i$ in $I$. Then
        \[
            \frac{
                \Pr_{J \sim \pdb[J]}(\rest{\Phi(J)}{i} = \rest{I}{i})
            }{ 
                \Pr_{J \sim \pdb[J]}(\rest{J}{i} \models \phi_{i})
            } 
                =
            \frac{
                p_{i,k} \cdot Z_i
            }{ 
                \sum_{ j = 1 }^{ \infty }p_{i,j} \cdot Z_i
            }
                =
            \frac{p_{i,k}}{\sum_{j = 1}^{ \infty } p_{i,j}}
                =
            p_{i,k}
                =
            \Pr_{I' \sim \pdb[I]}\big(\rest{I'}{i} = \rest{I}{i}\big)\text.
        \]
        \item If $r_i > 0$, then the probability of having no fact from $\widetilde{B}_i$ in $\pdb[J]$ is $Z_i$, and the probability of having only fact $\tilde f_{i,j}$ from $\widetilde{B}_i$ is
        \[
            \frac{ q_{i,j} }{ 1-q_{i,j} } \cdot Z_i 
                =
            \frac{ p_{i,j} }{ r_i + p_{i,j} } \cdot \Big( 1 - \frac{p_{i,j}}{r_i+p_{i,j}} \Big)^{-1} \cdot Z_i 
                = 
            \frac{p_{i,j}}{r_i}\cdot Z_i\text.
        \]
        Therefore, the probability of satisfying the condition $\phi_i$ is:
        \[
            \Pr_{J \sim \pdb[J]}\big(\rest{J}{i} \models \phi_{i}\big)
                =
            Z_i + \sum_{ j = 1 }^{ \infty } \frac{ p_{i,j} }{ r_i } \cdot Z_i
                = 
            \frac{Z_i}{r_i} \cdot \Big(r_i+\sum_{j = 1}^{\infty} p_{i,j}\Big) 
                = 
            \frac{Z_i}{r_i}\text.
        \]
        We distinguish two subcases.
        \begin{enumerate}
            \item In case there exists a fact $f_{i,k}$ from $B_i$ in $I$, it holds that
                \[
                    \frac{\Pr_{J \sim \pdb[J]}(\rest{\Phi(J)}{i} = \rest{I}{i})}{\Pr_{J \sim \pdb[J]}(\rest{J}{i} \models \phi_{i})} 
                    =
                    \frac{\frac{p_{i,k}}{r_i}\cdot Z_i}{\frac{Z_i}{r_i}}
                    = 
                    p_{i,k} 
                    = 
                    \Pr_{I' \sim \pdb[I]}\big(\rest{I'}{i} = \rest{I}{i}\big)\text.
                \]
            \item In case no fact from $B_i$ appears in $I$, it holds that
                \[
                    \frac{\Pr_{J \sim \pdb[J]}(\rest{\Phi(J)}{i} = \rest{I}{i})}{\Pr_{J \sim \pdb[J]}(\rest{J}{i} \models \phi_{i})} 
                    =
                    \frac{Z_i}{\frac{Z_i}{r_i}}
                    =
                    r_i 
                    =
                    \Pr_{I' \sim \pdb[I]}\big(\rest{I'}{i} = \rest{I}{i}\big)\text.
                \]
            \end{enumerate}
        \end{enumerate}
        Since $i$ is arbitrary, this holds for all $i$. Together, we have verified \cref{eq:bid-probs}. Thus, $\pdb[I]$ is an $\FO$-view of an $\FO$-conditioned $\TI$-PDB, witnessed by $\pdb[J]$, $\Phi$ and $\phi$, concluding the proof of \cref{lem:bid}.
\end{proof}

From \cref{lem:bid} and \cref{thm:conditional-views} we get the desired representability result for $\BID$-PDBs.

\begin{theorem}\label{thm:bid}
    $\BID \subseteq \FOTI$.
\end{theorem}

\begin{remark}\label{rem:fobid}
    Since 
    \[
        \FO(\BID) \subseteq \FO(\FO(\TI)) = \FO(\TI) \subseteq \FO(\BID)\text, 
    \]
    this  yields $\FO(\BID) = \FO(\TI)$. This also entails that $\FO(\BID)$ (and $\BID$ in particular) inherits the finite moments property from $\FOTI$.

    Similarly to \cref{remark:fo-everywhere}, $\FO(\BID)$ is also closed under $\FO$-conditioning, regardless of whether the conditioning is done before or after applying the view. So in particular it holds that 
    \[
        \FO(\BID\under\FO)=\FO(\BID)\under\FO=\FO(\BID)\text.
    \]
    The first equality above stems from the fact that \[\FO(\BID) \subseteq \FO(\BID\under\FO) \subseteq \FO\big( \FO(\TI)\under\FO \big) = \FOTI = \FO(\BID)\text.\]
\end{remark}
\section{Seeking Logical Reasons}\label{sec:incomplete}

In the previous sections, we mainly investigated the question \emph{whether} a given PDB is representable in a particular way without much focus on the \enquote{\emph{why}}. The arguments concerned the probability distributions of probabilistic databases and, in particular, on convergence issues regarding the marginal probabilities of facts. For example, in \cref{sec:negative} we showed probabilistic databases violating the finite moments property can not be in $\FOTI$. In essence, this is an arithmetical reason for non-representability that only takes the probabilities of possible worlds into account. There could, however, be purely logical reasons for non-representability in $\FOTI$ that emerge when we forget about the probabilities and are just interested in the mere possible worlds. In this section, we aim to separate \enquote{arithmetical} arguments from \enquote{purely logical} ones. 

Discarding the probabilities, we are left with sets of possible worlds such as they occur in the study of \emph{incomplete databases} \cite{Imielinski+1984,Abiteboul+1991}. Given a probabilistic database $\pdb = ( \DD, P )$, we focus on $\worlds( \pdb )$, the set of instances $D \in \DD$ with positive probability $P\big( \set{D} \big)$. If $\classstyle{D}$ is a class of PDBs, then $\worlds(\classstyle{D}) \coloneqq \set[\big]{ \worlds(\pdb) \mid \pdb\in \classstyle{D}}$. We now investigate what we can conclude about the possible representations of a probabilistic database based solely on its set of possible worlds. In \cref{ssec:nonrep-logical}, we show how this point of view can be used to establish the non-representability of infinite PDBs. In \cref{ssec:fo-logical}, we discuss the implications for the class $\FOTI$.

\subsection{Proving Non-Representability}\label{ssec:nonrep-logical}

A tool that falls into the category of logical arguments is the monotonicity of certain classes of views: A view $V\colon \DD_1 \to \DD_2$ is called \emph{monotone} if, for any instances $D, D'$ in $\DD_1$, we have that $D \subseteq D'$ implies $V(D) \subseteq V(D')$. For example, $\UCQ$- and Datalog views are monotone~\cite{Abiteboul+1995}. In \cite[Proposition 2.17]{Suciu+2011} the monotonicity of $\UCQ$-views is exploited to show that some finite probabilistic databases can not be represented by a $\UCQ$-view over a finite $\TI$-PDB. This argument only needs to consider the sets of possible worlds and is oblivious to the concrete probability distributions. More precisely, they use the existence of a unique (inclusion-wise) maximal possible world in any finite $\TI$-PDB and infer that a $\UCQ$-view over such will also contain a maximal possible world:

\begin{proposition}[{\cite[proof of Proposition 2.17]{Suciu+2011}}]\label{pro:monotone-max-world}
    Let $\classstyle{V}$ be any class of monotone views. Then for every PDB $\pdb \in \classstyle{V}(\TIfin)$ there exists $D_{\mathrm{max}} \in \worlds(\pdb)$ such that $D \subseteq D_{\mathrm{max}}$ for all $D \in \worlds(\pdb)$.
\end{proposition}

It is then easy to construct probabilistic databases without a unique maximal possible world, and it directly follows that such PDBs cannot belong to $\UCQ\big( \TIfin \big)$. This proof does not directly carry over to the infinite setting though, as infinite $\TI$-PDBs in general may not have maximal possible worlds at all. However, a simple argument that works in both the finite and the infinite setting demonstrates how monotonicity can be used with respect to the set of possible worlds in order to conclude non-representability in general. For this, we start with some easy observations regarding the structure of the possible worlds of $\TI$-PDBs.

\begin{observation}\label{obs:inc+ti}
    Let $\pdb[I]$ be a $\TI$-PDB, and let $p_f$ denote the marginal probability of fact $f$. Then $\facts(I)$ is partitioned into the sets $\Fa \coloneqq \set{ f \with p_f = 1 }$, $\Fs \coloneqq \set{ f \with p_f \in (0,1) }$ and $\Fn \coloneqq \set{ f \with p_f = 0 }$ and it holds that $\worlds( \pdb[I] ) = \set[\big]{ \Fa \cup F \with F \subseteq \Fs \text{ and } \size{ F } < \infty }$.
\end{observation}

We next inspect the effect of applying views over PDBs and observe that stripping PDBs of their probabilities commutes with applying views. That is, the diagram in \cref{fig:cd1} commutes.

\begin{figure}[hb]
    \centering%
    \hfill%
    \begin{subfigure}[t]{.45\textwidth}\centering%
        \begin{tikzcd}
            \pdb \arrow[r,"V"] \arrow[d,"\worlds"'] & V(\pdb) \arrow[d,"\worlds"]\\
            \worlds(\pdb) \arrow[r,"V"] & \DD'
        \end{tikzcd}
        \caption{Discarding probabilities in PDBs commutes with applying views.}\label{fig:cd1}
    \end{subfigure}
    \hfill%
    \begin{subfigure}[t]{.45\textwidth}\centering%
        \begin{tikzcd}
            \classstyle{D} = \set[\big]{ \pdb_i^{} }_i \arrow[r,"\classstyle{V}"] \arrow[d,"\worlds"'] & \classstyle{V}(\classstyle{D}) \arrow[d,"\worlds"]\\
            \worlds(\classstyle{D}) \arrow[r,"\classstyle{V}"] & \set[\big]{ \DD_i' }_{i}
        \end{tikzcd}
        \caption{Discarding probabilities in classes of PDBs commutes with a applying views from a class.}\label{fig:cd2}
    \end{subfigure}
    \hfill\mbox{}%
    \caption{Commutative diagrams illustrating \cref{obs:commutes} and \cref{pro:prob_to_incomplete}.}
\end{figure}

\begin{observation}\label{obs:commutes}
    Let $\pdb = ( \DD_1, P )$ be a PDB, let $\DD_2$ be a set of database instances, and let $V$ be a view $V \from \DD_1 \to \DD_2$. Then, it holds that
    $V\big( \worlds( \pdb ) \big) = \worlds\big( V( \pdb ) \big)$.
\end{observation}
\begin{proof}
    If $D_2 \in V\big( \worlds( \pdb ) \big)$, then there exists $D_1 \in \worlds( \pdb )$ such that $V( D_1 ) = D_2$. As $P\big( \set{ D_1 } \big) > 0$, this implies that $P_{ V(\pdb) }\big( \set{ D_2 } \big) > 0$. Thus, $V\big( \worlds(\pdb) \big) \subseteq \worlds\big( V( \pdb) \big)$.
    
    If $D_2 \in \worlds\big( V( \pdb ) \big)$, then $P_{\pdb}\big( V^{-1}\big( \set{ D_2 } \big) \big) = P_{V(\pdb)}\big( \set{D_2} \big) > 0$. Thus, there exists an instance $D_1 \in V^{-1}\big( \set{D_2} \big)$ with $P_{\pdb}( D_1 ) > 0$. In particular, $D_2 \in V\big( \worlds( \pdb ) \big)$, which shows $\worlds\big( V( \pdb ) \big) \subseteq V\big( \worlds( \pdb ) \big)$.
\end{proof}

\Cref{obs:commutes} directly implies the following.

\begin{proposition}\label{pro:prob_to_incomplete}
    Let $\classstyle{D}$ be a class of PDBs and let $\classstyle{V}$ be a class of views. Then, it holds that $\classstyle{V}\big( \worlds( \classstyle{D} ) \big) = \worlds\big( \classstyle{V}( \classstyle{D} ) \big)$. In particular, $\worlds( \pdb ) \in \classstyle{V}\big( \worlds(\classstyle{D}) \big)$ holds for all $ \pdb \in \classstyle{V}( \classstyle{D} )$.
\end{proposition}

The first part of \cref{pro:prob_to_incomplete} states that the diagram in \cref{fig:cd2} commutes. The latter part of the proposition can be understood as a purely logical necessary condition for representability: if $\worlds(\pdb) \not\in \classstyle{V}\big(\worlds(\classstyle{D})\big)$, then it follows immediately that $\pdb \not\in \classstyle{V}(\classstyle{D})$. We now build on \cref{pro:prob_to_incomplete} and \cref{obs:inc+ti}, and demonstrate how monotonicity can be used to conclude non-representability in the setting of infinite PDBs. 

\begin{proposition}\label{pro:mon+ti+disjoint}
    Let $\classstyle{V}$ be a class of monotone views, and let $\pdb$ be a PDB such that there exist two facts $f_1$ and $f_2$ of positive marginal probability that are mutually exclusive (that is, $\Pr_{ D \sim \pdb }( f_1, f_2 \in D ) = 0$). Then $\pdb \notin \classstyle{V}(\TI)$.
\end{proposition}

\begin{proof}
    Assume by contradiction that $\pdb \in \classstyle{V}( \TI )$ and let $D_1, D_2 \in \worlds( \pdb )$ with $f_1 \in D_1$ and $f_2 \in D_2$. By \cref{pro:prob_to_incomplete}, we know that $\worlds( \pdb ) \in \classstyle{V}\big(\worlds( \TI ) \big)$. Hence, there exists a set of instances $\II \in \worlds( \TI )$ and a view $V \in \classstyle{V}$ such that $V( \II ) = \worlds( \pdb )$. Thus, let $I_1, I_2 \in \II$ such that $V( I_1 ) = D_1$ and $V( I_2 ) = D_2$. By \cref{obs:inc+ti}, we also have $I_1 \cup I_2 \in \II$ and, thus,
    $V( I_1 \cup I_2 ) \in \worlds( \pdb )$. Since $V$ is monotone, it holds that
    \[
        V\big( I_1 \cup I_2 \big) \supseteq V( I_1 ) \cup V( I_2 ) = D_1 \cup D_2 \supseteq \set[\big]{ f_1, f_2 }\text,
    \]
    in contradiction to $f_1$ and $f_2$ being mutual exclusive.
\end{proof}

\Cref{pro:mon+ti+disjoint} immediately yields that $\UCQ(\TI)$ does not contain any $\BID$-PDBs that are not already tuple-independent.

\subsection{First-Order Views}\label{ssec:fo-logical}

\Cref{pro:mon+ti+disjoint} demonstrates that it is possible to determine non-representability of an infinite PDB based on purely logical reasons. However, in the following, we show that there are no logical reasons to show that a PDB is not in $\FOTI$. More precisely, for any collection of possible worlds (over the same schema), we can attach probabilities to the worlds in such a way that we end up with a PDB in $\FOTI$.

\begin{lemma}\label{lem:no-logical-reason}
    Let $\DD$ be any set of database instances over the same schema. Then there exists a PDB $\pdb \in \FOTI$ such that $\worlds(\pdb) = \DD$.
\end{lemma}
\begin{proof}
    Let $\DD = \set[\big]{ D_1, D_2, \dots }$ and define
    \[
        z_i \coloneqq 
        \begin{cases}
            \big( 2^{-i} \cdot \size{D_i}^{-1} \big)^{ \size{ D_i } }   & \text{if }D_i \neq \emptyset\text{ and}\\
            1   &\text{otherwise.}
        \end{cases}
    \]
    Note that $z_i \in (0,1]$. It holds that
    \[
        0 
            < 
        \sum_{ i = 1 }^{ \infty } z_i 
            \leq 
        1 + \sum_{ i, D_i \neq \emptyset } \Big(\frac{ 1 }{ 2^i \cdot \size{ D_i } }\Big)^{ \size{D_i} }
            \leq
        1 + \sum_{ i, D_i \neq \emptyset } \frac{ 1 }{ 2^i \cdot \size{ D_i } }
            \leq
        1 + \sum_{ i, D_i \neq \emptyset } \frac{ 1 }{ 2^i }
            \leq 
        2 \text.
    \]
    Thus, $\frac{1}{Z} \in \big[ \frac12, \infty \big)$ for $Z \coloneqq \sum_{ i = 1 }^{ \infty } z_i$. We now define probability distribution $P \from \DD \to [0,1]$ by letting $P\big( \set{ D_i } \big) \coloneqq \frac{ z_i }{ Z }$ for all $i \in \NN_+$. Then it holds that
    \[
        \sum_{ D \in \DD \setminus \emptyset } \size{ D } \cdot P\big( \set{ D } \big)^{ 1 / \size{ D } }
            = 
        \sum_{ i, D_i \neq \emptyset } \size{ D_i }\cdot \Big( \frac{ z_i }{ Z } \Big)^{ 1 / \size{D_i} }
            =
        \sum_{ i, D_i \neq \emptyset } \frac{1}{2^i} \cdot \Big(\frac{1}{Z}\Big)^{1/\size{D_i} } 
            \leq
        \sum_{ i = 1 }^{ \infty } \frac{1}{2^i} \cdot \max\set[\big]{ 1, \tfrac{1}{Z} }
            < 
        \infty\text.
    \]
    This means that $\pdb = (\DD,P)$ satisfies the condition from \cref{thm:condition}, and thus, $\pdb \in \FOTI$. 
\end{proof}

\Cref{lem:no-logical-reason} shows that in order to prove that a PDB is not in $\FOTI$, it is not enough to only consider the set of its possible worlds. Any such proof \emph{must} take probabilities of the possible worlds into account.

Similar to the terminology for PDBs, we call a set $\DD$ of database instances \emph{of unbounded instance size}, if for all $n \in \NN$ there exists some $D \in \DD$ with $\size{D} \geq n$. Recall that we already know from \cref{cor:bounded-size}, that if the underlying set of possible worlds of a PDB is of bounded instance size, then the PDB is in $\FOTI$, regardless of the probabilities. \Cref{lem:no-logical-reason} says that for any tentative set of possible worlds of unbounded instance size, we can always find a probability distribution such that the resulting PDB is in $\FOTI$. The following result, however, states that in this setting we can also always find probability distributions such that the resulting PDB will \emph{not} be in $\FOTI.$

\begin{lemma}\label{lem:unbounded-to-expected}
    Let\/ $\DD$ be a set of instances over the same schema such that\/ $\DD$ is of unbounded instance size. Then there exists a PDB $\pdb$ with $\Expectation\big( \size{ \? } \big) = \infty$ such that $\worlds( \pdb ) = \DD$.
\end{lemma}
\begin{proof}
    Since the size of instances in $\DD$ is unbounded, there exists an infinite sequence of non-empty instances $\big( D_{i_k} \big)_{k \in \NN_+}$ such that $\size[\big]{ D_{i_k} }$ is strictly increasing and, in particular, $\size[\big]{ D_{i_k} } \geq k$. For all $k \in \NN_+$, we define 
    \[
        P\big( \set[\big]{ D_{i_k} } \big) 
            \coloneqq 
        \frac{3}{\pi^2 \cdot k^2}\text.
    \]
    Then $\sum_{ k = 1 }^{ \infty } P\big( \set[\big]{ D_{i_k} } \big) = \frac12$. We assign positive probabilities to the instances in $\DD \setminus \set[\big]{ D_{ i_k } \with k \in \NN_+ }$ such that they also add up to $\frac12$. Then $P \from \DD \to [0,1]$ is a probability distribution and $\pdb = ( \DD, P )$ is a probabilistic database. However, the expected instance size of $\pdb$ is infinite:
    \[
        \Expectation_{\pdb}\big( \size{ \? } \big)
            =
        \sum_{ D \in \DD } \size{ D } \cdot P\big( \set{ D } \big)
            \geq
        \sum_{ k = 1 }^{ \infty } \size[\big]{ D_{i_k} } \cdot P\big( \set[\big]{ D_{i_k} } \big) 
            \geq
        \sum_{ k = 1 }^{ \infty } k \cdot \frac{3}{\pi^2\cdot k^2}
            =
        \frac{3}{\pi^2} \cdot \sum_{ k = 1 }^{ \infty } \frac{1}{k}
            = \infty\text.\qedhere
    \]
\end{proof}

The following theorem summarizes our results regarding the class $\FOTI$ that emerged from the inspection of the underlying sets of possible worlds.

\begin{theorem}\label{thm:logical-fo}
    Let\/ $\DD$ be a set of database instances over the same schema. 
    \begin{enumerate}
        \item If\/ $\DD$ has bounded instance size, then $\pdb \in \FOTI$ for every $\pdb$ with $\worlds(\pdb) = \DD$.
        \item If\/ $\DD$ has unbounded instance size, then there exist PDBs $\pdb_1 \in \FOTI$ and $\pdb_2 \notin \FOTI$ such that $\worlds( \pdb_1 ) = \worlds( \pdb_2 ) = \DD$.
    \end{enumerate}
\end{theorem}
\begin{proof}
    The first statement directly follows from \cref{cor:bounded-size}. For unbounded instance size, \cref{lem:no-logical-reason} shows the existence of $\pdb_1$ and \cref{lem:unbounded-to-expected} shows the existence of $\pdb_2$, where the latter is not in $\FOTI$ due to \cref{pro:finite-moments}.
\end{proof}

In short, assume we are given a PDB $\pdb$ and want to determine whether $\pdb \in \FOTI$. If $\worlds( \pdb )$ has bounded instance size, then $\pdb \in \FOTI$. Otherwise, we have to investigate the instance probabilities to settle the question.
\section{Views Using Fragments of First Order-Logic}\label{sec:fragments}

\subsection{The Situation for Finite Probabilistic Databases}

In this section, we review the relative expressive power of classes of views over finite PDBs. Recall that $\FO(\TIfin) = \PDBfin$, as shown in \cite[Proposition 2.16]{Suciu+2011}, but $\UCQ(\TIfin) \subsetneq \PDBfin$, as shown in \cite[Proposition 2.17]{Suciu+2011}. The following two examples show that the classes $\UCQ(\TIfin)$, $\CQ(\TIfin)$ and $\sjfCQ(\TIfin)$ are incomparable to $\BIDfin$ (with respect to $\subseteq$).

\begin{example}\label{exa:finmutex}
    Consider the $\BID$-PDB $\pdb[I]$ that consists of only a single block containing two facts $f$ and $f'$, each with marginal probability $\frac{1}{2}$. Then $\worlds( \pdb[I] ) = \set[\big]{ \set{ f }, \set{ f' } }$ and both worlds $\set{ f }$ and $\set{ f' }$ are maximal. It follows from \cref{pro:monotone-max-world} that $\pdb[I] \notin \classstyle{V}\big(\TIfin\big)$ for any class $\classstyle{V}$ of monotone views. In particular, $\pdb[I] \notin \TIfin$.
\end{example}

\begin{example}\label{exa:RSexample}
    Consider the $\TI$-PDB $\pdb[I] = (\II,P)$ with $\facts( \pdb[I] ) = \set[\big]{ R(1,1), R(1,2), R(2,2), S(1), S(2) }$, where the marginal probability of all $R$-facts is $1$ and the marginal probability of both $S$-facts is $\frac12$. Now consider the $\sjfCQ$-view
    $\Phi(x) \coloneqq \exists y \with R(x,y) \wedge S(y)$. Then $\Phi( \pdb[I] ) \in \sjfCQ(\TIfin)$. Note that the PDB $\pdb[I]$ has possible worlds and images according to $\Phi$ as shown in \cref{fig:RSexample}.
    \begin{table}[h!]
        \centering
        \caption{Possible worlds of the PDB $\pdb[I]$.}\label{fig:RSexample}
        \begin{tabular}{c c}\toprule
            $I$                                                 & $\Phi(I)$\\\midrule
            $\set[\big]{ R(1,1), R(1,2), R(2,2) }$              & $\emptyset$\\
            $\set[\big]{ R(1,1), R(1,2), R(2,2), S(1) }$        & $\set[\big]{ R_{\Phi}(1) }$\\
            $\set[\big]{ R(1,1), R(1,2), R(2,2), S(2) }$        & $\set[\big]{ R_{\Phi}(1), R_{\Phi}(2) }$\\
            $\set[\big]{ R(1,1), R(1,2), R(2,2), S(1), S(2)) }$ & $\set[\big]{ R_{\Phi}(1), R_{\Phi}(2) }$\\\bottomrule
        \end{tabular}
    \end{table}
    Thus, $\Phi( \pdb[I] )$ has exactly three possible worlds: the empty instance $\emptyset$ with probability $\frac14$, the instance $\set[\big]{ R_{\Phi}( 1 ) }$, also with probability $\frac14$, and the instance $\set[\big]{ R_{\Phi}(1), R_{\Phi}(2) }$ with probability $\frac12$. Since the facts $R_{\Phi}(1)$ and $R_{\Phi}(2)$ are neither independent nor mutually exclusive in $\Phi(\pdb[I])$, it holds that $\Phi(\pdb[I]) \notin \BIDfin$.
\end{example}

Note that the latter example also shows that $\sjfCQ(\TIfin) \supsetneq \TIfin$. Remarkably, it turns out that the classes $\CQ(\TIfin)$ and $\UCQ(\TIfin)$ collapse to $\sjfCQ(\TIfin)$. This is a consequence of the following, more general insight.

\begin{proposition}\label{pro:sjfCQTI}
    Let $\pdb = V( \pdb[I] )$ where $\pdb[I] \in \TIfin$ and\/ $V$ is a monotone view. Then $\pdb \in \sjfCQ(\TIfin)$.
\end{proposition}

\begin{proof}
    We show this for the case where $V = \set{ Q }$ is a single query with output relation symbol $R_Q = R$. This generalizes directly to general monotone views. As earlier, for $\pdb[I] = ( \II, P_{\pdb[I]} )$, we let $\Fa(\pdb[I])$ and $\Fs(\pdb[I])$ denote the sets of facts that appear with probability $1$ and $>0$, respectively. Suppose that $\Fs(\pdb[I]) = \set{ f_1, \dots, f_n }$ (with $f_i$ pairwise different). We construct a $\TIfin$-PDB $\pdb[J] = (\JJ,P_{\pdb[J]})$ over the same universe as $\pdb[I]$, augmented (disjointly) by $\set{0,1}$. The schema of $\pdb[J]$ consists of $n$ new unary relation symbols $S_1,\dots,S_n$, and a separate new $n+r$-ary relation symbol $S$, where $r$ is the arity of the query $Q$. The marginal probabilities of the facts in $\pdb[J]$ are given as follows:
    \begin{itemize}
        \item The fact $S_i(0)$ has marginal probability $1$ for all $i = 1, \dots, n$.
        \item The fact $S_i(1)$ has marginal probability $P_{\pdb[I]}( f_i )$ for all $i = 1,\dots, n$.
        \item For all $a_1,\dots,a_n \in \set{0,1}$ and all facts $R(b_1,\dots,b_r) \in Q\big( \Fa(\pdb[I]) \cup \set{ f_i \with a_i \neq 0 }\big)$, the fact $S( a_1, \dots, a_n, b_1, \dots, b_r )$ has marginal probability $1$.
        \item All other facts have probability $0$.
    \end{itemize}
    Consider the following $\sjfCQ$-query $\Phi$ with output relation symbol $R_\Phi = R$ defined by
    \[
        \Phi( y_1, \dots, y_r )
            =
        \exists x_1 \dotsc \exists x_n \with S_1(x_1) \wedge \dots \wedge S_n(x_n) \wedge S( x_1, \dots, x_n, y_1, \dots, y_r )\text.
    \]
    We claim that $\Phi( \pdb[J] ) = Q( \pdb[I] )$. We show this by constructing a bijection $\beta$ between the worlds of $\pdb[I]$ and the worlds of $\pdb[J]$ with the properties that \begin{enumerate*} \item $P_{\pdb[I]}( I ) = P_{\pdb[J]}( \beta(I) )$; and \item $Q( I ) = \Phi( J )$. \end{enumerate*}
    
    Let $I$ be any possible world of $\pdb[I]$. Then $I = \Fa( \pdb[I] ) \cup \set{ f_{i_1}, \dots, f_{i_k} }$ for some distinct facts $f_{i_1},\dots,f_{i_k} \in \Fs( \pdb[I] )$. We let $\beta( I )$ be the instance of $\pdb[J]$ that contains all facts of marginal probability $1$, along with (precisely) the facts $S_{i_1}(1), \dots, S_{i_k}(1)$. Clearly, $\beta$ is both injective and surjective, i.\,e., a bijection. Moreover, we have that
    \[
        P_{\pdb[I]}\big( \set{I} \big) 
            =
        \prod_{ j = 1 }^k P_{\pdb[I]}\big( f_{i_j} \big)
            =
        \prod_{ j = 1 }^k P_{\pdb[J]}\big( S_{i_j}(1) \big)
            =
        P_{\pdb[J]}\big( \set{J} \big)\text.
    \]
    Therefore, it only remains to show that for all $I \in \worlds(\pdb[I])$, it holds that $Q(I) = \Phi( \beta(I) )$.
    
    Let $I = \Fa( \pdb[I] ) \cup \set{ f_{i_1}, \dots, f_{i_k} }$ as before and suppose that $R(b_1,\dots,b_r) \in Q(I)$. By construction, the instance $\beta(I)$ contains the facts $S_{i_1}(1), \dots, S_{i_k}(1)$. Moreover, as $R(b_1,\dots,b_r) \in Q( I )$, by definition, the fact $S( a_1,\dots,a_n,b_1,\dots,b_r )$ where
    \[
        a_i = 
        \begin{cases}
            1 & \text{if }i \in \set{i_1,\dots,i_k}\\
            0 & \text{otherwise.}
        \end{cases}
    \]
    for all $i=1,\dots,n$ is one of the facts with marginal probability $1$ in $\pdb[J]$. Together this shows that $J \models \Phi(b_1,\dots,b_r)$, i.\,e. $R(b_1,\dots,b_r) \in \Phi(J)$.
    
    For the other direction, suppose that $R(b_1,\dots,b_r) \in \Phi\big(\beta(I)\big)$. This means that there exist $a_1,\dots,a_n$ such that $S_1(a_1),\dots,S_n(a_n) \in \beta(I)$ and $R(b_1,\dots,b_r) \in Q\big( \Fa( \pdb[I] ) \cup \set{ f_j \with a_j\neq 0} \big)$. Note that 
    \[
        \set[\big]{ f_j \with a_j \neq 0 } 
        \subseteq
        \set[\big]{ f_j \with S_j(1) \in \beta(I)}\text.
    \]
    As $Q$ is monotonous, we have that
    \[
    Q\Big( \Fa( \pdb[I] ) \cup \set[\big]{ f_j \with a_j\neq 0} \Big) 
    \subseteq 
    Q\Big( \Fa( \pdb[I] ) \cup \set[\big]{ f_j \with a_j\neq 0} \Big) 
            = 
        Q( I )\text,
    \] 
    which finally yields
    \(
        R(b_1,\dots,b_r) \in 
        Q( I )\text.
    \)
\end{proof}

\begin{corollary}\label{cor:sjftifin}
    $\sjfCQ(\TIfin) = \CQ(\TIfin) = \UCQ(\TIfin)$.
\end{corollary}

For finite $\BID$-PDBs, the classes $\sjfCQ(\BIDfin)$, $\CQ(\BIDfin)$ and $\UCQ(\BIDfin)$: in \cite[Proposition 2.18]{Suciu+2011} it is shown that $\PDBfin = \CQ(\BIDfin)$. In fact, the conjunctive query used in the proof is self-join free, showing that $\PDBfin = \sjfCQ(\BIDfin)$. 

\begin{proposition}
    The classes of views over finite probabilistic databases have the relative expressive power shown in the Hasse diagram of \cref{fig:hassefin}.
\end{proposition}

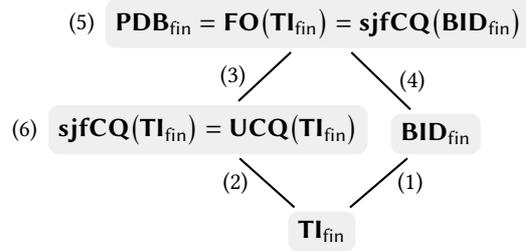
\begin{figure}[h]\centering
    \tikzset{ucqtifin/.style={fill=lightgray!25!white}}
    \tikzset{fotifin/.style= {fill=lightgray!25!white}}
    \tikzset{bid/.style=     {fill=lightgray!25!white}}
    \tikzset{ti/.style=      {fill=lightgray!25!white}}
    \begin{tikzpicture}
        \newcommand*{\texteq}[1]{\mathrel{\overset{\smash{\scriptscriptstyle{\text{#1}}}}{=}}}
        \tikzset{class/.style={rectangle,rounded corners,font={\small}}}
        \tikzset{line/.style={draw,thick,shorten <= 2pt, shorten >=2pt}}
        \matrix[rectangle,anchor=south] (REPfin) at (0,5) [column sep={1.5cm,between origins}, row sep=.75cm] {
            \node[missing] {}; &
            \node[class,align=center,fotifin,label={[font=\footnotesize]left:(5)}] (PDBfin) {%
                $\PDBfin = \FO\big(\TIfin\big) = \sjfCQ\big(\BIDfin\big)$
            }; &
            \node[missing] {};\\
            \node[class,align=center,ucqtifin,label={[font=\footnotesize]left:(6)}] (UCQTIfin) {$\sjfCQ\big(\TIfin\big)= \UCQ\big(\TIfin\big)$}; & 
            \node[missing] {}; & 
            \node[class,bid] (BIDfin) {$\BIDfin$};\\
            \node[missing] {}; & 
            \node[class,ti] (TIfin) {$\TIfin$}; &
            \node[missing] {};\\
        };
        \draw[line] (TIfin) to      node[circle,right,font=\footnotesize]{(1)} (BIDfin);
        \draw[line] (TIfin) to      node[circle,left, font=\footnotesize]{(2)}(UCQTIfin);
        \draw[line] (UCQTIfin) to   node[circle,left, font=\footnotesize]{(3)} (PDBfin);
        \draw[line] (BIDfin) to     node[circle,right,font=\footnotesize]{(4)} (PDBfin);
    \end{tikzpicture}
    \caption{Hasse diagram for the relative expressive power of views over finite probabilistic databases.}\label{fig:hassefin}
\end{figure}

\begin{proof}
    Consider \cref{fig:hassefin}. The inclusions depicted therein (shown as solid lines) are trivial. We now give arguments for (1) throughout (4) that these inclusions are proper and that the remaining classes (5) and (6) collapse.
    \begin{enumerate}
        \item This is witnessed by \cref{exa:finmutex}.
        \item This is witnessed by \cref{exa:RSexample}.
        \item This is shown in \cite[Proposition 2.17]{Suciu+2011}.
        \item This is shown in \cite[Chapter 2.7]{Suciu+2011}, another example is given by \cref{exa:RSexample}.
        \item This is shown in \cite[Propositions 2.16 and 2.18]{Suciu+2011}.
        \item This is the statement of \cref{cor:sjftifin}.\qedhere
    \end{enumerate}
\end{proof}

\subsection{The Situation for Countably Infinite Probabilistic Databases}

In the previous section, we discussed finite $\TI$- and $\BID$-PDBs and views in the $\UCQ$-, $\CQ$- and $\sjfCQ$-fragment of first-order logic. We now want to investigate the corresponding classes in the countably infinite setting. Interestingly, the different classes that collapse in the finite setting do no longer collapse in the infinite setting, yielding pairwise different classes of probabilistic databases. A reason for is intuitively that in a representation, we can no longer hard-code the complete structure of possible worlds using facts of probability one. This was possible in the finite setting, and is essentially what is done in the proofs of the respective statements \cite[Propositions 2.16 and 2.18]{Suciu+2011}.

The tools we have seen and developed so far are ineffective for a possible separation of the classes $\UCQ(\TI)$, $\CQ(\TI)$ and $\sjfCQ(\TI)$ and the corresponding classes for $\BID$-PDBs as they only concern the class $\FOTI$ and we already know from \cref{pro:mon+ti+disjoint,thm:bid} that $\UCQ(\TI)$ is strictly weaker than $\FOTI$. Thus, we have to develop new methods.

The separations we find in this section rely on the varying level of capability for the respective classes to represent symmetries. For our arguments, we thus focus on a special class of probabilistic databases that facilitates such discussions. A \emph{graph database} $D$ is a database instance over a single binary relation $E$ with universe $\UU^2$. It is called 
\begin{itemize}
    \item \emph{simple}, if it contains no fact $E(a,a)$ and
    \item \emph{undirected}, if $E(a,b) \in D$ implies $E(b,a) \in D$.
\end{itemize}
A \emph{probabilistic graph} is a PDB whose possible worlds are graph databases.

\begin{definition}\label{def:edgeindepgraphs}
    Consider the class $\edgeindepgraphs$ of probabilistic databases $\pdb[D]$ over ordered, countably infinite universes $(\UU,<)$ satisfying all of the following properties:
    \begin{itemize}
        \item every possible world of $\pdb[D]$ is a simple and undirected graph database,
        \item it is \emph{edge-independent}, that is, for any sequence of pairwise distinct two-element subsets $\set{a_1,b_1}, \dots, \set{ a_k,b_k }$ of $\UU$, the facts $E(a_1,b_1),\dots, E(a_k,b_k)$ are stochastically independent in $\pdb[D]$, that is,
        \[
            \Pr_{ D \sim \pdb } \big( E(a_i,b_i) \in D \in D \text{ for all } i = 1, \dots, k \big)
                =
            \prod_{ i = 1 }^{ k } p_{a_i,b_i}
                \text.
        \]
        and 
        \item it is \emph{unbounded}, that is, for every $n \in \NN$ there exists $D \in \worlds(\pdb[D])$ with $\size{D} > n$.
    \end{itemize}
\end{definition}

If $\pdb[D]$ is a PDB in $\edgeindepgraphs$, we let $p_{a,b}$ denote the marginal probability of the fact $E(a,b)$, that is, $p_{a,b} = \Pr_{D \sim \pdb[D]}\big(E(a,b) \in D\big)$. Since every possible world of $\pdb[D]$ is an undirected (and simple) graph database, we have $p_{a,b} = p_{b,a}$ for all $a$ and $b$. Note that $\pdb[D]$ is uniquely determined by $(p_{a,b})_{a \neq b}$.

\begin{remark}
    Intuitively, every PDB $\pdb[D]$ in $\edgeindepgraphs$ represents a random graph with vertex set $\UU$ and edges $\set{a,b} \subseteq \binom{\UU}{2}$ where all edges are independent, and the probability of the edge $\set{a,b}$ being present is precisely $p_{a,b}$. Thus, these PDBs are tightly connected to a particular model that arises in the study of random graphs \cite[Chapter 10.1]{FriezeKaronski2016}.
\end{remark}

Naturally, despite not being tuple-independent, the PDBs of $\edgeindepgraphs$ are strongly related to $\TI$-PDBs.

\begin{proposition}\label{pro:edgeindepgraphsti}
   Let $\big(p_{a,b}\big)_{a<b}$ be a family of numbers in $[0,1]$ and let $p_{b,a} = p_{a,b}$ for all $a,b$.
   Then $\big( p_{a,b} \big)_{ a \neq b }$ defines a probabilistic graph in $\edgeindepgraphs$ if and only if $\big( p_{a,b} \big)_{ a < b }$, taken as marginal probabilities for the facts $\set[\big]{E(a,b)}_{a<b}$, defines a $\TI$-PDB.
\end{proposition}

In particular, $\big( p_{a,b} \big)_{ a \neq b }$ defines a probabilistic graph in $\edgeindepgraphs$ if and only if $\sum_{ a \neq b } p_{a,b}$ is finite.

\begin{proof}
    First, let $\pdb[D] \in \edgeindepgraphs$. We construct a new PDB $\pdb[D]'$ by restricting $\pdb[D]$ to the facts $\set[\big]{E(a,b)}_{a<b}$. In detail, for every world $D \in \worlds( \pdb[D] )$ we only keep the facts $E(a,b)$ with $a<b$. Note that this transformation is a one-to-one-correspondence between $\worlds( \pdb[D] )$ and $\worlds( \pdb[D]' )$. Therefore, $\pdb[D]'$ is well-defined. Moreover, all the facts in $\pdb[D]'$ are independent and have the same marginal probabilities as before, rendering $\pdb[D]'$ a $\TI$-PDB with marginal probabilities $\big( p_{a,b} \big)_{ a < b }$.
    
    Second, let $\I$ be a $\TI$-PDB over $\set[\big]{E(a,b)}_{a<b}$ with marginal probabilities $\big( p_{a,b} \big)_{ a < b }$. Inverting the construction from the first step, we define $\pdb[D]$ by adding $E(b,a)$ to every possible world that contains $E(a,b)$ for every $a<b$.
\end{proof}

It is quite easy to see that every PDB in $\edgeindepgraphs$ admits a $\UCQ(\TI)$-representation.

\begin{lemma}\label{lem:symgraphucq}
    If $\pdb[D] \in \edgeindepgraphs$, then $\pdb[D] \in \UCQ(\TI)$.
\end{lemma}

\begin{proof}
    Let $\pdb[D] = (\DD, P_{\D})$ be a probabilistic graph from $\edgeindepgraphs$ over $\UU$, and let $\pdb[I]$ be the $\TI$-PDB constructed from $\pdb[D]$ as in \cref{pro:edgeindepgraphsti}.
    We show that $\pdb[D]$ can be reconstructed from $\pdb[I]$ using the $\UCQ$-view $\Phi(x,y) \coloneqq E(x,y) \vee E(y,x)$. That is, we claim that $\Phi( \pdb[I] ) = \pdb[D]$. To see this, consider any possible world $D \in \worlds( \pdb[D] )$. Note that 
    \begin{equation}
        P_{\pdb[D]}(D) = \prod_{ \substack{ a < b \text{ and}\\E(a,b) \in D } } p_{a,b} \cdot \prod_{ \substack{ a < b \text{ and}\\E(a,b) \notin D } } \big(1-p_{a,b}\big)\text.
        \label{eq:probgraphucqprob}
    \end{equation}
    The preimage of $D$ under $\Phi$ consists of only the single instance
    \[
        I = \set[\big]{ E( a,b ) \with a < b \text{ and } E(a,b) \in D ) } \in \II\text.
    \]
    By construction, the probability of this instance in $\I$ is equal to \eqref{eq:probgraphucqprob} as well. That is, $\Phi$ is a probability preserving one-to-one correspondence between the worlds of $\I$ and the worlds of $\D$, so $\Phi(\I) = \D$.
\end{proof}

The next lemma gives an example of the other extreme: \emph{no} PDB in $\edgeindepgraphs$ admits a $\sjfCQ(\BID)$-presentation.

\begin{lemma}\label{lem:symgraphNOTsjfcqbid}
    If $\pdb[D] \in \edgeindepgraphs{}$, then $\pdb[D] \notin \sjfCQ(\BID)$.
\end{lemma}

The idea behind the proof of this lemma is as follows. We start by assuming that we have a $\sjfCQ(\BID)$-representation of some $\pdb \in \edgeindepgraphs$. The property we exploit is that facts in the underlying $\BID$-PDB can only be either independent or mutually exclusive. We pick a particularly chosen instance from the $\BID$-PDB and show that we can modify it in such a way that the resulting instance still has positive probability, but its image necessarily contains exactly one of the edges $E(a,b)$ or $E(b,a)$ for some $(a,b)$. This then contradicts the assumption that we started with a representation of a PDB from $\edgeindepgraphs$. 

\begin{proof}[{Proof of \cref{lem:symgraphNOTsjfcqbid}}]
    Let $\pdb[D] = ( \DD, P_{\pdb[D]} )$ be a probabilistic graph in $\edgeindepgraphs$ over $\UU$. 
    Assume that there exists a $\BID$-PDB $\pdb[I] = ( \II, P_{\pdb[I]} )$ and a self-join free $\CQ$-view of the following shape
    \[
        \Phi(x,y) \coloneqq \exists z_1\dotsc \exists z_m \with R_1( \tup t_1 ) \wedge \dotsm \wedge R_n( \tup t_n )
    \]
    such that $\pdb[D] = \Phi( \pdb[I] )$. This can be assumed since equality atoms can be removed using substitution and, as the PDBs in $\edgeindepgraphs$ are simple and undirected, $\Phi$ can not be of the shape $\Phi(x,x)$ or $\Phi(x,a)$. Without loss of generality, we can assume that all the variables $z_1,\dots,z_m$ appear among $R_1( \tup t_1 ) \wedge \dotsm \wedge R_n( \tup t_n )$. In the following, we let
    \[
        f_i[a; b; \tup c] \coloneqq R_i\big( \tup t_i[ x/a; y/b; z_1/c_1; \dotsc; z_m/c_m ] \big)
    \]
    for all $a,b \in \UU$ and all $\tup c = (c_1,\dots,c_m) \in \UU^m$.
 
    Observe that for all $I \in \II$ it holds that 
        \begin{align}
        I \models \Phi(a,b) &\iff \text{there exists }\tup c\text{ such that for all } i = 1,\dots,n \text{ it holds that } f_i[a;b;\tup c] \in I\text{;}\label{eq:symmetrycq1}
        \shortintertext{and}
        I \models \Phi(b,a) &\iff \text{there exists }\tup c'\text{ such that for all } i =1,\dots, n\text{ it holds that } f_i[b;a;\tup c'] \in I\label{eq:symmetrycq2}
        \text.
    \end{align}
    From now on, let $i \in \set{1,\dots,n}$ be fixed such that $R_i( \tup t_i )$ contains at least one of the variables $x$ or $y$. Without loss of generality, we assume that $R_i( \tup t_i )$ contains $x$ (the argument for $y$ follows analogously).
    
    Let $(a,b), (a',b') \in \UU^2$. We say that a block $B$ from $\pdb[I]$ is \emph{full}, if the marginal probabilities of the facts in $B$ add up to $1$. We say that $B$ is $(a,b)$-\emph{saturated} if $B$ is a full block such that every fact of $B$ is of the shape $f_i[a;b;\tup c]$ for some $\tup c \in \UU^m$. Note that if $(a,b)$ and $(a',b')$ have $a \neq a'$, then no block $B$ is \emph{both} $(a,b)$- and $(a',b')$-saturated, since $f_i$ contains the variable $x$ which yields $f_i[a;b;\tup c] \neq f_i[a';b';\tup c']$ for all $\tup c, \tup c' \in \UU^m$.
    
    Since $\pdb[D]$ is unbounded, there exists a family $T$ of infinitely many pairs $(a,b)$ with pairwise different values of $a$ such that $E(a,b)$ appears among the instances of $\worlds( \pdb[D] )$. Since $\pdb[I]$ contains at most finitely many full blocks, by the pigeonhole-principle, there exists a pair $(a^*,b^*) \in T$ with $a^*\neq b^*$ such that no block of $B$ is $(a^*, b^*)$-saturated.
    
    Let $D$ be one of the instances in $\worlds(\pdb[D])$ that contain both $E(a^*, b^*)$ and $E(b^*, a^*)$, and let $I \in \worlds( \pdb[I] )$ be an arbitrary preimage of $D$ under $\Phi$. Then both $I \models \Phi( a^*, b^* )$ and $I \models \Phi( b^*, a^* )$. In particular, $I$ contains facts $f_i[ a^*; b^*; \tup c ]$ and $f_i[ b^*; a^*; \tup c' ]$ for some (possibly multiple) $\tup c,\tup c' \in \UU^m$.
    
    Let $F(a^*,b^*) \coloneqq \set{ f_i[ a^*; b^*, \tup c ] \with \tup c \in \UU^m }$. We can partition $I$ into
    \begin{equation}\label{eq:Ia*b*}
        I = I' \cup I_{=1}(a^*,b^*) \cup I_{<1}(a^*,b^*)
    \end{equation}
    such that
    \begin{itemize}
        \item $I' \coloneqq I \setminus F(a^*,b^*)$,
        \item $I_{=1}(a^*,b^*) \coloneqq I \cap F(a^*,b^*) \cap \bigcap_{B \text{ full block}} B$, and
        \item $I_{<1}(a^*,b^*) \coloneqq I \cap F(a^*,b^*) \cap \bigcap_{B \text{ not full}} B$.
    \end{itemize}
    Recall that no block of $I$ is $(a^*,b^*)$-saturated. Hence, if $f \in I_{=1}(a^*,b^*)$, then there exists $f'$ in the same block as $f$ such that $f' \notin F(a^*,b^*)$. We simply drop all facts from $I_{<1}(a^*,b^*)$ and transform $I_{=1}( a^*, b^* )$ into a new set of facts $I''$ by replacing every fact $f$ in $I_{=1}( a^*,b^* )$ with such an $f'$. This yields a new set of facts
    \[
        J \coloneqq I' \cup I'' \in \worlds( \pdb[I] ),
    \]
    Now, $J \not\models \Phi( a^*,b^* )$, since $\Phi$ is self-join free and $J$ does not contain any of the facts $F(a^*,b^*)$. However, $I \cap F( b^*, a^* ) \subseteq J \cap F( b^*, a^* )$ (where $F( b^*,a^* )$ is defined analogously to $F( a^*,b^* )$), since $F( a^*,b^* ) \cap F( b^*,a^* ) = \emptyset$. This latter disjointness stems from the fact that $f_i$ contains the variable $x$ and the property that $\Phi$ is self-join free. Therefore, using the monotonicity of $\Phi$, it follows from \eqref{eq:symmetrycq2} that $J \models \Phi( b^*, a^* )$.
    
    Together we have that $(b^*,a^*) \in \Phi(J)$ but $(a^*,b^*) \notin \Phi(J)$, contradicting $\Phi(J) \in \worlds( \pdb[D] )$.
\end{proof}

\begin{remark}
    The above proof is similar to the discussions of \cref{sec:incomplete} in the sense that it only concerns the sets of possible worlds and is, apart from this, oblivious to probabilities.
\end{remark}

Note that as $\sjfCQ(\BID)$ contains $\sjfCQ(\TI)$, \cref{lem:symgraphNOTsjfcqbid} in particular shows that $\sjfCQ(\TI)$ contains none of the PDBs from $\edgeindepgraphs$. Recall that, on the contrary, $\UCQ(\TI)$ contains all the PDBs from $\edgeindepgraphs$. Naturally, the class $\CQ(\TI)$ should be investigated next. In fact, when it comes to the representation of PDBs from $\edgeindepgraphs$, the class $\CQ(\TI)$ proves to constitute a middle ground between the previous two extreme cases.

\begin{lemma}\label{lem:symgraphcqti}
    Let $\pdb[D]_1$ be a probabilistic graph in $\edgeindepgraphs$ with edge set $E = \set{a_i,b_i}_{i \in \NN}$ where the marginal probabilities of the edges $E(a_i,b_i)$ are given by $p_{a_i,b_i} = p_{b_i,a_i} = \frac{1}{i^4}$. Then $\pdb[D]_1 \in \CQ(\TI)$.
\end{lemma}

\begin{proof}
    First, note that $\sum_{ i = 1 }^{ \infty } \frac1{i^4} < \infty$, so $\pdb_1$ is well-defined. We will now construct a $\TI$-PDB $\pdb[I]$ and a $\CQ$ $\Phi$ such that $\pdb_1 = \Phi(\pdb[I])$. For this, we define $\pdb[I]$ to be the $\TI$-PDB over the fact set $\set[\big]{ E(a_i,b_i) \with i = 1, 2, \dots } \cup \set[\big]{ E(b_i,a_i) \with i=1,2,\dots }$ such that the marginal probability of $E(a_i,b_i)$ and $E(b_i,a_i)$, respectively, is given as $\frac{1}{i^2}$. As $\sum_{ i = 1 }^{ \infty }\frac1{i^2} < \infty$, this makes $\pdb[I]$ a well-defined PDB. Consider the $\CQ$ 
    \[
        \Phi( x,y ) = E( x,y ) \wedge E(y,x)\text.
    \]
    We claim that $\Phi( \pdb[I] ) = \pdb_1$. For this, it suffices to establish that $\Phi(\pdb[I])$ has the correct marginal probabilities and that $\Phi(\pdb[I]) \in \edgeindepgraphs$. For all $I \in \worlds(\pdb[I])$ and all $i = 1,2, \dots$, it holds that $I \models \Phi(a_i,b_i)$ if and only if $I$ contains both $E(a_i,b_i)$ and $E(b_i,a_i)$. The probability of this event is precisely $\frac{1}{i^2} \cdot \frac{1}{i^2} = \frac{1}{i^4}$, which is the marginal probability of $E(a_i,b_i)$ in $\pdb_1$.
    
    It remains to show that $\Phi(\pdb[I]) \in \edgeindepgraphs$. It is clear that every possible world of $\Phi(\pdb[I])$ is an undirected and simple graph database. Because the presence of any double edge $\set{E(a_i,b_i), E(b_i,a_i)}$ in $\Phi(\pdb[I])$ is determined solely by the presence of the facts $E(a_i,b_i)$ and $E(b_i,a_i)$ in $\pdb[I]$, the PDB $\Phi(\pdb[I])$ is also edge-independent. Moreover, $\Phi(\pdb[I])$ is unbounded, since the there are infinitely many facts of positive marginal probability and the edges are independent. Together, $\Phi(\pdb[I]) = \pdb_1$.
\end{proof}

Our current status is that $\UCQ(\TI)$ contains \emph{all} PDBs from $\edgeindepgraphs$, that $\sjfCQ(\TI)$ contains \emph{no} PDB from $\edgeindepgraphs$, and that $\CQ(\TI)$ contains \emph{some} PDBs of $\edgeindepgraphs$. Our next lemma yields, in particular, that $\CQ(\TI)$ is strictly weaker than $\UCQ(\TI)$ by giving an example of a PDB in $\edgeindepgraphs$ that is not contained in $\CQ(\BID)$.

\begin{lemma}\label{lem:symgraphNOTcqbid}
    Let $\pdb[D]_2$ be the probabilistic graph representing a countably infinite set of pairwise disjoint edges, where the $i$-th edge $\set{a_i,b_i}$ has probability $p_{a_i,b_i} = p_{b_i,a_i} = \frac{1}{i^2}$. Then $\pdb[D]_2 \notin \CQ(\BID)$.
\end{lemma}

The structure of this counterexample is strikingly similar to the $\edgeindepgraphs$-PDB we investigated in \cref{lem:symgraphcqti}. We comment on this after the proof. 

\begin{proof}[{Proof of \cref{lem:symgraphNOTcqbid}}]
    We assume otherwise and let $\pdb[I]$ be a $\BID$-PDB and $\Phi$ a $\CQ$-view 
    \[
        \Phi( x,y ) = \exists z_1 \dotsc \exists z_m \with R_1( \tup t_1 ) \wedge \dotsm \wedge R_n( \tup t_n )
    \]
    such that $\pdb[D]_2 = \Phi(\pdb[I])$. That we can assume $\Phi$ to be of the above shape has the same reason as in the proof of \cref{lem:symgraphNOTsjfcqbid}. We also reuse the notation from there and let $f_j[a;b;\tup c]$ denote the fact originating from $R_j( \tup t_j )$ by replacing $x,y, \tup z$ with $a,b$ and $\tup c$, respectively, for $j = 1,\dots, n$. Note that the equivalences from \eqref{eq:symmetrycq1} and \eqref{eq:symmetrycq2} still hold. Again, we pick an index $j \in \set{1,\dots,n}$ such that $x$ or $y$ or both appear in $R_j( \tup t_j )$.
    
    Let
    \[
        F(a,b) \coloneqq \set[\big]{ f_j[ a; b; \tup c ] \with \tup c \in \UU^m }
    \]
    for all $a,b \in \UU$. Because $\pdb[I]$ is a $\BID$-PDB, for all sets $F$ of facts (in particular, for the sets $F(a,b)$), it holds that
    \[
        \Pr_{ I \sim \pdb[I] }\big( \, I \cap F = \emptyset \, \big)
        =
        \prod_{ B\text{ block of }\pdb[I] } \bigg( 1 - \sum_{ f \in B \cap F } p_f \bigg)
        \text.
    \]
    
    Let $i \in \NN_+$ be arbitrary. Since $f_j$ contains $x$ or $y$, we have $F(a_i,b_i) \cap F(b_i,a_i) = \emptyset$. Moreover, for any distinct $i_1, i_2 \in \NN_+$, it holds that
    \begin{equation}\label{eq:disjointedges}
        \Big( F\big( a_{i_1}, b_{i_1} \big) \cup F\big( b_{i_1},a_{i_1} \big) \Big) 
        \cap
        \Big( F\big( a_{i_2}, b_{i_2} \big) \cup F\big( b_{i_2},a_{i_2} \big) \Big)
        = 
        \emptyset
        \text,
    \end{equation}
    since the edges of $\pdb[D]_2$ are pairwise disjoint. We now want to show that $\sum_{ f \in F(a_i,b_i) \cup F(b_i,a_i) } p_f \geq \frac1i$ which, together with \eqref{eq:disjointedges}, leads to 
    \[
        \sum_{ f \in \facts(\pdb[I]) } p_f 
            \geq
        \sum_{ i = 1 }^{ \infty } 
            \sum_{ f \in F(a_i,b_i) \cup F(b_i,a_i) }
                p_f
            \geq
        \sum_{ i = 1 }^{ \infty } \frac1i
            = 
        \infty
        \text.
    \]
    This is a contradiction to the assumption that $\pdb[I]$ is a well-defined $\BID$-PDB via \cref{thm:well-def-bid}. 
    
    To prove $\sum_{ f \in F(a_i,b_i) \cup F(b_i,a_i) } p_f \geq \frac1i$, we first observe that
    \begin{align*}
        \frac1{i^2}
            =
        \Pr_{ D \sim \pdb[D] }\big( \, E(a_i,b_i) \in D \text{ and } E(b_i,a_i) \in D \, \big)
            &=
        \Pr_{ I \sim \pdb[I] }\big( \, I \models \Phi( a_i,b_i ) \text{ and } I \models \Phi( b_i,a_i ) \, \big)\\
            &\leq
        \Pr_{ I \sim \pdb[I] }\big( \, I \cap F(a_i,b_i) \neq \emptyset \text{ and } I \cap F(b_i,a_i) \neq \emptyset \, \big)
    \end{align*}
    Suppose that the blocks of $\pdb[I]$ are $B_1,B_2,\dots$. For $k = 1, 2, \dots$, we define events $X_k,Y_k$ as follows:
    \[
        X_k \coloneqq \set[\big]{ I \in \worlds( \pdb[I] ) \with I \cap B_k \cap F(a_i,b_i) \neq \emptyset }
            \quad\text{and}\quad
        Y_k \coloneqq \set[\big]{ I \in \worlds( \pdb[I] ) \with I \cap B_k \cap F(b_i,a_i) \neq \emptyset }
            \text.
    \]
    Because $\pdb[I]$ is $\BID$, all events $X_k$ are independent, and so are all $Y_k$. Moreover, the two events $X_k$ and $Y_k$ are mutually exclusive for all $k$. We let $x_k$ and $y_k$ denote the probabilities of $X_k$ and $Y_k$ in $\pdb[I]$, respectively. As the events $X_k$, respectively $Y_k$, only concern facts in the block $B_k$, we have
    \[
        x_k = \sum_{ f \in F(a_i,b_i) \cap B_k } p_f
            \quad\text{and}\quad
        y_k = \sum_{ f \in F(b_i,a_i) \cap B_k } p_f\text.
    \]
    Now, since 
    \[
        \sum_{ f \in F(a_i,b_i) \cup F(b_i,a_i) } p_f = \sum_{k = 1}^{\infty} x_k + \sum_{k = 1}^{\infty} y_k\text,
    \]
    it remains to show that $\sum_{k = 1}^{\infty} x_k + \sum_{k = 1}^{\infty} y_k \geq \frac1i$. For this, we proceed using the inclusion-exclusion-formula to get
    \begin{align*}
        &\Pr_{ I \sim \pdb[I] } \big( \, I \cap F(a_i,b_i ) \neq \emptyset \text{ and }I \cap F(b_i,a_i) \neq \emptyset \, \big)\\
            &=
        1 - \Pr_{ I \sim \pdb[I] } 
            \big( \, I \cap F(a_i,b_i ) = \emptyset 
                \text{ or }
            I \cap F(b_i,a_i) = \emptyset \, \big)\\
            &=
        1 
            - 
                \Pr_{ I \sim \pdb[I] }  
                    \big( \, I \cap F(a_i,b_i ) = \emptyset \, \big)
            -
                \Pr_{ I \sim \pdb[I] } 
                    \big( \, I \cap F(b_i,a_i) = \emptyset \, \big)
            +
                \Pr_{ I \sim \pdb[I] } 
                    \big( \, I \cap F(a_i,b_i ) 
                        =
                    I \cap F(b_i,a_i) = \emptyset \, \big)\\
        &= 1 
        - \prod_{k=1}^\infty (1-x_k) - \prod_{k=1}^\infty (1-y_k) + \prod_{k=1}^\infty (1-x_k-y_k)\\
        &\leq 
        1 - 
        \prod_{k=1}^\infty (1-x_k) - \prod_{k=1}^\infty (1-y_k) + \prod_{k=1}^\infty (1-x_k-y_k+x_k y_k)\\
        &=
        1 
        - \prod_{k=1}^\infty (1-x_k) 
        - \prod_{k=1}^\infty (1-y_k)
        + \prod_{k=1}^\infty (1-x_k) \prod_{k=1}^\infty (1-y_k)\\
        &=
        \bigg(1- \prod_{k=1}^\infty (1-x_k)\bigg) \cdot \bigg( 1 - \prod_{k=1}^\infty (1-y_k) \bigg)\text.
    \end{align*}
    There are two cases. If one of $\sum_{ k = 1 }^{ \infty } x_k$ or $\sum_{ k = 1 }^{ \infty } y_k$ is larger than or equal to $1$, then we are done. In the other case, when both sums are smaller than $1$, we can apply a Weierstrass inequality (cf. \cite[Chapter VI]{Bromwich1926}) to get
    \[
        \bigg(1- \prod_{ k = 1 }^{ \infty } (1-x_k)\bigg) \cdot \bigg( 1 - \prod_{ k = 1 }^{ \infty } (1-y_k) \bigg)
            \leq 
        \bigg( \sum_{ k = 1 }^{ \infty } x_k \bigg) \cdot \bigg(\sum_{ k = 1 }^{ \infty } y_k \bigg)\text.
    \]
    Together with the inequalities from before, we get
    \begin{equation}\label{eq:pseudoroot}
        \frac{1}{i^2} 
        \leq 
       \bigg(\sum_{ k = 1 }^{ \infty } x_k \bigg) \cdot \bigg(\sum_{ k = 1 }^{ \infty } y_k \bigg)\text.
    \end{equation}
    Therefore, one of $\sum_{ k = 1 }^{ \infty } x_k$ and $\sum_{ k = 1 }^{ \infty } y_k$ is at least $\frac{1}{i}$, so in particular
    \[
        \sum_{ k = 1 }^{ \infty } x_k + \sum_{ k = 1 }^{ \infty } y_k \geq \frac{1}{i}\text.
        \qedhere
    \]
\end{proof}

Let us now reconsider \cref{lem:symgraphcqti,lem:symgraphNOTcqbid} in direct comparison. First, we note that, the view from the proof of \cref{lem:symgraphcqti}, given \cref{lem:symgraphNOTsjfcqbid} \emph{had} to make use of a self-join. Looking deeper into the proof, we can see that this self-join essentially brought about a squaring of the marginal probabilities when applying the view to the underlying $\TI$-PDB. This is why the marginal probabilities in \cref{lem:symgraphcqti} were chosen to be $\frac{1}{i^4}$, since taking the root leaves us with $\frac{1}{i^2}$, leading to a convergent series. Note that instead, the marginal probabilities of the edges in the example from \cref{lem:symgraphNOTcqbid} were chosen to be $\frac{1}{i^2}$, where taking the root leaves us with $\frac{1}{i}$, leading to a divergent series. At the beginning of the proof, it is not clear why we would end up with the consideration of this root. However, our arguments led us to the related situation of \eqref{eq:pseudoroot}. All in all, the salient point is that in \cref{lem:symgraphcqti} the series of marginal probabilities converges fast enough for our argument to go through, whereas the series from \cref{lem:symgraphNOTcqbid} converges too slowly to admit a representation in the desired class.

\medskip

To complete the picture, we still need to determine whether the classes $\UCQ(\BID)$ and $\FO(\BID)$ are separated or equal. This is done in the next proposition.

\begin{proposition}\label{pro:ucqbidNOTfobid}
    $\UCQ(\BID) \subsetneq \FO(\BID)$.
\end{proposition}

\begin{proof}
    We show that there exists a $\FO(\TI)$-PDB $\pdb$ such that $\pdb \notin \UCQ(\BID)$. Consider the database schema with a single unary relation $R$ over a countably infinite domain $\UU$. We let $\DD = \set{ D_1, D_2, \dots }$ be a set of pairwise disjoint fact sets over this schema such that for all $n \in \NN_+$, the set $D_n$ has size $\size{D_n} \geq n$. An easy example is $D_n = \set{R(a_i)\colon i = n^2, n^2 + 1, \dots, (n+1)^2-1}$. We know from \cref{thm:logical-fo}, that there exists a PDB $\pdb \in \FO(\TI) = \FO(\BID)$ with $\worlds(\pdb) = \DD$.
    
    We now show that $\pdb \notin \UCQ(\BID)$. Assume otherwise and let $\pdb[I]$ be a $\BID$-PDB and $\Phi$ be a $\UCQ$-view such that $\pdb = \Phi(\pdb[I])$. Recall that a full block of $\pdb[I]$ is a block in which all marginal probabilities sum to $1$. We call an instance $I$ from $\pdb[I]$ \emph{minimal}, if it only contains facts from full blocks. Note that if $I$ is an instance of positive probability in $\pdb[I]$, then it is a superset of a minimal instance, and that all minimal instances themselves have positive probability. Moreover, all minimal instances have the same cardinality, say $m$, which is equal to the number of full blocks in $\pdb[I]$. As $\Phi$ is (in particular) first-order, $\size[\big]{ \Phi(I) } \leq q\big(\size{I}\big)$ for some fixed polynomial $q$ (we saw that in the proof of \cref{lem:fo-preserves-finite-moments}). Therefore, $\size[\big]{ \Phi(I) } \leq q(m)$ for all minimal instances $I$. Let $n > q(m)$ and pick an arbitrary instance $I$ from the preimage of $D_n$ under $\Phi$. Let $I'$ be the restriction of $I$ to the facts from the full blocks of $\pdb[I]$. Then $I'$ is a minimal instance with $I' \subseteq I$. Since $\Phi$ is monotone, it holds that $\Phi(I') \subseteq \Phi(I) = D_n$. Note that $\Phi(I') \neq \emptyset$, since $I'$ has positive probability, but $\pdb[D]$ does not contain the empty instance. Since $\pdb$ is domain-disjoint, this implies $\Phi(I') = \Phi(I) = D_n$. Then, however, $\size[\big]{ \Phi(I') }\leq q(m) < n \leq \size{D_n}$ contradicts $\Phi(I') = D_n$.
\end{proof}

Combining all the previous results from this subsection, we are now finally ready to state the main result of our investigation of the relative expressive power of views over PDBs.

\begin{figure}[hbt]
    \mbox{}\hfill%
    \tikzset{ucqtifin/.style={fill=lightgray!25!white}}
    \tikzset{fotifin/.style= {fill=lightgray!25!white}}
    \tikzset{bid/.style=     {fill=lightgray!25!white}}
    \tikzset{ti/.style=      {fill=lightgray!25!white}}
    \begin{subfigure}[t]{.45\textwidth}\centering%
        \begin{tikzpicture}
            \newcommand*{\texteq}[1]{\mathrel{\overset{\smash{\scriptscriptstyle{\text{#1}}}}{=}}}
            \tikzset{class/.style={rectangle,rounded corners,font={\small}}}
            \tikzset{line/.style={draw,thick,shorten <= 2pt, shorten >=2pt}}
            \matrix[rectangle,anchor=south] (REPinfin) [column sep={1.5cm,between origins}, row sep=.75cm] at (0,5) {
                \node[missing] {}; &
                \node[class,fotifin] (PDB) {$\PDB$}; &
                \node[missing] {};\\
                \node[missing] {}; &
                \node[class,align=center,fotifin] (FOTI) {$\FOTI = \FOBID = \FOTIFO$}; &
                \node[missing] {};\\ 
                \node[missing] {}; &
                \node[class,align=center,fotifin] (UCQBID) {$\UCQ(\BID)$}; &
                \node[missing] {};\\
                \node[class,ucqtifin] (UCQTI) {$\UCQ(\TI)$}; & 
                \node[missing] {}; & 
                \node[class,fotifin] (CQBID) {$\CQ(\BID)$};\\
                \node[class,ucqtifin] (CQTI) {$\CQ(\TI)$}; & 
                \node[missing] {}; & 
                \node[class,fotifin] (sjfCQBID) {$\sjfCQ(\BID)$};\\
                \node[class,ucqtifin] (sjfCQTI) {$\sjfCQ(\TI)$}; & 
                \node[missing] {}; & 
                \node[class,bid] (BID) {$\BID$};\\
                \node[missing] {}; & 
                \node[class,ti] (TI) {$\TI$}; &
                \node[missing] {};\\
            };
            \begin{scope}[on background layer]
                \draw[line] (TI.center)        to (BID.center);
                \draw[line] (TI.center)        to (sjfCQTI.center);
                \draw[line] (sjfCQTI.center)   to (CQTI.center);
                \draw[line] (CQTI.center)      to (UCQTI.center);
                \draw[line] (BID.center)       to (sjfCQBID.center);
                \draw[line] (sjfCQBID.center)  to (CQBID.center);
                \draw[line] (CQBID.center)     to (UCQBID.center);
                \draw[line] (sjfCQTI.center)   to (sjfCQBID.center);
                \draw[line] (CQTI.center)      to (CQBID.center);
                \draw[line] (UCQTI.center)     to (UCQBID.center);
                \draw[line] (UCQBID.center)    to (FOTI.center);
                \draw[line] (FOTI.center)      to (PDB.center);
            \end{scope}
        \end{tikzpicture}
        \caption{Hasse diagram for the relative expressive power of views over countable PDBs.}
        \label{fig:hasseinfin}
    \end{subfigure}
    \hfill%
    \begin{subfigure}[t]{.45\textwidth}\centering%
        \begin{tikzpicture}
            \tikzset{class/.style={fill=ShadingLight,rectangle,rounded corners,font={\small}}}
            \tikzset{line/.style={draw,lightgray!50!white,line width=.15cm,shorten <= 2pt, shorten >=2pt}}
            \tikzset{enil/.style={draw,thick,dotted,shorten <= 2pt, shorten >=2pt,-stealth'}}
            \matrix[rectangle,anchor=south] (REPinfin) [column sep={1.5cm,between origins}, row sep=.75cm] at (0,5) {
                \node[missing] {}; &
                \node[class,fotifin] (PDB) {$\PDB$}; &
                \node[missing] {};\\
                \node[missing] {}; &
                \node[class,align=center,fotifin] (FOTI) {$\FOTI = \FOBID = \FOTIFO$}; &
                \node[missing] {};\\ 
                \node[missing] {}; &
                \node[class,align=center,fotifin] (UCQBID) {$\UCQ(\BID)$}; &
                \node[missing] {};\\
                \node[class,ucqtifin] (UCQTI) {$\UCQ(\TI)$}; & 
                \node[missing] {}; & 
                \node[class,fotifin] (CQBID) {$\CQ(\BID)$};\\
                \node[class,ucqtifin] (CQTI) {$\CQ(\TI)$}; & 
                \node[missing] {}; & 
                \node[class,fotifin] (sjfCQBID) {$\sjfCQ(\BID)$};\\
                \node[class,ucqtifin] (sjfCQTI) {$\sjfCQ(\TI)$}; & 
                \node[missing] {}; & 
                \node[class,bid] (BID) {$\BID$};\\
                \node[missing] {}; & 
                \node[class,ti] (TI) {$\TI$}; &
                \node[missing] {};\\
            };
            
            \begin{scope}[on background layer]
                \draw[line] (TI.center)        to (BID.center);
                \draw[line] (TI.center)        to (sjfCQTI.center);
                \draw[line] (sjfCQTI.center)   to (CQTI.center);
                \draw[line] (CQTI.center)      to (UCQTI.center);
                \draw[line] (BID.center)       to (sjfCQBID.center);
                \draw[line] (sjfCQBID.center)  to (CQBID.center);
                \draw[line] (CQBID.center)     to (UCQBID.center);
                \draw[line] (sjfCQTI.center)   to (sjfCQBID.center);
                \draw[line] (CQTI.center)      to (CQBID.center);
                \draw[line] (UCQTI.center)     to (UCQBID.center);
                \draw[line] (UCQBID.center)    to (FOTI.center);
                \draw[line] (FOTI.center)      to (PDB.center);
            \end{scope}
            
            \draw[enil] (BID)       to node[circle,below left,pos=.25,font={\footnotesize}] {(1)} (UCQTI);
            \draw[enil] (sjfCQTI)   to node[circle,below,font={\footnotesize}]              {(2)} (BID);
            \draw[enil] (CQTI)      to node[circle,below,pos=.2,font={\footnotesize}]       {(3)} (sjfCQBID);
            \draw[enil] (UCQTI)     to node[circle,below,font={\footnotesize}]              {(4)} (CQBID);
            \draw[enil] (FOTI)      to node[circle,right,font={\footnotesize}]              {(5)} (UCQBID);
            \draw[enil] (PDB)       to node[circle,right,font={\footnotesize}]              {(6)} (FOTI);
        \end{tikzpicture}
        \caption{The central relationships in the proof of \cref{thm:ultimate}. Gray lines indicate trivial inclusions. A dotted arrow from $\classstyle{D}$ to $\classstyle{D}'$ means $\classstyle{D} \not\subseteq \classstyle{D}'$.}
        \label{fig:hasseinfinproof}
    \end{subfigure}
    \hfill\mbox{}%
    \caption{Diagrams for the proof of \cref{thm:ultimate}.}
    \label{fig:hasse}
\end{figure}
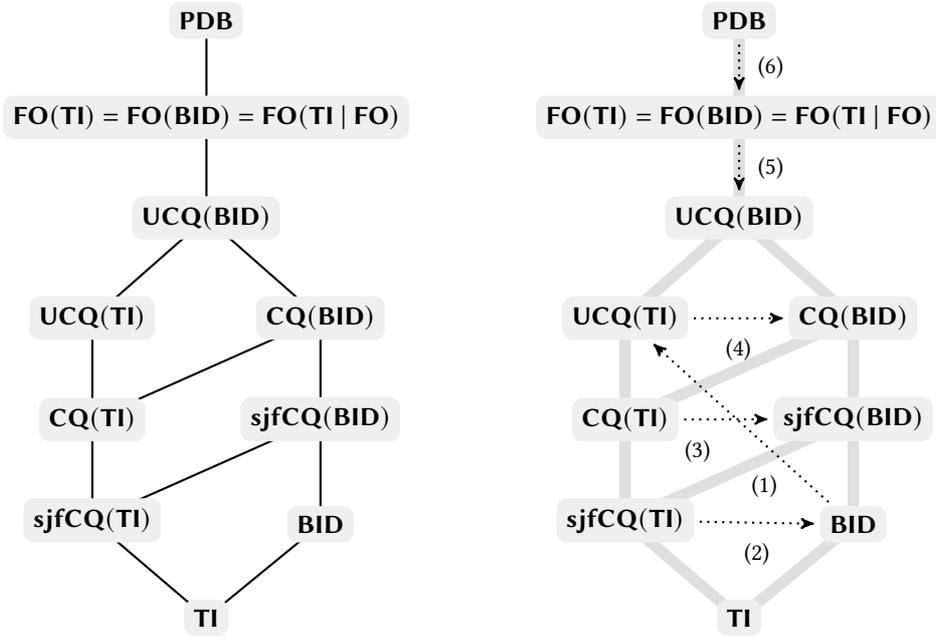

\begin{theorem}\label{thm:ultimate}
    The classes of views of probabilistic databases have exactly the relative expressive power as shown in the Hasse diagram of \cref{fig:hasseinfin}.
\end{theorem}

\begin{proof}
    Consider the diagram shown in \cref{fig:hasseinfinproof}. The thick lines indicate the trivial inclusions that arise from $\TI \subseteq \BID \subseteq \PDB$ and $\set{\mathrm{Id}} \subseteq \sjfCQ \subseteq \CQ \subseteq \UCQ \subseteq \FO$ where $\mathrm{Id}$ denotes the identity view. It remains to show that these are proper inclusions and that no other inclusions hold. For this purpose, consider the dotted arrows. Once they are shown, all remaining relationships are entailed by transitivity as follows: 
    \begin{itemize}
        \item If $\classstyle{D}_1 \not\subseteq \classstyle{D}_2$ and $\classstyle{D}'_2 \subseteq \classstyle{D}_2$, then $\classstyle{D}_1 \not\subseteq \classstyle{D}'_2$.
        \item If $\classstyle{D}_1 \not\subseteq \classstyle{D}_2$ and $\classstyle{D}_1 \subseteq \classstyle{D}'_1$, then $\classstyle{D}'_1 \not\subseteq \classstyle{D}_2$.
    \end{itemize}
    The relationships of the dotted arrows follow from previous results as follows.
    \begin{enumerate}
        \item This follows from \cref{pro:mon+ti+disjoint}.
        \item This is witnessed by \cref{exa:RSexample}.
        \item \Cref{lem:symgraphNOTsjfcqbid} states that no PDB from $\edgeindepgraphs$ are in $\sjfCQ(\BID)$. However, \cref{lem:symgraphcqti} presents an example of a PDB from $\edgeindepgraphs$ that is in $\CQ(\TI)$.
        \item \Cref{lem:symgraphucq} states that all PDBs from $\edgeindepgraphs$ are in $\UCQ(\TI)$, yet \cref{lem:symgraphNOTcqbid} gives an example of a PDB from $\edgeindepgraphs$ that is not in $\CQ(\BID)$.
        \item See \cref{pro:ucqbidNOTfobid}.
        \item See \cref{pro:foti-not-everything} (which is a restatement of \cite[Proposition 4.9]{Grohe+2019}).\qedhere
    \end{enumerate}
\end{proof}

Concluding this section, we point out that apart from \cref{lem:symgraphcqti,lem:symgraphNOTcqbid} the arguments we found for the separations are of a purely logical nature. In \cref{lem:symgraphcqti,lem:symgraphNOTcqbid}, however, the argument depended on our choice of probabilities that allowed us to find a well-defined representation in the first case, and enabled us to rule out the existence of such in the second case. It is not clear whether this separation can also be obtained using an argument that only involves the structure of the possible worlds rather than their probabilities. %

\section{Concluding Remarks}\label{sec:conclusion}
This article initiates research on representations of infinite probabilistic databases. Our focus is on studying representability by first-order views over tuple-independent PDBs. Although it is known that $\FOTI \subsetneq \PDB$, we show that $\FOTI$ is quite robust, as $\FOTI = \FOBID = \FOTIFO$, and it can represent any bounded size PDB. In addition, while the finite moments property is necessary for membership in $\FOTI$, we show that it is not sufficient.
We also study how restricting the expressive power of the views (using fragments of first-order logic) affects the PDBs that can be represented. We show that the different classes obtained using $\sjfCQ$-, $\CQ$- and $\UCQ$-views over $\TI$- and $\BID$-PDBs are distinct (even though many of them collapse for finite PDBs).
We establish the concrete relationships between these classes, as shown in \cref{fig:finalhasse}.
In our analysis, we pay attention to the causes for (non-)rep\-re\-sentabil\-i\-ty.
These can be twofold in general: arithmetical or purely logical. The class $\FOTI$, however, (unlike some of its fragments) eludes purely logical arguments for non-rep\-re\-sentabil\-i\-ty.
When considering logical reasons, the mentioned relationships to just sets of possible worlds (in the context of incomplete databases) extend existing work to infinite PDBs \cite{Benjelloun+2006,Green+2006}.

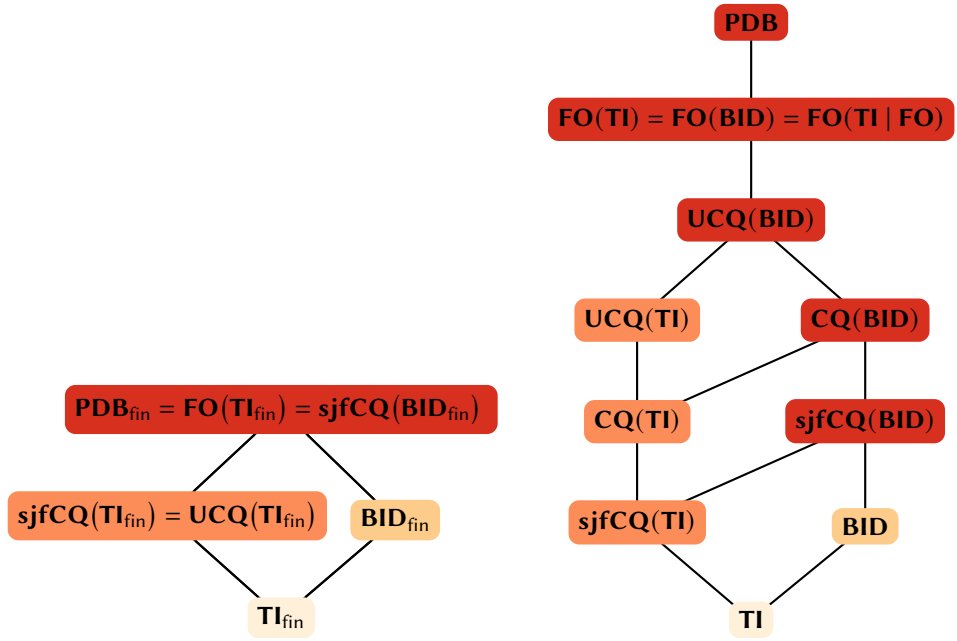
\begin{figure}[htb]
    \mbox{}\hfill%
    \tikzset{ucqtifin/.style={fill=ShadingDark}}
    \tikzset{fotifin/.style={fill=ShadingVeryDark}}
    \tikzset{bid/.style={fill=ShadingMedium}}
    \tikzset{ti/.style={fill=ShadingLight}}
    \begin{subfigure}[b]{.45\textwidth}\centering%
    \begin{tikzpicture}
        \tikzset{class/.style={rectangle,rounded corners,font={\small}}}
        \tikzset{line/.style={draw,thick}}
        \matrix[rectangle,anchor=south] (REPfin) at (0,5) [column sep={1.5cm,between origins}, row sep=.75cm] {
            \node[missing] {}; &
            \node[class,align=center,fotifin] (PDBfin) {%
                $\PDBfin = \FO\big(\TIfin\big) = \sjfCQ\big(\BIDfin\big)$
            }; &
            \node[missing] {};\\
            \node[class,align=center,ucqtifin] (UCQTIfin) {$\sjfCQ\big(\TIfin\big)= \UCQ\big(\TIfin\big)$}; & 
            \node[missing] {}; & 
            \node[class,bid] (BIDfin) {$\BIDfin$};\\
            \node[missing] {}; & 
            \node[class,ti] (TIfin) {$\TIfin$}; &
            \node[missing] {};\\
        };
        \draw[line] (TIfin) to     (BIDfin);
        \draw[line] (TIfin) to    (UCQTIfin);
        \draw[line] (UCQTIfin) to  (PDBfin);
        \draw[line] (BIDfin) to    (PDBfin);
        
        \begin{scope}[on background layer]
            \draw[line] (TIfin.center)        to (BIDfin.center);
            \draw[line] (TIfin.center)        to (UCQTIfin.center);
            \draw[line] (BIDfin.center)       to (PDBfin.center);
            \draw[line] (UCQTIfin.center)     to (PDBfin.center);
        \end{scope}
    \end{tikzpicture}
        \caption{Hasse diagram for the relative expressive power of views over finite PDBs.}
    \end{subfigure}
    \hfill%
    \begin{subfigure}[b]{.45\textwidth}\centering%
        \begin{tikzpicture}
            \newcommand*{\texteq}[1]{\mathrel{\overset{\smash{\scriptscriptstyle{\text{#1}}}}{=}}}
            \tikzset{class/.style={rectangle,rounded corners,font={\small}}}
            \tikzset{line/.style={draw,thick,shorten <= 2pt, shorten >=2pt}}
            \matrix[rectangle,anchor=south] (REPinfin) [column sep={1.5cm,between origins}, row sep=.75cm] at (0,5) {
                \node[missing] {}; &
                \node[class,fotifin] (PDB) {$\PDB$}; &
                \node[missing] {};\\
                \node[missing] {}; &
                \node[class,align=center,fotifin] (FOTI) {$\FOTI = \FOBID = \FOTIFO$}; &
                \node[missing] {};\\ 
                \node[missing] {}; &
                \node[class,align=center,fotifin] (UCQBID) {$\UCQ(\BID)$}; &
                \node[missing] {};\\
                \node[class,ucqtifin] (UCQTI) {$\UCQ(\TI)$}; & 
                \node[missing] {}; & 
                \node[class,fotifin] (CQBID) {$\CQ(\BID)$};\\
                \node[class,ucqtifin] (CQTI) {$\CQ(\TI)$}; & 
                \node[missing] {}; & 
                \node[class,fotifin] (sjfCQBID) {$\sjfCQ(\BID)$};\\
                \node[class,ucqtifin] (sjfCQTI) {$\sjfCQ(\TI)$}; & 
                \node[missing] {}; & 
                \node[class,bid] (BID) {$\BID$};\\
                \node[missing] {}; & 
                \node[class,ti] (TI) {$\TI$}; &
                \node[missing] {};\\
            };
            \begin{scope}[on background layer]
                \draw[line] (TI.center)        to (BID.center);
                \draw[line] (TI.center)        to (sjfCQTI.center);
                \draw[line] (sjfCQTI.center)   to (CQTI.center);
                \draw[line] (CQTI.center)      to (UCQTI.center);
                \draw[line] (BID.center)       to (sjfCQBID.center);
                \draw[line] (sjfCQBID.center)  to (CQBID.center);
                \draw[line] (CQBID.center)     to (UCQBID.center);
                \draw[line] (sjfCQTI.center)   to (sjfCQBID.center);
                \draw[line] (CQTI.center)      to (CQBID.center);
                \draw[line] (UCQTI.center)     to (UCQBID.center);
                \draw[line] (UCQBID.center)    to (FOTI.center);
                \draw[line] (FOTI.center)      to (PDB.center);
            \end{scope}
            
        \end{tikzpicture}
        \caption{Hasse diagram for the relative expressive power of views over countable PDBs.}
    \end{subfigure}
    \hfill\mbox{}%
    \caption{Hasse diagrams of PDB classes with independence assumptions in the finite and the countable setting. Edges and their upwards closure mean proper inclusion. All remaining pairs of classes are incomparable. Areas of the same shade indicate which classes from the countable setting collapse in the finite setting.}
    \label{fig:finalhasse}
\end{figure}

Our characterization of $\FOTI$ is only partial, as there is a gap between our necessary and sufficient conditions for membership. Thus, a natural direction for future work is obtaining a complete characterization of this class or the other classes of PDBs that appeared in this paper.
This is especially challenging since, even in the finite setting, we know of no algorithm to decide whether a given finite PDB is a $\UCQ(\TIfin)$-PDB.
Regarding the causes for (non-)rep\-re\-sentabil\-i\-ty, most of our separations have been purely logical. It remains open whether such a separation of $\CQ(\TI)$ and $\UCQ(\TI)$ exists.
Further possible extensions of our work include conditioning with respect to fragments of $\FO$ or using first-order views with inbuilt relations (e.g., over ordered universes). Further research could also investigate which of our results can be extended to uncountable PDBs. Other interesting insights might be obtained by relating our work to recent notions of measuring uncertainty in incomplete databases~\cite{Libkin2018,Console+2020}. Regarding query evaluation in infinite PDBs, beyond preliminary results on TI-PDBs~\cite{Grohe+2019}, an abstract investigation remains open.

\begin{acks}
    This work was supported by the German Research Foundation (DFG) under grants GR 1492/16-1 and KI 2348/1-1 \enquote{Quantitative Reasoning About Database Queries} and GRK 2236 (UnRAVeL). Nofar Carmeli was supported by the Google PhD Fellowship.
\end{acks}

\bibliographystyle{plainurl}

\end{document}